\documentclass[accepted]{uai2023} 

\usepackage[american]{babel}

\usepackage{natbib} 
\usepackage{mathtools} 
\usepackage{booktabs} 
\usepackage{tikz} 
\usepackage{amsmath}
\usepackage{amssymb}
\usepackage{amsthm}
\usepackage{subcaption}

\newtheorem{theorem}{Theorem}
\newtheorem{lemma}[theorem]{Lemma}
\newtheorem{corollary}{Corollary}[theorem]

\title{Keep-Alive Caching for the Hawkes process}

\author[]{Sushirdeep Narayana}
\author[]{Ian A. Kash}
\author[]{\\ snaray25@uic.edu}
\author[]{iankash@uic.edu}
\affil[]{%
    Department of Computer Science \\
    University of Illinois at Chicago \\
    Chicago, Illinois, USA
}

\begin{document}
\maketitle

\begin{abstract}
We study the design of caching policies in applications such as serverless computing where there is not a fixed size cache to be filled, but rather there is a cost associated with the time an item stays in the cache.  We present a model for such caching policies which captures the trade-off between this cost and the cost of cache misses. We characterize optimal caching policies in general and apply this characterization by deriving a closed form for Hawkes processes. Since optimal policies for Hawkes processes depend on the history of arrivals, we also develop history-independent policies which achieve near-optimal average performance. We evaluate the performances of the optimal policy and approximate polices using simulations and a data trace of Azure Functions, Microsoft's FaaS (Function as a Service) platform for serverless computing.
\end{abstract}

\section{Introduction}\label{sec:intro}

Datacenters provide a variety of caching services in the interest of reducing latency.  While traditional caches have a fixed size determined by the underlying hardware, datacenter resources are fungible and could be used for a variety of purposes.  For example, in serverless computing (also known as Function-as-a-Service) the cloud provider handles all the provisioning and configuration of resources for each user, while the user simply pays for the amount of time of the function execution. The cloud provider could, in principle, keep a container in memory for every user to ensure that function calls can be executed immediately. However, this would be excessively expensive in terms of resources allocated but unused. Instead, user application code which has not been used recently may be kept in persistent storage~\citep{shahrad2020serverless}. Another example involves Time-to-Live (TTL) caches where the cache controller has a timer-based parameter for when to evict objects from the cache. Work on TTL-based caches for Content Delivery Networks (CDNs) has observed that hit rate guarantees can be provided by controlling the cache space used by a customer and suggested pricing based on this space consumption~\citep{basu2018adaptive}.

In contrast to traditional caches, in keep-alive caching the decision is not about \em{what} object to evict from the cache when space is needed.  Rather the relevant question is \em{when} is it worth keeping the object in the cache.  Such a decision trades off the opportunity cost of not putting those resources to other uses (which may be caching other objects or some entirely different purpose) against the cost of a cache miss. While prior work has examined the trade-off between cache misses and the overall cache size~\citep{basu2018adaptive}, we focus on precise answers regarding optimal and approximately optimal caching decisions for a single object.  

We model this problem as trading off between the expected {\em time} an object is kept in the cache before it is next accessed and the probability of a {\em cache miss}.  Since a given object has a fixed size, this lets us precisely quantify the two relevant costs.  We characterize the optimal cache policy and show how this characterization yields simple policies for objects whose arrivals have a monotone hazard rate.  This includes Poisson and Hawkes processes~\citep{laub2021elements}.  

One downside of the optimal policy for Hawkes processes is that it depends on the history of arrivals for the object. So, it needs to be recalculated after every arrival.  A natural alternative is a simple TTL-style policy which keeps the object in the cache for a fixed amount of time after each arrival before evicting it. Such policies have been used in practice in serverless systems such as AWS Lambda and Azure Functions~\citep{shahrad2020serverless}.  The classic ski rental problem shows that, if the TTL is optimized solely based on costs, the resulting policy is a worst-case 2-approximation. We derive an approach to optimize the TTL based on the parameters of the Hawkes process (but still independent of the history of arrivals).

We simulate the average performance of all three policies (optimal, optimized-TTL, and fixed-TTL) on arrival requests that follow a Hawkes process. The simulations illustrate how an optimized-TTL in general and our specific optimization procedure in particular yield near-optimal performance.

We also evaluate these policies on Azure traces released by \citet{shahrad2020serverless}.  Here, Hawkes processes naturally capture applications such as web servers where a recent function invocation makes it more likely for additional invocations to occur (perhaps because the same user makes another request).  We show that applying the optimal policy to the 25\% of applications best fit by a  Hawkes process yields meaningful improvements over the fixed-TTL policy at the scale of a datacenter.  Furthermore, our optimized-TTL approach again yields near-optimal performance.  

\section{Related Work} \label{sec:related-work}

Our modeling decisions draw motivation from serverless computing. Cloud providers that use Function as a Service (FaaS) serverless models, such as AWS Lambda, Google Cloud Functions, IBM Cloud Functions, and Azure Functions, handle all the system administration and resource allocations for customers. As \citet{jonas2019cloud} highlight, one of the advantages of serverless computing is that the users have to pay solely for the usage of resources of their applications. The customers do not pay for maintaining and starting up the resources when their application is not running. On the other hand, they point out the challenges for customers in terms of the unpredictable latency of cold starts \footnote{Cold start refers to a cache miss in the context of serverless computing}. \citet{wang2018peeking} measure the cold start latency of popular FaaS serverless platforms. They discuss the historical usage of fixed caching policies by these platforms and quantify the behavior assumed by our model: cloud providers regularly shutting down instances providing FaaS services to reallocate those resources to other uses. \citet{lin2020serverless} discuss the desirability of approaches, such as ours, that could allow customers to agree to pay a higher price in exchange for personalized warm-start performance guarantees.  

Closest to our work, \citet{shahrad2020serverless} propose a keep-alive cache policy based on an approach to predict future arrivals. They show a significant reduction in memory use when compared with the standard fixed keep-alive policies. In contrast, we take the predictions as given and optimize decisions based on them. Other works that have considered approaches to mitigating cold starts include having multiple tiers of hardware~\citep{roy2022icebreaker}, pre-warming just the networking components~\citep{mohan2019agile}, caching common Python libraries~\citep{oakes2018sock}, and overbooking~\citep{kesidis2019overbooking}. \citet{fuerst2021faascache} build a system which allows the size of a FaaS cache to be scaled dynamically based on arrival rates, but they do not provide any theoretical analysis. ~\citep{romero2021faa,mvondo2021ofc} have looked at caching of other aspects pertaining to function execution such as data access.

Most work on cache algorithms works to optimize their decisions about which items to evict when the cache is full.  Closer to our work is work on TTL caches which evict objects after a fixed amount of time.  These are studied both for applications in settings like Content Delivery Networks (CDNs) as well as their use as a more tractable way of approximately analyzing the performance of traditional caches introduced by \citet{che2002hierarchical}. \citet{berger2014exact} analyze TTL cache networks. They consider the inter-request arrival times of objects and TTL values to be two independent renewal processes. They  study three types of TTL cache policies based on the TTL resets and eviction times (we show what they term $\mathcal{R}$ policies are optimal in our setting). \citet{basu2018adaptive} design two TTL based caching algorithms for CDNs like Akamai, but focus on working within a fixed target cache size. \citet{ferragut2016optimizing} study optimal TTL cache policies. Like us they examine the consequences of the monotonicity of the Hazard rate of the inter-arrival distributions. However, unlike us they focus on trading off hit rate and overall cache size.  They formulate the TTL caching problem as a non-linear optimization problem with non-linear constraint, and prove that the convexity of the optimization problem is related to monotonicity  of the Hazard rate. They provide explicit solutions when the inter-arrival distribution follows the Zipf's law. They evaluate the performance of their optimal policies by comparing their policy hit rate against the LRU (Least Recently Used) policy hit rate. \citet{ali2011survey, alamash2004overview} give a broader introduction to and survey of web page caches.

Closer to our work, \citet{dehghan2019utility} share our notion that users may have utilities which depend on the hit rate and they include the possibility of increasing the cache size at a cost. They assume the utility functions for each file to be concave and the arrival requests to follow a Poisson point process, while our model is more general. \citet{babaie2019cache} extend this to a cache hierarchy network. \citet{panigrahy2017and} allow heterogeneity of user preferences for hit rates, but focus on network of capacity-constrained caches with requests modeled as simple Poisson arrivals.

There is substantial work in the AI literature on other optimization problems that arise in the context of cloud computing such as pricing~\cite{blocq2014shared,friedman2015dynamic,babaioff2017era,kash2019simple,dierks2021competitive}, reservation scheduling~\cite{azar2015truthful,wang2015selling}, information elicitation~\cite{ceppi2015personalized,dierks2019cluster}, and fair division of resources~\cite{parkes2015beyond,kash2014no,friedman2014strategyproof,narayana2021fair}.

\section{Model} \label{sec:model}

We consider a cache system where there is a cost for an object to stay in the cache. There is also a cost for a cache miss. For concreteness we describe our model using terms from one natural application (serverless computing), but it is also relevant for other applications such as CDNs~\cite{ferragut2016optimizing}. First, we describe a cache policy in this setting. Next, we detail the parameters of the cache associated with the cloud provider. Then, we express the cost of a cache policy. Since the time of arrival of future requests is unknown, we assume that the cloud provider has access to the distribution of the arrival requests, as they could be estimated by the past arrival requests. 

\textbf{Cache Policy}
We assume the cache has infinite capacity because the provider can always dedicate more resources to its serverless offering, which differentiates our model from those driven by capacity. Thus, for us a cache policy is not about which object should be evicted when space is needed but rather how long (or more generally when) we should keep it in the cache.

Let $\mathcal{H}_{m-1} = \{t_{1}, \; t_{2}, \; \cdots, \; t_{m-1}\}$ denote the history of $m-1$ previous requests for the application.  Here, $t_{1}, \; t_{2}, \cdots \text{and}, \; t_{m-1}$ denote the time of 1st, 2nd , $\cdots$ and, (m-1)-st request respectively. Let $x_{m} = t_{m} - t_{m-1}$ denote the \textit{m}-th inter-arrival time. Our analysis of policies is based on these inter-arrival times.

A cache policy $\pi$ of an application is a sequence of time windows during which a possible requested application is moved in and out of the cache. A keep-alive window is the time interval during which the application is kept in the cache. The policy is reset after each arrival request for an application. More formally, a policy $\pi(x|\mathcal{H}_{m-1})$ can be expressed as an indicator function of a sequence of keep-alive windows for the $m$-th inter-arrival as: 
\begin{equation*}
\pi(x|\mathcal{H}_{m-1}) =  \begin{cases}1 \; , \quad x \in [L_{0} , \; L_{1}] \; \bigcup \cdots \; [L_{2k-2}, L_{2k-1}]\\ 0 \; , \quad \text{otherwise} \end{cases}  
\end{equation*}

\noindent Here, $L_{2i}$ denotes amount of time after $t_{m-1}$ (the time of the most recent request) where the $i$-th keep-alive window starts, while $L_{2i+1}$ denotes the point where the $i$-th keep-alive window end for some $k$ and all \; $i \in \{0,\; 1,\;2, \cdots,k-1\}$.  

A policy $\pi(\cdot)$ consisting of a single keep-alive window across the $m$-th inter-arrival can be represented by the length of pre-warming window and the length of keep-alive window. 
That is, if
\begin{equation*}
\pi(x) =  \begin{cases}1 \; , \quad \text{for} \quad x \; \in [\tau_{\text{pw}}, \; \tau_{\text{pw}} + \tau_{\text{ka}}] \\ 0 \; , \quad \text{otherwise} \end{cases}   
\end{equation*}

\noindent then, the policy can be summarized by the parameters $\tau_{\text{pw}}$, and $\tau_{\text{ka}}$. The parameter $\tau_{\text{pw}}$ refers to the length of pre-warming window which is the time interval during which the policy waits before bringing in the possibly requested application into the cache. The parameter $\tau_{\text{ka}}$ denotes the length of keep-alive window which is the time interval when the application is kept in the cache. The pre-warming window is especially useful when the requests for an application occur at regular intervals, such as requests governed by a timer function.   

\textbf{Cost of a Cache Policy}
In the context of serverless computing, an application encounters a \textit{warm} start (cache hit) if its invocation (arrival request) occurs when the keep-alive window is active. An application has a \textit{cold} start (cache miss) when the keep-alive window is not active during its invocation. Thus, we can describe the costs of a policy using:
\begin{itemize}
    \item $c_{cs}$ denotes the cost associated with a cold start.  This is primarily the latency cost experienced by the user when there is a cold start.  It may also include the cost for the cloud provider  to load the application image (requested object) into the cache.
    \item $c_{p}$ denotes the cost per unit time for the cloud provider for keeping the application image (requested object) active in the cache.
\end{itemize}
    
\noindent  The cost of a policy with a single keep-alive window 
$(\tau_{\text{pw}}, \tau_{\text{ka}})$ 
when the inter-arrival time is $x_m$ is given by 
\begin{align*} 
cost(x_{m}, (\tau_{\text{pw}}, \tau_{\text{ka}})) = \; \begin{cases}
&c_{cs}, \mbox{ if }  x_m < \tau_{\text{pw}} \\
&c_{p} \cdot( x_m - \tau_{\text{pw}} ), \mbox{ if }   \tau_{\text{pw}} \; \leq x_m \; \leq \; \tau_{\text{pw}} + \tau_{\text{ka}}  \\
&c_{p}  \tau_{\text{ka}} + c_{cs}, \mbox{ if }  x_m > \; \tau_{\text{pw}} + \tau_{\text{ka}} 
\end{cases}
\end{align*}
The three cases for the cost of a policy correspond to that of a cold start if the invocation occurs during pre-warming, a warm start if the invocation occurs when the keep-alive window is active, and a cold start if the invocation is after the keep-alive window being active, respectively. In the more general case of multiple keep-alive windows the cost $c_p$ must be paid for all prior windows as well (see Equation 1 in the proof of Lemma \ref{expect_cost} for the full formula). This formula implicitly assumes that the time and costs associated with a decision to load the application into the cache (e.g. in the case of multiple windows where it may be moved in and out repeatedly) are zero, and so is in that sense a lower bound on the ``true'' cost. However, we show that for Hawkes processes in particular a single window which starts immediately ($\tau_{\text{pw}} = 0$) is optimal. Such a policy never has to pay such cost except of course during cold starts, but those are already accounted for by $c_{cs}$.  

To compute the cost of a cache policy, $x_{m}\;$ must be known. However, the cloud provider  does not have access to such information. Instead, they can estimate the distribution of $x_m$ from past arrival requests $\mathcal{H}_{m-1}$. We model the application invocations as a point process. Let $f(x| \mathcal{H}_{m-1})$ and $F(x| \mathcal{H}_{m-1})$ denote the conditional probability distribution and the conditional cumulative distribution of an invocation at $x$ units after the most recent invocation given the history $\mathcal{H}_{m-1}$ of previous invocations, respectively. We assume that both are continuous. Let the hazard rate of the based on the inter-arrival be expressed as $\lambda(x|\mathcal{H}_{m-1}) = \displaystyle \frac{f(x|\mathcal{H}_{m-1})}{1 - F(x|\mathcal{H}_{m-1})}$. We use these probability distributions to derive an expression for the expected cost of a cache policy.

\section{Characterization of Optimal Policies} \label{sec:opt-policies}

We start this section by deriving the expected cost of a caching policy over an inter-arrival. Then in Theorem \ref{opt_policy}, we characterize the optimal policy for application invocations in terms of the behavior of the Hazard rate. We apply this to derive optimal policies when the arrival of application invocations follow a Poisson process or a Hawkes process, which are specific cases of arrival requests that have a constant Hazard rate and a monotone decreasing Hazard rate respectively.     

Proofs omitted from this and subsequent sections of the paper can be found in the Appendix.

\begin{lemma} \label{expect_cost}
The expected cost of a cache policy over an inter-arrival is

\begin{equation*}
\mathbb{E}[cost(\pi(\cdot |\mathcal{H}_{m-1}))] = \; c_{cs} + \displaystyle \int_{0}^{\infty} \pi(x|\mathcal{H}_{m-1}) \cdot g(x|\mathcal{H}_{m-1})\; dx ,    
\end{equation*} 
where the instantaneous cost at \textit{x} units after the most recent arrival at $t_{m-1}$ is
\begin{equation*}
g(x|\mathcal{H}_{m-1}) = c_{p} \cdot \big(1 - F(x|\mathcal{H}_{m-1})\big) - c_{cs} \cdot f(x| \mathcal{H}_{m-1}).
\end{equation*} 
\end{lemma}

Using the characterization of the cost of a policy from Lemma~\ref{expect_cost}, Theorem \ref{opt_policy} observes that the sign of $g$ (which determines the optimal policy) is entirely determined by the hazard rate and costs.

\begin{theorem} \label{opt_policy}
The points $L_i$ of the sequence of keep-alive windows over an inter-arrival for the optimal policy $\pi_{\text{opt}}(\cdot | \mathcal{H}_{m-1})$ are at 0, $\infty$, or solutions to the equation $c_{p} - (c_{cs} \cdot \lambda(x | \mathcal{H}_{m-1})) = 0 \;$ where the sign changes.
\end{theorem}

We now examine several special cases of Theorem~\ref{opt_policy} with particularly natural structure.  Our first, Corollary \ref{optimal_decreasing}, describes the optimal policy when the hazard rate of the arrival requests are (weakly) decreasing. This case includes Poisson and Hawkes processes and is the main case we evaluate in our simulations.

\begin{corollary} \label{optimal_decreasing}
If the hazard rate is weakly decreasing, the optimal policy $\pi_{\text{opt}}(x|\mathcal{H}_{m-1})$  is a single keep-alive window starting at $\tau_{pw}=0$, and is given by \\ $\pi_{\text{opt}}(\cdot |\mathcal{H}_{m-1}) =  \begin{cases}1 \; , \quad \forall x \in [0 , \; \tau_{\text{opt}, \mathcal{H}_{m-1}}]\\ 0 \; , \quad \text{otherwise}\end{cases},$ where
\begin{enumerate}
    \item $\tau_{\text{opt}, \mathcal{H}_{m-1}} = \infty$, i.e., to have the keep-alive window always be active when  $\forall x, \; \displaystyle \frac{c_{p}}{c_{cs}} < \; \lambda(x | \mathcal{H}_{m-1}), \; $
    \item $\tau_{\text{opt}, \mathcal{H}_{m-1}} = 0$, i.e.,  not cache and always have a cold start when $\displaystyle \frac{c_{p}}{c_{cs}} >\; \lambda(x = 0 | \mathcal{H}_{m-1})\;$
    \item a keep-alive window of length $\tau_{\text{opt}, \mathcal{H}_{m-1}}$ given by the solution to the equation \\ $ \displaystyle  \frac{c_{p}}{c_{cs}} = \frac{f(x = \tau_{\text{opt}, \mathcal{H}_{m-1}} | \mathcal{H}_{m-1})}{1 - F(x = \tau_{\text{opt}, \mathcal{H}_{m-1}} | \mathcal{H}_{m-1})}$, otherwise. 
\end{enumerate}
\end{corollary}

Next, Corollary \ref{optimal_increasing}, states the optimal policy when the hazard rate is instead (weakly) increasing. 

\begin{corollary} \label{optimal_increasing}
If the hazard rate is weakly increasing, the optimal policy $\pi_{\text{opt}}(\cdot|\mathcal{H}_{m-1})$  is a single keep-alive window with $\tau_{ka} = \infty$ and a pre-warming window, and is given by \\ $\pi_{\text{opt}}(x|\mathcal{H}_{m-1}) =  \begin{cases}1 , \; \tau_{\text{pw}, \mathcal{H}_{m-1}} \leq x \\ 0, \; \text{otherwise}\end{cases},$ where
\begin{enumerate}
    \item $\tau_{\text{pw}, \mathcal{H}_{m-1}} = 0$, i.e., to have the keep-alive window always be active when $\forall x, \quad \displaystyle \frac{c_{p}}{c_{cs}} < \; \lambda(x | \mathcal{H}_{m-1}), \; $
    \item $\tau_{\text{pw}, \mathcal{H}_{m-1}} = \infty$, i.e., to always have a cold start when $\forall x, \; \displaystyle \frac{c_{p}}{c_{cs}} > \; \lambda(x | \mathcal{H}_{m-1}) \;$. 
    \item $\tau_{\text{pw}, \mathcal{H}_{m-1}}$ satisfies the equation\\ $\; \displaystyle \frac{c_{p}}{c_{cs}} = \displaystyle \frac{f(x = \tau_{\text{pw}, \mathcal{H}_{m-1}} |  \mathcal{H}_{m-1})}{1 - F(x = \tau_{\text{pw}, \mathcal{H}_{m-1}} |  \mathcal{H}_{m-1})}$, i.e., an infinite keep-alive window after a pre-warming window of length $\tau_{\text{pw}, \mathcal{H}_{m-1}}$ when $c_{p} - c_{cs} \lambda(x = 0 | \mathcal{H}_{m-1}) \; > 0$ and changes sign. 
\end{enumerate}
\end{corollary}

More generally, we can combine these to understand 
the optimal policy when the hazard rate has a single peak: it is first increasing and then decreasing.  In this case using both $\tau_{pw}$ and $\tau_{ka}$ is optimal (apart from degenerate cases).  This is the form used by \citet{shahrad2020serverless}, and so our results characterize the class of applications for which their approach could be optimal if properly tuned as well as determining how to optimally tune the parameters.\footnote{They use a simple rule of pre-warming at the $5^{th}$ percentile and ending the keep-alive at the $99^{th}$ percentile while our approach would optimize those based on the costs and distribution characteristics.} One notable example of this class of application is applications triggered by a timer.  We would expect $\lambda$ to be 0 until the timer is due to elapse, rapidly increase as we approach our estimate of when the timer will trigger, and then eventually decrease if we appear to be wrong in our estimate of when the timer will next trigger.

We can also apply Theorem \ref{opt_policy} to applications where the hazard rate of the arrival requests for application invocation has a single valley. This might be the case if the initial invocation is likely to trigger several others (in the same spirit as a Hawkes process) but then once the application finishes there is a gap before it is invoked again, perhaps due to a timer.  As Figure~\ref{opt_single_peak} shows the optimal keep-alive policy has at-most two keep-alive windows, whose lengths can be computed using Theorem \ref{opt_policy}.  Of course, we can apply this approach to more complex situations resulting in even more windows if the distributional information available supports that (e.g. for a timer that only causes an arrival under certain additional conditions).
\begin{figure}[t!]
    \centering
    \includegraphics[height=1.5in,trim={6cm 3cm 7cm 0cm},clip]{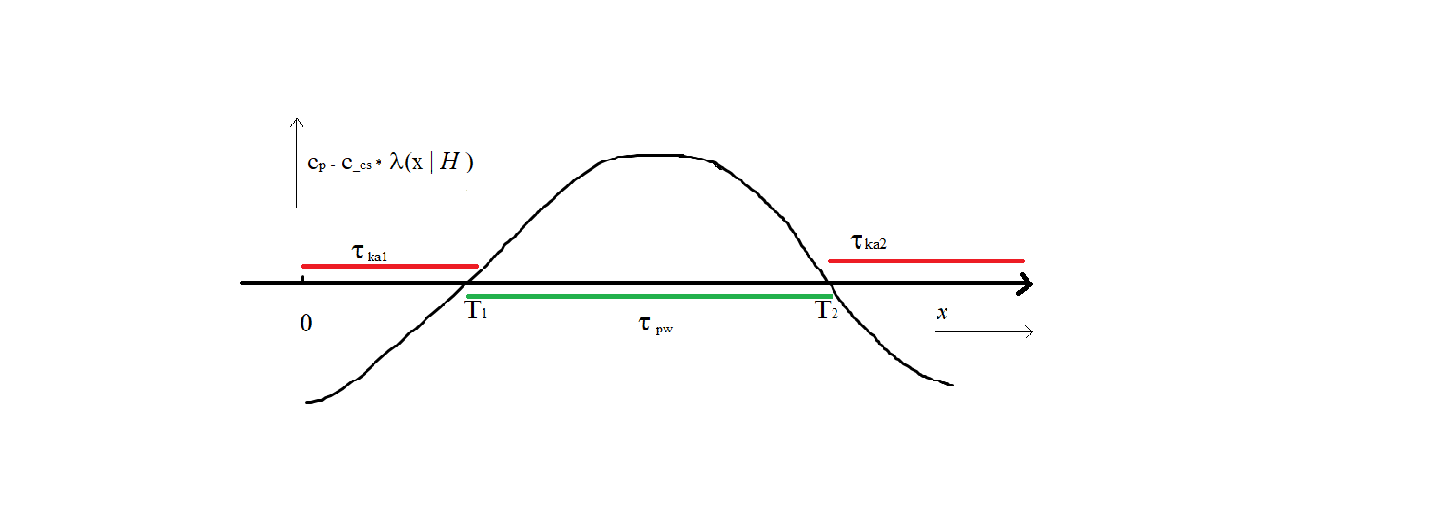}
        \caption{Optimal policy with a single valley hazard rate}
    \label{opt_single_peak}
\end{figure}

\subsection{Poisson Process} \label{sec:poisson}

Our first concrete application of Theorem~\ref{opt_policy}  is when the distribution of the arrival requests follows a Poisson process. For the Poisson process, the relevant quantities are:
\begin{itemize}
    \item $f(t) = \lambda \cdot e^{- \lambda \cdot t}$ for $t \geq 0$, where $\lambda > 0$ is a constant parameter.
    \item $F(t) = \displaystyle  1 - e^{-\lambda \cdot t}$.
    \item $ \displaystyle \lambda(t) =\frac{f(t)}{1 - F(t)} = \lambda$.
\end{itemize}

From Corollary~\ref{optimal_decreasing} we see that, barring degeneracy, the optimal policy for a Poisson process is one of two possibilities since the Hazard rate is constant. The optimal policy  when $\displaystyle \frac{c_{p}}{c_{cs}} \leq \; \lambda$ is to have a keep-alive window active at all times, while when $\displaystyle \frac{c_{p}}{c_{cs}} > \; \lambda$ the optimal policy is to always experience a cold start. The cost of the optimal policy when $\displaystyle \frac{c_{p}}{c_{cs}} > \; \lambda $ is $c_{cs}$. Since the expected inter-arrival time is $\displaystyle \frac{1}{\lambda}$, the expected cost of the optimal policy when $c_{p} - (c_{cs} \cdot \lambda) \leq 0$  is $\displaystyle \frac{c_p}{\lambda}$.

\subsection{Hawkes Process} \label{sec:hawkes_optimal}

In this subsection, we explore the more interesting case of optimal policies when the distribution of the application invocations follow a Hawkes process. The Hawkes process is a point process whose hazard rate, or conditional intensity, is given by 
\begin{equation*}
    \lambda(t| \mathcal{H}) = \frac{f(t| \mathcal{H})}{1 - F(t| \mathcal{H})} = \lambda_{0} + \sum_{t_{j} \in \mathcal{H}} \mu(t - t_{j})
\end{equation*}
for all $t > t_{i-1}$, where $\mathcal{H} = \{t_{1}, t_{2}, \cdots, t_{i-1} \}$ denotes the history of invocations of the function, $\lambda_{0}$ refers to the background intensity, and $\mu$ is the excitation function. We limit our analysis to the exponential excitation function because we need to pick a concrete instantiation to solve for a closed form for optimal window length. The exponential excitation function is expressed as, $\mu(t) = \alpha e^{ - \beta \cdot t}$. The constant $\alpha > 0$ captures the increase in the intensity from an arrival, while the constant $\beta > 0$ is the decay rate of the arrival's influence. The conditional intensity of the self-exciting Hawkes process with an exponential excitation function is thus,
\begin{equation*}
    \lambda(t| \mathcal{H}) = \lambda_{0} + \sum_{t_{j} \in \mathcal{H}} \alpha \cdot e^{- \beta \cdot (t - t_{j})}  \quad, \forall \; t > t_{i-1} .
\end{equation*}

Corollary \ref{optimal_decreasing} characterizes the optimal policy when the distribution of arrival requests follows the Hawkes process to be one of the following policies,
\begin{itemize}
    \item The keep-alive window is to always be active with $\tau_{\text{opt}, \mathcal{H}_{m-1}} = \infty$ when, as in the Poisson case, the background intensity is sufficiently high: $\displaystyle \frac{c_{p}}{c_{cs}} < \; \lambda_{0}$.
    \item Experience a cold start with $\tau_{\text{opt}, \mathcal{H}_{m-1}} = 0$ when $\displaystyle \frac{c_{p}}{c_{cs}} > \; \lambda(x| \mathcal{H}_{m-1}))$, after the most recent request.
    \item The keep-alive window is given by the expression
    \begin{align*}
    \tau_{\text{opt}, \mathcal{H}_{m-1}} = \frac{1}{\beta} \Big( \log \alpha + \log \Big(\sum_{j=1}^{m-1} e^{\beta  (t_{j}-t_{m-1})} \Big) - \log \Big( \frac{c_{p}}{c_{cs}} - \lambda_{0}\Big) \Big)
    \end{align*}
    otherwise. This expression is obtained by substituting the conditional intensity of the Hawkes process in Corollary \ref{optimal_decreasing} and solving for $\tau_{\text{opt}, \mathcal{H}_{m-1}}$ as detailed in the Appendix.
\end{itemize}

\subsection{Optimized-TTL Keep-Alive Windows for Hawkes Processes} \label{subsec:avg_optimal}

Corollary~\ref{optimal_decreasing} provides the optimal history-dependent policy for Hawkes processes. We conclude this section by showing how it also provides motivation for a history-independent heuristic for Hawkes processes. In Section~\ref{sec:simulations} we show this heuristic has strong empirical performance.

Our proposal is to empirically determine a keep-alive window length that works well for ``typical'' windows.  Intuitively, this can be done by sampling a number of histories, computing the optimal policy for each history, and computing a summary statistic.  We find that the average of the optimal policies works well.  We discuss the simulation of a Hawkes process from its parameters in Section~\ref{sec:simulations}

\begin{corollary}
When the parameters of the Hawkes process are such that $c_{p} - (c_{cs} \cdot \lambda(x | \mathcal{H})) = 0$ has a solution, the optimal policy has a history independent lower bound,  and an upper bound expressed as follows
\begin{align*}
\tau_{\text{opt}, \mathcal{H}} & \geq \frac{1}{\beta} \cdot \left(\log \alpha - \log \Big(\frac{c_{p}}{c_{cs}} - \lambda_{0}\Big)\right) \\
\tau_{\text{opt}, \mathcal{H}} & \leq \frac{1}{\beta} \cdot \left(\log \alpha + \log \delta + 1 - \log \Big(\frac{c_{p}}{c_{cs}} - \lambda_{0}\Big)\right)
\end{align*}
\noindent where $\delta$ satisfies
\begin{equation*}
\sum_{i=m-\delta}^{m-1} e^{\beta \cdot (t_i - t_{m-1})} 
\geq \frac{1}{2} \sum_{i=1}^{m-1} e^{\beta \cdot (t_i - t_{m-1})} 
\end{equation*}
\end{corollary}

That is, the most recent $\delta$ arrivals provide at least half the total weight of the history dependent term.  This can be thought of as only having $\delta$ arrivals that are recent enough to matter.  Apart from rare scenarios where $\delta$ is much larger or smaller than typical given the Hawkes process parameters, this bound is relatively insensitive to the exact history due to the log.  In our simulations, particularly Section \ref{subsec:azure_experiments}, and Section \ref{azure_results}, rather than estimate $\delta$ we directly estimate a threshold by simulating sample points for the Hawkes process with the estimated parameters.

We also considered other approaches to deriving history-independent policies.  The solution to the classic ski rental problem shows that setting $\tau_{ka} = c_{cs} / c_p$ always has a cost within a factor of 2 of that of the optimal history-dependent policy for Hawkes processes.  This bound can be tightened by setting $\tau_{ka}$ in a way that depends on the parameters of the Hawkes process.  However, we found that in practice this approach was overly conservative, so we defer its analysis to the appendix.

\section{Simulations} \label{sec:simulations}

\subsection{Performances on Simulated Hawkes Processes} \label{sec:hawkes_sim}

Our theoretical results derived the optimal policy and argued that optimized-TTL, which uses averaging, provided a good heuristic which is independent of history.  To better understand its expected performance, we evaluate both policies on simulated application invocations governed by a Hawkes process. We use \textit{Ogata's modified thinning algorithm} \cite{ogata1981lewis} to generate the samples of the Hawkes process.  We generate 600 sample points of function invocations in a single realization of the Hawkes process. We evaluate the policies by taking the mean over 100 realizations. 

\begin{figure*}[t!]
    \centering
    \begin{subfigure}[t]{0.3\textwidth}
        \centering
        \includegraphics[height=1.4in]{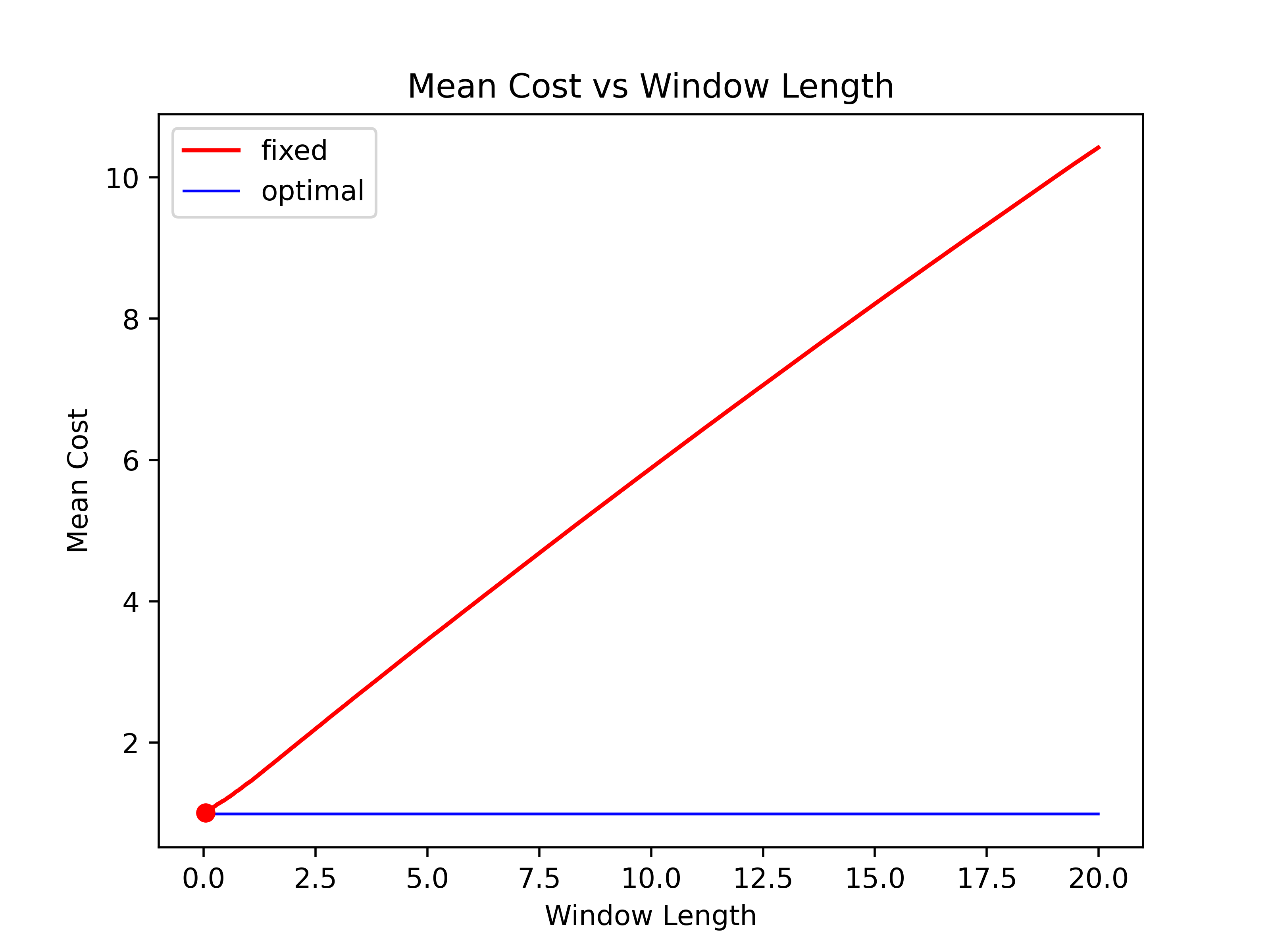}
        \caption{$c_{p} =1, c_{cs} = 1$}
    \end{subfigure}
    ~
    \begin{subfigure}[t]{0.3\textwidth}
        \centering
        \includegraphics[height=1.4in]{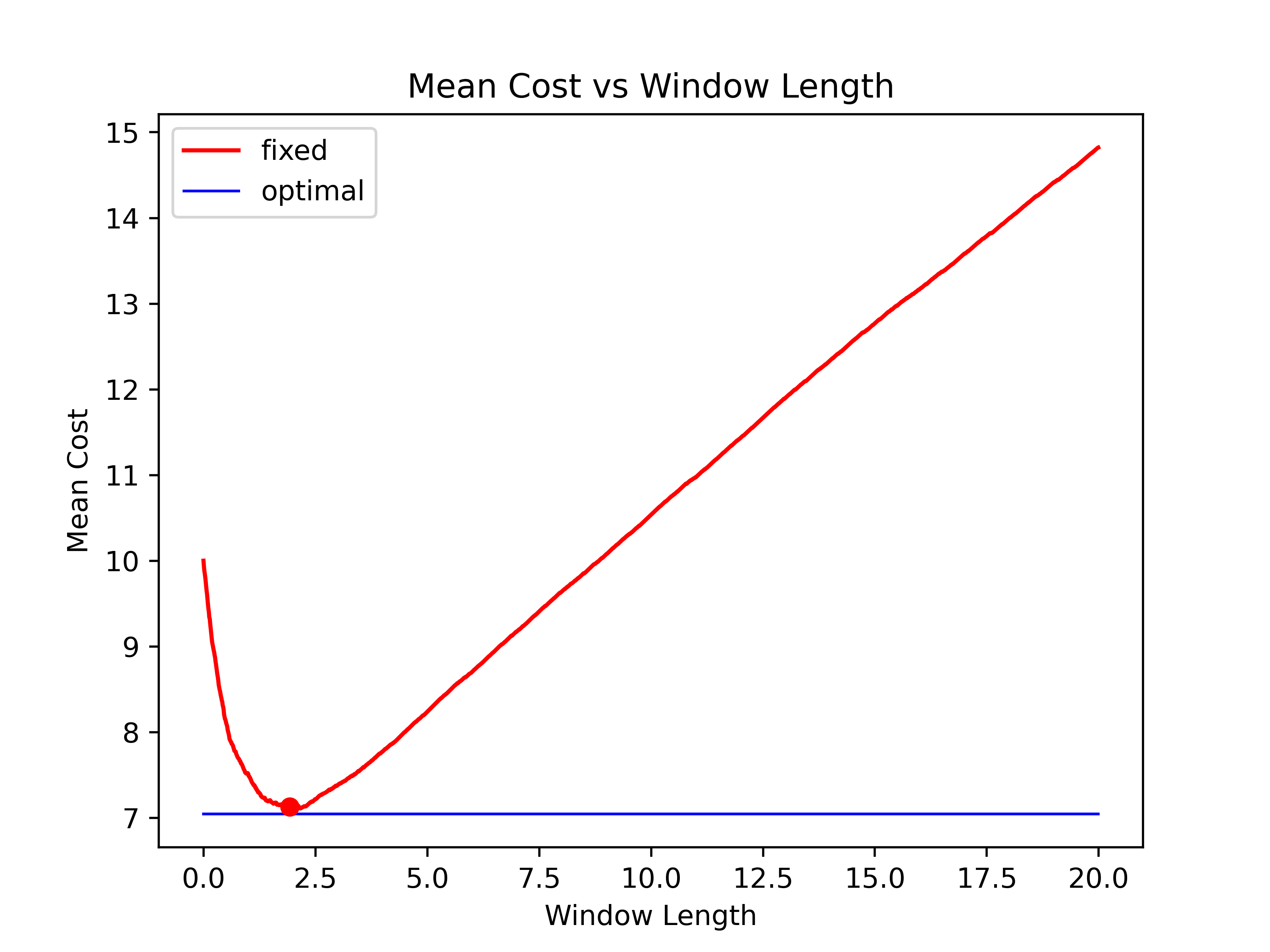}
        \caption{$c_{p} =1, c_{cs} = 10$}
    \end{subfigure}
    ~
    \begin{subfigure}[t]{0.3\textwidth}
        \centering
        \includegraphics[height=1.4in]{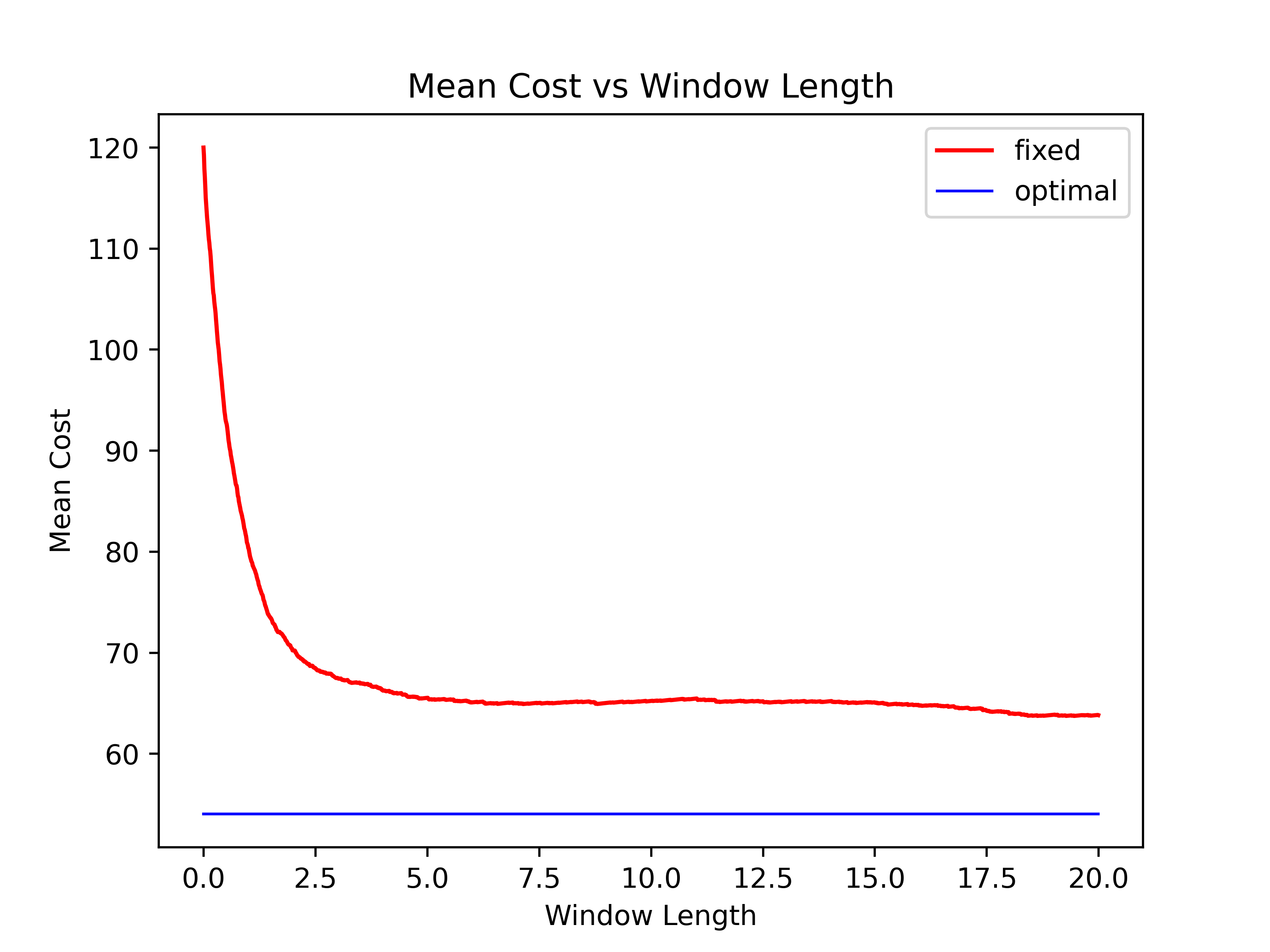}
        \caption{$c_{p} =1, c_{cs} = 120$}
    \end{subfigure}    
    \caption{Comparisons between optimal and fixed policies with Hawkes process parameters $\lambda_{0} = 0.01, \alpha = 0.5, \beta = 1.0$}
    \label{cost_sim}
\end{figure*}

Figure \ref{cost_sim} shows the average cost of the policies for the three possible cases of the optimal policy given in Corollary \ref{optimal_decreasing}.  It does so by varying $c_{cs}$ while the value of $c_{p}$ is normalized to 1. The red curve shows the performance of the fixed policy as the length of the fixed keep-alive window is varied, while the blue horizontal line indicates the performance of the optimal history dependent policy. The red dot indicates the cost of the fixed policy evaluated at the length of the average optimal keep-alive window (i.e optimized-TTL).

In Figure \ref{cost_sim} (a), since $c_{p}/c_{cs}$ is large compared to $\lambda_{0}$, the optimal keep-alive policy is to have a keep-alive window length of 0. In other words, the optimal policy is to encounter a cold start for every invocation. Here, the optimized-TTL policy is the same as optimal since the optimized-TTL window is 0.

In Figure \ref{cost_sim} (b), the red dot indicating the  average window length used by the optimal policy is also near the point where the average cost of the fixed window keep-alive policy curve changes from a decreasing function of window length to an increasing function of window length. Intuitively, this point on the red curve corresponds to the point $\tau_{\text{opt}}$ where the averaged expected cost function $\overline{g}(x) = c_{p}(1 - \overline{F}(x)) - c_{cs}\overline{f}(x)$, where the averages are taken over the different histories and runs,  changes from negative to positive. It is suggestive of some of the underlying regularities of Hawkes processes that the average of the optimal solutions and the optimal solution to the average problem are close to each other.  Thus, optimized-TTL finds a near-optimal point in the space of fixed policies and the performance there is close to the true optimum.

In Figure \ref{cost_sim} (c), we see that the optimal keep-alive policy is to always be active. Here, there is no indication of a red dot as it is out of bounds for the plot, where the red curve asymptotes toward the blue.

In the appendix we provide results for additional parameter settings showing the strong performance of optimized-TTL across a wide range of parameter settings.

\subsection{Azure Data trace Experimental Setup} \label{subsec:azure_experiments}

We evaluate the performance of the optimal policy and the optimized-TTL policy by comparing them with the fixed keep-alive policy on a subset of Azure traces released by \citet{shahrad2020serverless}.\footnote{These traces are available at \url{https://github.com/Azure/AzurePublicDatset}}. The traces  collect invocation counts of functions binned in 1-minute intervals. In Azure Functions, an application comprises of multiple functions where each function performs a specific task for the application. Since allocation of resources is based on applications (which are the unit of caching), we aggregate the bin counts of the function invocations belonging to the same application.  

We evaluate performance using two metrics. The first is the amount of memory time that is wasted, normalized to the amount wasted by the default policy of a 10 minute keep-alive window. We accumulate the wasted memory time across all applications for a given policy. We assume the function execution times to be zero, to quantify the worst case wasted memory time. For this calculation we assume that all the applications use the same amount of memory. The second is the number of cold starts. We evaluate the cold start behavior by computing the average number of cold starts per application. We assume the first invocation of an application to be a cold start. These modeling decisions are generally consistent with those of \citet{shahrad2020serverless}.

The fixed keep-alive policy was implemented by adding a fixed keep-alive window after an application invocation of, 5 minutes, 10 minutes, 20 minutes, 30 minutes, 45 minutes, 60 minutes, 90 minutes, or 120 minutes. The other two policies required fitting a Hawkes Process to the invocation pattern of an application. The Azure traces collect the data of application invocations for 14 days from July 15th to July 28th 2019. To avoid horizon effects or assuming unreasonable amounts of prior data about an application, we estimated the Hawkes process parameters based on the application invocations on day 8. We test for the appropriateness of the estimated Hawkes process parameters by using the corresponding application invocations on day 7 . We use a separate day of application invocations to check for the appropriateness of the Hawkes process parameters, since the goodness of fit test is known to have some limitations when the same data is used both to estimate the parameters and to compute the KS- statistic. \cite{reynaud2014goodness} propose a solution associated with sub-sampling. (\cite{van2016deep}  show similar issues for training and applying double Deep Q-learning Networks.) In the Appendix, we analyse this issue in detail. For simplicity, rather than sub-sampling the data we take advantage of additional data we are not currently using (e.g. day 7). We evaluate the policies for the application invocations during day 9. In our initial exploration we found that the results were largely the same when the policies were tested on other days instead, so to save on simulation time (since we are working with a datacenter-scale trace) we limited the evaluation to a single day.

We know that the data contains applications triggered by timers and other patterns which are quite different from Hawkes processes. Therefore, we applied our policies only to those applications which were a good fit to a  Hawkes process. The estimates of the Hawkes process parameters were computed as the minimum of the negative log-likelihood as described by \citet{laub2021elements}. The log-likelihood for the estimation of Hawkes process parameters $\lambda_{0}, \; \alpha , \;\text{and} \; \beta$  is
\begin{align*}
    \sum_{i=1}^{k} \log(\lambda_{0} + \alpha A(i)) - \lambda t_{k} + \frac{\alpha}{\beta} \sum_{i=1}^{k} \Big[ e^{- \beta(t_{k} - t_{i})} - 1 \Big]
\end{align*}
where $t_{1}, \; t_{2}, \cdots, \; t_{k}$ are the \textit{k} invocations of the application, and $A(i) = \sum_{j=1}^{i-1} e^{- \beta (t_{i} - t_{j})}$. The appropriateness of the application invocations being modeled by a Hawkes process is then determined by applying the Random time change theorem on the estimated parameters as detailed by  \citet{laub2021elements}. The similarity measure of the distribution of the application invocations to a Hawkes process was determined via the Kolmogorov–Smirnov (KS) test.  After testing various thresholds, we applied the optimal keep-alive policy based on the estimated Hawkes process parameters to the 25\% of application processes which had the best goodness of fit.  The default fixed policy was used for the remaining 75\%.

From Section \ref{sec:hawkes_optimal}, we know that, apart from degenerate cases, the optimal policy is given by \begin{equation*}
\tau_{\text{opt}, \mathcal{H}} = \frac{1}{\beta} \cdot \left( \log \alpha + \log \big(\sum_{j=1}^{m-1}e^{\beta (t_{j}-t_{m-1})} \big) - \log \Big( \frac{c_{p}}{c_{cs}} - \lambda_{0}\Big) \right).
\end{equation*} The second term of logarithmic sum of exponential expressed as $\log \big(\sum_{j=1}^{m-1}e^{\beta (t_{j}-t_{m-1})} \big)$ represents the weight of the previous arrivals. In order to compute the optimal keep-alive window efficiently for applications with frequent arrivals we consider no more than 200  previous arrivals. We normalize $c_{p} = 1$ and compute the optimal policy for $c_{cs} =$ 5, 10, 20, 30, 45, 60, 90, and 120. These values of $c_{cs}$ are chosen in order to match the lengths used by the fixed policy.

After obtaining the applications that are a good fit for the Hawkes process with its estimated parameters, we apply the optimized-TTL policy described as the average optimal keep-alive window policy over simulated arrivals. We compute the optimized-TTL policy based on the simulated arrivals instead of the arrivals available in the dataset because we eliminate any possible bias from using the same data to both fit the distribution and optimize the window lengths. Also, this shows that if the serverless provider has a lack of data for a particular Hawkes process application, it is still possible to get nearly equivalent performance to the optimal policy. We generate a single simulation of arrivals for each application based on the estimated Hawkes process parameters for 24 hours (1440 minutes). The average of the optimal keep-alive windows were computed for each application as their optimized-TTL based on the simulated arrivals for the various values of $c_{cs}$.

\subsection{Azure Datatrace Performance Results} \label{azure_results}


\begin{figure*}[t!]
    \centering
    \begin{subfigure}[t]{0.44\textwidth}
        \centering
        \includegraphics[height=2.1in]{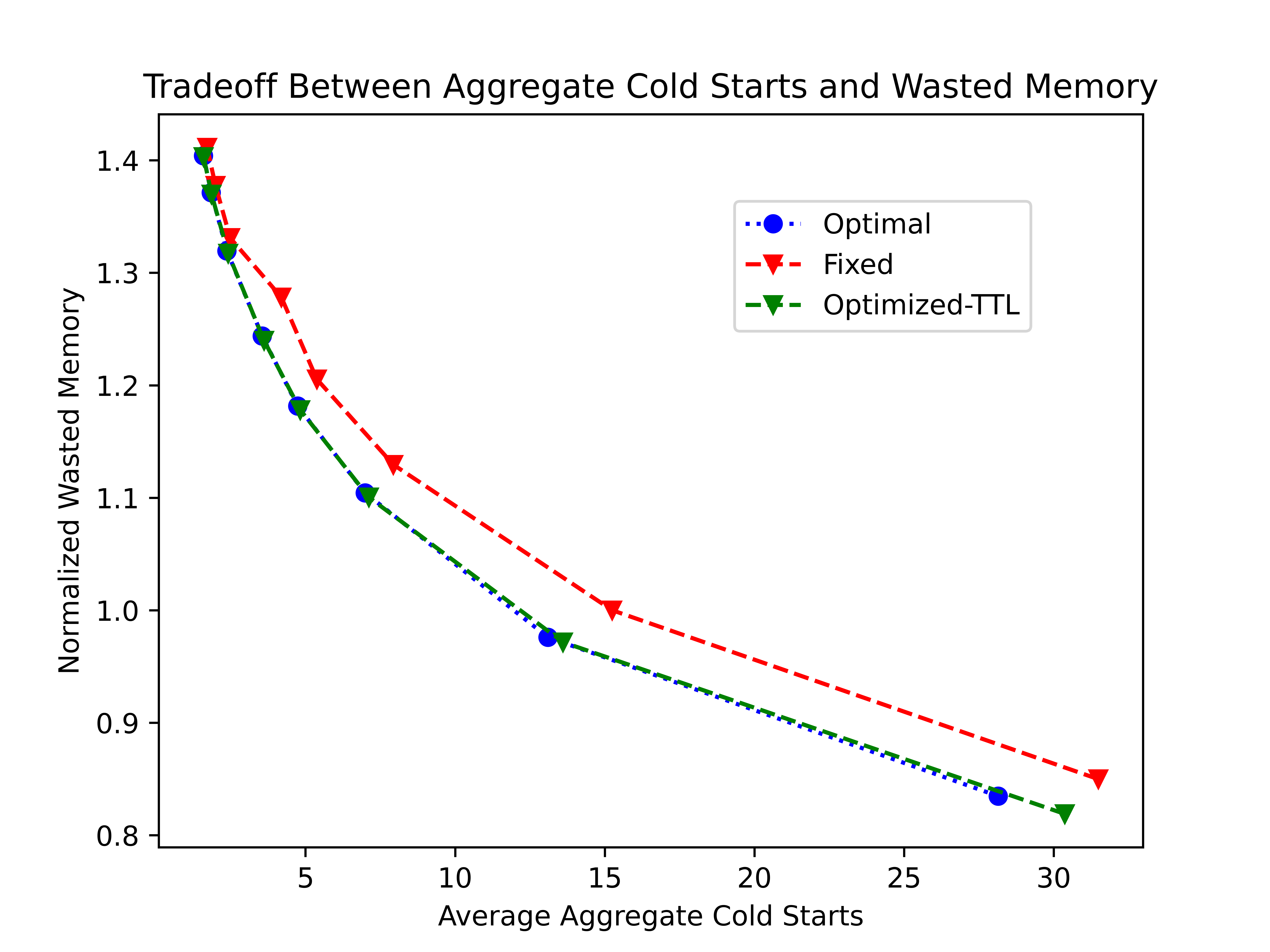}
        \caption{Hawkes process applications}
    \end{subfigure}%
    ~ 
    \begin{subfigure}[t]{0.44\textwidth}
        \centering
        \includegraphics[height=2.1in]{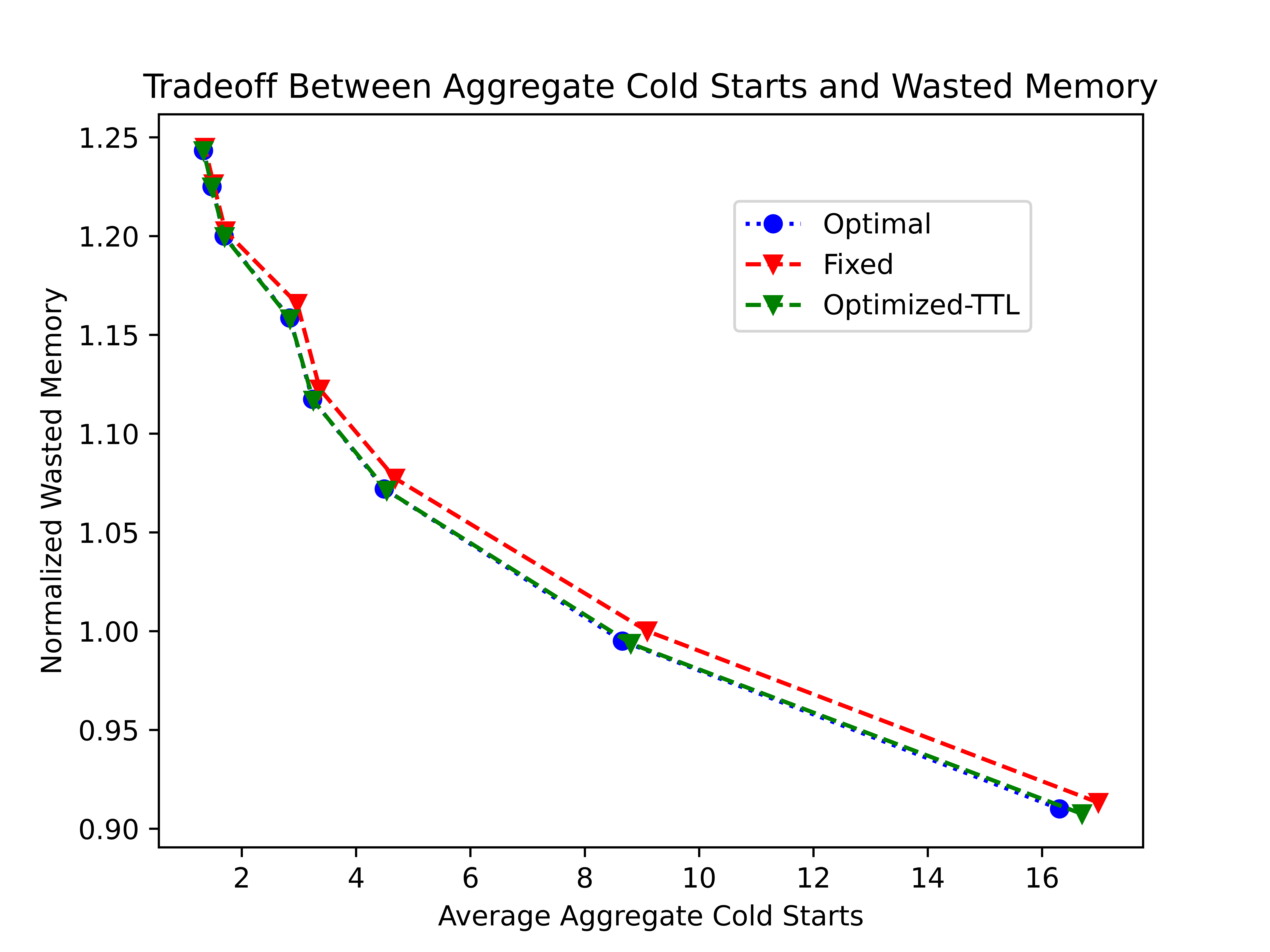}
        \caption{All applications}
    \end{subfigure}
    
    \begin{subfigure}[t]{0.44\textwidth}
        \centering
        \includegraphics[height=2.1in]{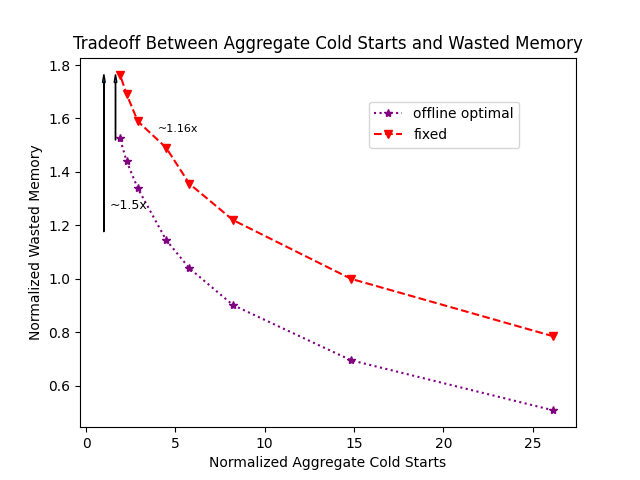}
        \caption{Hawkes process applications}
    \end{subfigure}%
    ~ 
    \begin{subfigure}[t]{0.44\textwidth}
        \centering
        \includegraphics[height=2.1in]{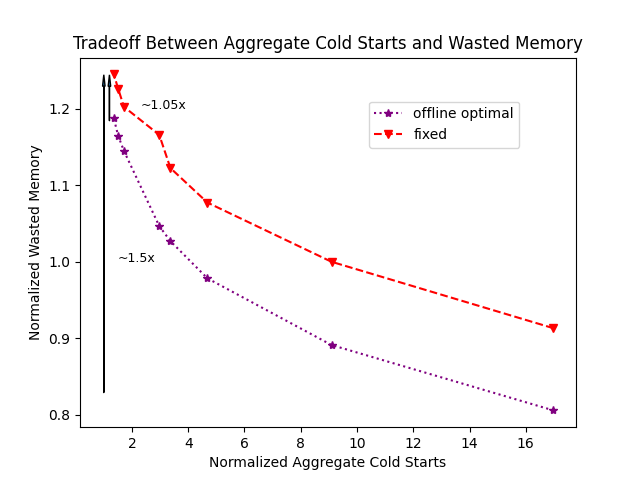}
        \caption{All applications}
    \end{subfigure}
   \caption{Trade-off curve of average number of cold starts vs normalized wasted memory for various policies}
    \label{azure_fig}
\end{figure*}

\begin{table*}[t]
\centering
\begin{tabular}{|c|c|c|c|c|}
\hline
Policy & Avg. Cold Start Savings (Hawkes) & (All) & Avg. Memory Savings (Hawkes) & (All)\\
\hline
Optimal & 1.012 &  0.1296 &  0.0457 & 0.0095 \\
\hline
Optimized-TTL & 0.965 & 0.1196 & 0.0436 & 0.0088\\
\hline
Offline-Optimal & 4.038 & 1.29 & 0.269 & 0.0948\\
\hline
\end{tabular}
\caption{\label{tab:1} Average performance improvement over fixed policy}
\end{table*}

Overall, our results show that applying the optimal policy to applications fit by a Hawkes process yields benefits that are economically significant at the scale of a datacenter.  Furthermore, the optimized-TTL policy yields near-equivalent results with no need to update based on history.

Figure \ref{azure_fig} (a,b) plots the trade-off curve between the average number of cold starts per application and the normalized wasted memory for optimal, optimized-TTL, and fixed policies. Figure \ref{azure_fig} (a) includes only the treated (Hawkes-process-like) applications, while Figure \ref{azure_fig} (b) includes all applications. We observe from Figure \ref{azure_fig}, that the optimal policy Pareto dominates the fixed keep-alive policy.  We see that for high values of $c_{cs}$ both metrics are almost the same for all policies. This is because, the optimal policy for many applications was selected to be the upper bound.  Similarly, for low values of $c_{cs}$ many applications use the lower bound. For intermediate values of $c_{cs}$ we see that the optimal policy has almost the same number of normalized cold starts but lower amount of wasted memory. Furthermore, the performance of the optimized-TTL policy is very similar to the optimal policy so we can get these benefits.  Of course, since we only treat 25\% of applications, these benefits are attenuated when considering all applications.

While the benefits are small in absolute terms we argue they are still economically significant.  We quantify them by computing the area between the Pareto curve of the fixed and the other two policies. By  dividing the area between the Pareto curves with the maximum number of average aggregate cold starts encountered by an application (effective x-axis length) or maximum normalized wasted memory (effective y-axis length). This gives a sense of the average improvement across the curve. We record the results in Table 1.


They show that the improvement of the optimized-TTL policy over the fixed policy is around 95\% of that of the optimal policy indicating that the optimized-TTL policy performs effectively close to the history dependent policy. Given that wasted memory in the existing system is normalized to 1,  average improvements of approximately 0.045 on treated applications and 0.0095 overall represent an overall decrease in resource use for the cache on the order of 4.5\% on treated applications (0.95\% overall).\footnote{Because of the shape of the curves in Figure~\ref{azure_fig}, the benefits may be modestly larger in practice because current operating point is toward the right end of the plots where the gap tends to be larger.} While small in absolute terms, percentage improvements of this scale translate to tens or hundreds of millions of dollars in cost savings for a major cloud provider~\cite{dierks2019cluster}.  Alternatively, this same improvement could be used to improve customer experience,  reducing the number of cold starts by an average of about 1.012 per day for treated customers.

To put these results in further context, recall that we are comparing against the fixed policy {\em tuned optimally per our theory}.  In terms of average cost, simply using the default choice of 10 minutes would be substantially worse in many cases.  Furthermore, our theory shows that this fixed policy is not a weak baseline but has strong theoretical properties in its own right. (See discussion at the end of Section~\ref{subsec:avg_optimal}.)

Finally, to give a bound on how much of the possible performance improvement our approach achieves, we also evaluate the performance of the offline optimal policy in Figure \ref{azure_fig} for (c) only the applications that have a Hawkes process distribution (d) and all applications. Here, the offline optimal policy chooses to have a cold start if $c_{p} x_m \geq c_{cs}$, otherwise the keep-alive window is of length $x_m$. Here, $x_{m}$ is the inter-arrival time for the $m$-th arrival of the application invocation. That is, this policy cheats in that it is tuned to the actual realized pattern of arrivals rather than any prediction. In this sense it provides a bound on how much we could hope to achieve. The results from Figure \ref{azure_fig} (c), and also quantified in Table 1, show that the optimal Hawkes policy achieves a meaningful fraction of it.  Figure \ref{azure_fig} (d) includes a comparison to the improvements reported by \citet{shahrad2020serverless}.  They report a 1.5x improvement in wasted memory at essentially no cost in cold starts, which is substantially larger than what even the offline optimal can achieve (1.05x).  This highlights how much of their improvement comes from pre-warming apps, for example timers, which our theoretical results show is not necessary for Hawkes processes.  They combine sophisticated predictive modeling with a simple rule for determining pre-warming and keep-alive decisions.  In contrast, our results provide a sophisticated rule for these decisions, making them complementary.

\section{Conclusion} \label{sec:conclusion}

Motivated by applications such as serverless computing, we presented a model of caching policies which captures the trade-off between the cost of keeping objects in the cache and the cost of cache misses.
We characterized optimal caching policies and examined the optimal policies in detail for Hawkes processes. Since optimal policies for Hawkes processes depend on the history of arrivals, we also developed history-independent policies based on the heuristic of averaging the optimal keep-alive window from simulated predictions of arrivals. Evaluation on Hawkes process simulations provided insights into the tuning and expected performance of these approximations.  Evaluation on a data trace of Azure functions showed this approach can yield small,yet economically meaningful improvements at the scale of a datacenter. Our results point to several avenues for future work.  Since our approach allows us to characterize optimal policies on a per-item basis, it is naturally suited for exploring customization based on individual customer utilities rather than an overall system average as done in our experiments and most prior work.  Another direction would be to use our model to examine optimal policies for more complex scenarios, such as a hierarchy of caches.

\subsection*{Acknowledgements}

This material is based upon work supported by the National Science Foundation under Grant No. 2110707.

\bibliography{keep_alive_caching_621}

\clearpage
\newpage

\section*{Appendix}
\appendix

This Supplementary Material contains proofs and other material omitted from the main manuscript.

\section{Omitted Proofs}

\begin{lemma}
The expected cost of a cache policy over an inter-arrival is
\begin{equation*}
\mathbb{E}[cost(\pi(\cdot |\mathcal{H}_{m-1}))] = c_{cs} + \displaystyle \int_{0}^{\infty} \pi(x|\mathcal{H}_{m-1}) \cdot g(x|\mathcal{H}_{m-1})\; dx , 
\end{equation*} 
where the instantaneous cost at \textit{x} units after the most recent arrival at $t_{m-1}$ is
\begin{equation*}
g(x|\mathcal{H}_{m-1}) = c_{p} \cdot \big(1 - F(x|\mathcal{H}_{m-1})\big) - c_{cs} \cdot f(x| \mathcal{H}_{m-1}).
\end{equation*} 
\end{lemma}

\begin{proof}
Let $\mathcal{L} = \{L_{0}, L_{1}, L_{2}, \cdots, L_{2k-1}\}$ denote the set of points on the sequence of keep-alive windows for policy $\pi(\cdot | \mathcal{H}_{m-1})$ where even indices are the start of the windows and odd indices are the endpoints of the windows.  Let $Z(\mathcal{L}, j) = \sum_{p = 0}^{j} (L_{2p+1} - L_{2p})$ for $j \geq 0$. The function $Z(\mathcal{L}, j)$ represents the time accumulated in the cache through the $j$-th sequence of the keep-alive window. For $j < 0$, we have $Z(\mathcal{L},j) = 0$.  Then we have

\begin{align*}
\begin{split}
& \mathbb{E}[cost(\pi(\cdot|\mathcal{H}_{m-1}))]\\ 
& =  c_{cs} \cdot \displaystyle \int_{0}^{L_{0}} f(x | \mathcal{H}_{m-1}) \; dx \; + \;  \displaystyle \sum_{j = 0}^{k-1} \int_{L_{2j}}^{L_{2j+1}} \; c_{p} \cdot \Big( Z(\mathcal{L}, j-1) + x - L_{2j} \Big) \cdot f(x|\mathcal{H}_{m-1}) \; dx \; \\ & \quad + \sum_{j = 0}^{k-2} \int_{L_{2j+1}}^{L_{2j+2}} \; \Big( c_{cs} + c_{p} Z(\mathcal{L}, j) \Big) f(x|\mathcal{H}_{m-1}) \; dx \; + \int_{L_{2k-1}}^{\infty} \; \Big( c_{cs} + c_{p} Z(\mathcal{L}, k-1) \Big) f(x|\mathcal{H}_{m-1}) \; dx  \quad \text{(1)}\\
& = c_{cs} \cdot F(L_{0} | \mathcal{H}_{m-1}) + \displaystyle \sum_{j = 0}^{k-1} \Bigg( c_{p} \cdot \Big( Z(\mathcal{L}, j-1) + x - L_{2j} \Big) \cdot F(x|\mathcal{H}_{m-1})\Big|_{L_{2j}}^{L_{2j+1}}  - \int_{L_{2j}}^{L_{2j+1}} c_{p} F(x|\mathcal{H}_{m-1}) \; dx\Bigg) \;  \\  & \quad + \sum_{j = 0}^{k-2} \; \Big( c_{cs} + c_{p} \cdot Z(\mathcal{L}, j) \Big) \cdot F(x|\mathcal{H}_{m-1})\Big|_{L_{2j+1}}^{L_{2j+2}} + \Big( c_{cs} + c_{p} \cdot Z(\mathcal{L}, k-1) \Big) \cdot F(x|\mathcal{H}_{m-1})\Big|_{L_{2k-1}}^{\infty}  \qquad  \text{(2)}\\
& =  c_{cs} F(L_{0} | \mathcal{H}_{m-1}) + \displaystyle \sum_{j = 0}^{k-1} \; c_{p} \Big( Z(\mathcal{L}, j-1) + L_{2j+1} - L_{2j} \Big) F(L_{2j+1}|\mathcal{H}_{m-1})   - \sum_{j = 0}^{k-1} \; c_{p} Z(\mathcal{L}, j-1) F(L_{2j}|\mathcal{H}_{m-1}) \\ & \quad - \sum_{j = 0}^{k-1} \; \int_{L_{2j}}^{L_{2j+1}} c_{p} F(x|\mathcal{H}_{m-1}) \; dx + \sum_{j = 0}^{k-2} \; \Big( c_{cs} + c_{p} \cdot Z(\mathcal{L}, j) \Big) \cdot F(L_{2j+2}|\mathcal{H}_{m-1}) \\ & \quad - \sum_{j = 0}^{k-2} \; \Big( c_{cs} + c_{p} \cdot Z(\mathcal{L}, j) \Big) F(L_{2j+1}|\mathcal{H}_{m-1}) + \Big( c_{cs} + c_{p} \cdot Z(\mathcal{L}, k-1) \Big) \cdot 1 \\ & \quad - \Big( c_{cs} + c_{p} \cdot Z(\mathcal{L}, k-1) \Big) F(L_{2k-1}|\mathcal{H}_{m-1}) \qquad \qquad \text{(3)}\\
\end{split}    
\end{align*}

\begin{align*}
\begin{split}
 &\mathbb{E}[cost(\pi(\cdot|\mathcal{H}_{m-1}))] \\
& =  c_{cs} F(L_{0} | \mathcal{H}_{m-1}) \; + \displaystyle \sum_{j = 0}^{k-1} \; c_{p} Z(\mathcal{L}, j) F(L_{2j+1}| \mathcal{H}_{m-1})   - \sum_{j = 0}^{k-1} \; c_{p} Z(\mathcal{L}, j-1) F(L_{2j}| \mathcal{H}_{m-1}) \; \\ & \quad - \sum_{j = 0}^{k-1} \; \int_{L_{2j}}^{L_{2j+1}} c_{p} F(x|\mathcal{H}_{m-1}) \; dx \quad + \sum_{j = 0}^{k-2} \; c_{cs}  \cdot F(L_{2j+2}| \mathcal{H}_{m-1}) + \sum_{j = 0}^{k-2} \; c_{p} \cdot Z(\mathcal{L}, j) \cdot F(L_{2j+2}|\mathcal{H}_{m-1}) \; \\ & \quad - \sum_{j = 0}^{k-2} \;  c_{cs} \cdot F(L_{2j+1}|\mathcal{H}_{m-1}) \quad - \sum_{j = 0}^{k-2} c_{p} \cdot Z(\mathcal{L}, j) \cdot F(L_{2j+1}|\mathcal{H}_{m-1}) \\& \quad +  c_{cs} + c_{p} \cdot Z(\mathcal{L}, k-1) \; - c_{cs} \cdot F(L_{2k-1}|\mathcal{H}_{m-1}) \quad - c_{p} \cdot Z(\mathcal{L}, k-1) \cdot F(L_{2k-1}|\mathcal{H}_{m}) \qquad \qquad \text{(4)}\\   
\end{split}    
\end{align*}

We apply integration by parts to the term $\int_{L_{2j}}^{L_{2j+1}} \; c_{p} \cdot \big( Z(\mathcal{L}, j-1) + x - L_{2j} \big) \cdot f(x|\mathcal{H}_{m-1}) \; dx \;$ in Equation (1) to get the terms $c_{p} \cdot \big( Z(\mathcal{L}, j-1) + x - L_{2j} \big) \cdot F(x|\mathcal{H}_{m-1})\big|_{L_{2j}}^{L_{2j+1}}  - \int_{L_{2j}}^{L_{2j+1}} c_{p} F(x|\mathcal{H}_{m-1}) \; dx \;$ in Equation (2). In Equation (4), we have substituted $ \; Z(\mathcal{L},j) \;$ for the terms $\; Z(\mathcal{L},j-1) + (L_{2j+1} - L_{2j})$ in Equation (3) since, $\; Z(\mathcal{L},j) = Z(\mathcal{L},j-1) + (L_{2j+1} - L_{2j})$. The remainder of the proof consists of combining and canceling terms to simplify (4), then applying the fundamental theorem of calculus.

\begin{align*}
\begin{split}
&\mathbb{E}[cost(\pi(\cdot|\mathcal{H}_{m-1}))] \\
& = \; c_{cs} F(L_{0} | \mathcal{H}_{m-1}) \; + \displaystyle \sum_{j = 0}^{k-1} \; c_{p} Z(\mathcal{L}, j) \cdot F(L_{2j+1}| \mathcal{H}_{m-1}) - \sum_{j = 0}^{k-1} \; c_{p} Z(\mathcal{L}, j-1) \cdot F(L_{2j}|\mathcal{H}_{m-1}) \; \\ & \quad - \sum_{j = 0}^{k-1} \; \int_{L_{2j}}^{L_{2j+1}} c_{p} F(x|\mathcal{H}_{m-1}) \; dx  \quad + \sum_{j = 0}^{k-2} \;  c_{cs} \cdot F(L_{2j+2}|\mathcal{H}_{m-1}) \quad + \sum_{j = 0}^{k-2} \;  c_{p} \cdot Z(\mathcal{L}, j) \cdot F(L_{2j+2}|\mathcal{H}_{m-1}) \; \\ & \quad  - \sum_{j = 0}^{k-1} \;  c_{cs} \cdot F(L_{2j+1}|\mathcal{H}_{m-1}) \quad- \sum_{j = 0}^{k-1} c_{p} \cdot Z(\mathcal{L}, j) \cdot F(L_{2j+1}|\mathcal{H}_{m-1}) +  c_{cs} + c_{p} \cdot Z(\mathcal{L}, k-1)\\
& = \; \displaystyle - \sum_{j=0}^{k-1} \int_{L_{2j}}^{L_{2j+1}} c_{p} F(x|\mathcal{H}_{m-1}) \;dx + \; \sum_{j= 0}^{k-1} c_{cs} F(L_{2j}|\mathcal{H}_{m-1}) - \sum_{j = 0}^{k-1} c_{cs} F(L_{2j+1}|\mathcal{H}_{m-1}) + c_{cs} + c_{p} Z(\mathcal{L}, k-1) \quad \text{(5)}\\
& = \quad \displaystyle   c_{p} \cdot Z(\mathcal{L}, k-1) - \sum_{j = 0}^{k-1} \; \int_{L_{2j}}^{L_{2j+1}} c_{p} F(x|\mathcal{H}_{m-1}) \; dx  - \sum_{j = 0}^{k-1} \;  c_{cs} \cdot \Big( F(L_{2j+1}|\mathcal{H}_{m-1}) - F(L_{2j}|\mathcal{H}_{m-1}) \Big) + c_{cs}\\
& = \quad \displaystyle \int_{\mathcal{I}} c_{p}(1 - F(x|\mathcal{H}_{m-1})) -c_{cs} \cdot f(x|\mathcal{H}_{m-1}) \; dx \; + \; c_{cs}
\end{split}    
\end{align*}

\noindent where $\mathcal{I} =\begin{cases}1, \; \text{for} \quad x \in [L_{0} , \; L_{1}] \bigcup \cdots \bigcup \; [L_{2k-2}, L_{2k-1}]\\ 0, \; \text{otherwise} \end{cases} $.\\

\noindent In Equation (5),  we combine \; $c_{cs} \cdot F(L_{0} | \mathcal{H}_{m-1})$ \; and \; $\displaystyle \sum_{j=0}^{k-2} \; c_{cs} F(L_{2j+2} | \mathcal{H}_{m-1})$ \; to obtain $\displaystyle \sum_{j=0}^{k-1} c_{cs} F(L_{2j} | \mathcal{H}_{m-1})$.

\end{proof}

\begin{theorem} 
The points $L_i$ of the sequence of keep-alive windows over an inter-arrival for the optimal policy $\pi_{\text{opt}}(\cdot | \mathcal{H}_{m-1})$ are at 0, $\infty$, or solutions to the equation $c_{p} - (c_{cs} \cdot \lambda(x | \mathcal{H}_{m-1})) = 0 \;$ where the sign changes.
\end{theorem}

\begin{proof}
From Lemma \ref{expect_cost} we know that the expected cost is given by $\mathbb{E}[cost(\pi(\cdot|\mathcal{H}_{m-1}))] = \int_{0}^{\infty} \pi(x|\mathcal{H}_{m-1}) \cdot g(x|\mathcal{H}_{m-1}) \; dx +c_{cs}$. The points of the sequence of keep-alive windows $L_{k}\; \text{for} \; k = 0, 1, 2, \cdots$ for the optimal policy are the points where the first order partial derivative of the expected cost is zero, that is, $\displaystyle \frac{\partial\; \mathbb{E}[cost(\pi(\cdot|\mathcal{H}_{m-1}))]}{\partial x} = \; 0$ at $x = L_{k}\; \forall k$. The first order derivative of the expected cost is $\displaystyle \frac{\partial\; \mathbb{E}[cost(\pi(\cdot|\mathcal{H}_{m-1}))]}{\partial x} \; = \; g(x |\mathcal{H}_{m-1})$. We simplify $g(x |\mathcal{H}_{m-1})$ as follows,

\begin{align*}
\begin{split}
& g(x | \mathcal{H}_{m-1}) \\ 
& = c_{p}(1 - F(x | \mathcal{H}_{m-1})) - c_{cs} f(x | \mathcal{H}_{m-1}) \\
& = c_{p} (1 - F(x | \mathcal{H}_{m-1})) - c_{cs} \lambda(x | \mathcal{H}_{m-1}) (1 - F(x|\mathcal{H}_{m-1})) \\
& = (1 - F(x | \mathcal{H}_{m-1})) \cdot (c_{p} - c_{cs}\lambda(x | \mathcal{H}_{m-1}))
\end{split}    
\end{align*}

\begin{align*}
& g(x = L_{k} |\mathcal{H}_{m-1}) \; = \; 0   \\
\implies & 1 - F(x = L_{k} |\mathcal{H}_{m-1}) = 0 \quad \text{or} \quad  c_{p} - c_{cs} \lambda(x = L_{k} |\mathcal{H}_{m-1}) = 0 \\
\implies & F(x = L_{k}|\mathcal{H}_{m-1}) = 1 \quad \text{or} \quad \displaystyle \frac{c_{p}}{c_{cs}} = \lambda(x = L_{k} |\mathcal{H}_{m-1}) = \frac{f(x = L_{k}|\mathcal{H}_{m-1})}{1 - F(x = L_{k}| \mathcal{H}_{m-1})}
\end{align*}

\noindent Hence, the points where $g(x | \mathcal{H}_{m-1}) = 0$ are also the points where $c_{p} - c_{cs} \lambda(x | \mathcal{H}_{m-1}) = 0$. We know that the instantaneous cost of the policy over an inter-arrival is given by $g(x | \mathcal{H}_{m-1})$.  Let $x = L_{k}$ be an arbitrary root of the equation $c_{p}- c_{cs} \lambda(x |\mathcal{H}_{m-1}) = 0$. If $g(x |\mathcal{H}_{m-1})$ changes sign from positive to negative as it goes through $x = L_{k}$, that is, $g(x |\mathcal{H}_{m-1}) >\; 0\;$ for $x < L_{k}$ and changes sign to $g(x | \mathcal{H}_{m-1}) < \; 0\;$ for $x > L_{k}$. It would be optimal for the cache policy to start the keep-alive window from $x = L_{k}$, since the cost of caching the object from $x = L_{k}$ benefits the policy. Similarly, if $g(x |\mathcal{H}_{m-1})$ changes sign from negative to positive as it passes $x = L_{k}$, that is, $g(x |\mathcal{H}_{m-1}) <\; 0\;$ for $x < L_{k}$ and changes sign to $g(x | \mathcal{H}_{m-1}) > \; 0\;$ for $x > L_{k}$. It would be optimal to stop the keep-alive window after $x = L_{k}$, since the cost of caching the object after $x = L_{k}$ will not benefit the policy.  Since $1 - F(x | \mathcal{H}_{m-1}) \; \geq \; 0,  \quad \forall x,\;$ the sign of $c_{p} - c_{cs} \cdot \lambda(x | \mathcal{H}_{m-1})$ determines the sign of $g(x |\mathcal{H}_{m-1})$. The sign of $c_{p} - c_{cs} \lambda(x | \mathcal{H}_{m-1})$ should change as it passes through the root of the equation $x = L_{k}$ for $x = L_{k}$ to be considered as a point where the keep-alive window of the optimal policy starts or ends.

This leaves the end cases where there is no solution to the equation $c_{p} - c_{cs} \lambda(x | \mathcal{H}_{m-1}) = 0$ or the sign of $c_{p} - c_{cs} \lambda(x | \mathcal{H}_{m-1})$ does not change $\forall \; x$. When $g(x |\mathcal{H}_{m-1})$ is always non-negative, it is optimal to have a keep-alive window length of 0. This is because having an active keep-alive window length in this case would be more expensive than a cold start. On the other hand, when $g(x |\mathcal{H}_{m-1})$ is always non-positive, it is optimal for the keep-alive window to always be active since keeping the object in cache is beneficial to the policy. 
    
\end{proof}

\begin{corollary}
If the hazard rate is weakly decreasing, the optimal policy $\pi_{\text{opt}}(\cdot|\mathcal{H}_{m-1})$  is a single keep-alive window starting at $\tau_{pw}=0$, and is given by \\ $\pi_{\text{opt}}(x|\mathcal{H}_{m-1}) =  \begin{cases}1 \; , \quad \forall x \in [0 , \; \tau_{\text{opt}, \mathcal{H}_{m-1}}]\\ 0 \; , \quad \text{otherwise}\end{cases}$ where,
\begin{enumerate}
    \item $\tau_{\text{opt}, \mathcal{H}_{m-1}} = \infty$, i.e., the optimal policy is to have the keep-alive window always be active when \\ $\forall x, \; \displaystyle \frac{c_{p}}{c_{cs}} < \; \lambda(x | \mathcal{H}_{m-1}), \; $
    \item $\tau_{\text{opt}, \mathcal{H}_{m-1}} = 0$, i.e., the optimal policy would be to not cache and always have a cold start when \\ $\displaystyle \frac{c_{p}}{c_{cs}} >\; \lambda(x = 0 | \mathcal{H}_{m-1})\;$
    \item The optimal policy is a keep-alive window of length $\tau_{\text{opt}, \mathcal{H}_{m-1}}$ given by the solution to the equation \\ $ \displaystyle  \frac{c_{p}}{c_{cs}} = \frac{f(x = \tau_{\text{opt}, \mathcal{H}_{m-1}} | \mathcal{H}_{m-1})}{1 - F(x = \tau_{\text{opt}, \mathcal{H}_{m-1}} | \mathcal{H}_{m-1})}$, otherwise. 
\end{enumerate}
\end{corollary}

\begin{proof}
Assume that the hazard rate of the arrival of function invocations over an inter-arrival $\lambda (x|\mathcal{H}_{m-1})$ is (weakly) decreasing. From Lemma \ref{expect_cost} we know that the expected cost is given by $\mathbb{E}[cost(\pi(\cdot|\mathcal{H}_{m-1}))] = \int_{0}^{\infty} \pi(x|\mathcal{H}_{m-1}) \cdot g(x|\mathcal{H}_{m-1}) \; dx +c_{cs}$. We prove a single keep-alive window is optimal by showing that $g(x|\mathcal{H}_{m-1})$ can only change its sign from negative to positive at most once. 
Thus $g$ is optimized by single window policy that keeps the object in cache until the transition from negative to positive occurs.
To begin,
\begin{equation*}
g(x |\mathcal{H}_{m-1})  = \; (1 - F(x|\mathcal{H}_{m-1})) \cdot (c_{p} - c_{cs} \lambda(x| \mathcal{H}_{m-1}))
\end{equation*}

\noindent Since $\lambda(x|\mathcal{H}_{m-1})$ is weakly decreasing, $c_{p} - c_{cs} \lambda(x| \mathcal{H}_{m-1})$ is weakly increasing. Also, $1 - F(x|\mathcal{H}_{m-1})$ is always positive. If $g(x|\mathcal{H}_{m-1}) \; \geq 0$, then the optimal policy is to have a keep-alive window length of 0. This is because having an active keep-alive window in this case would be more expensive than a cold start. If $g(x|\mathcal{H}_{m-1})$ is always negative, then it is always beneficial for the keep-alive window to be active. Otherwise, $g(x|\mathcal{H}_{m-1})$ can change its sign at most once and such a change must be from negative to positive. It is no longer beneficial for the provider to keep things in memory after $g(x|\mathcal{H}_{m-1})$ has changed from negative to positive because the cost of keeping in memory outweighs the cost of a cold start. Thus, the optimal policy is of the form of a single keep-alive window. Now, it only remains to determine the point$\tau_{\text{opt}, \mathcal{H}_{m-1}}$ at which $g(x|\mathcal{H}_{m-1})$ changes from negative to positive.

\begin{align*}
& g(x = \tau_{\text{opt}, \mathcal{H}_{m-1}}|\mathcal{H}_{m-1}) \; = \; 0   \\
\implies & F(x = \tau_{\text{opt},\mathcal{H}_{m-1}}|\mathcal{H}_{m-1}) = 1 \quad \text{or} \quad \displaystyle \frac{c_{p}}{c_{cs}} = \lambda(x = \tau_{\text{opt},\mathcal{H}_{m-1}}|\mathcal{H}_{m-1}) = \frac{f(x =\tau_{\text{opt},\mathcal{H}_{m-1}}|\mathcal{H}_{m-1})}{1 - F(x = \tau_{\text{opt},\mathcal{H}_{m-1}}| \mathcal{H}_{m-1})}
\end{align*}

\end{proof}

\begin{corollary} 
If the hazard rate is weakly increasing, the optimal policy $\pi_{\text{opt}}(\cdot|\mathcal{H}_{m-1})$  is a single keep-alive window with $\tau_{ka} = \infty$ and a pre-warming window, and is given by \\ $\pi_{\text{opt}}(x|\mathcal{H}_{m-1}) =  \begin{cases}1 , \; \tau_{\text{pw}, \mathcal{H}_{m-1}} \leq x \\ 0, \; \text{otherwise}\end{cases}$ where,
\begin{enumerate}
    \item $\tau_{\text{pw}, \mathcal{H}_{m-1}} = 0$, i.e., the optimal policy is to have the keep-alive window always be active when\\ $\forall x, \quad \displaystyle \frac{c_{p}}{c_{cs}} < \; \lambda(x | \mathcal{H}_{m-1}), \; $
    \item $\tau_{\text{pw}, \mathcal{H}_{m-1}} = \infty$, i.e., the optimal policy is to always have a cold start when $\forall x, \; \displaystyle \frac{c_{p}}{c_{cs}} > \; \lambda(x | \mathcal{H}_{m-1}) \;$. 
    \item $\tau_{\text{pw}, \mathcal{H}_{m-1}}$ satisfies the equation\\ $\; \displaystyle \frac{c_{p}}{c_{cs}} = \displaystyle \frac{f(x = \tau_{\text{pw}, \mathcal{H}_{m-1}} |  \mathcal{H}_{m-1})}{1 - F(x = \tau_{\text{pw}, \mathcal{H}_{m-1}} |  \mathcal{H}_{m-1})}$, i.e., an infinite keep-alive window after a pre-warming window of length $\tau_{\text{pw}, \mathcal{H}_{m-1}}$ when $c_{p} - c_{cs} \lambda(x = 0 | \mathcal{H}_{m-1}) \; > 0$ and changes sign. 
\end{enumerate}
\end{corollary}

\begin{proof}

Following the proof of Theorem \ref{opt_policy}, we know that $g(x | \mathcal{H}_{m-1}) = (1 - F(x| \mathcal{H}_{m-1}))\cdot (c_{p} - c_{cs} \cdot \lambda(x| \mathcal{H}_{m-1}))$. Since $\lambda(x | \mathcal{H}_{m-1})$ is weakly increasing, $c_{p} - c_{cs} \lambda(x|\mathcal{H}_{m-1})$ is weakly decreasing. Also, $1 - F(x|\mathcal{H}_{m-1})$ is always positive. If initially $g(x | \mathcal{H}_{m-1}) < \; 0$, then $g(x |\mathcal{H}_{m-1})$ will always be negative. Hence, it is optimal to have the keep-alive window always be active. If $g(x|\mathcal{H}_{m-1}) > 0, \quad \forall x$, that is, $g$ is always positive, then the optimal policy is to encounter a cold start. If $g(x |\mathcal{H}_{m-1})$ is initially positive, then changes to a negative sign as $\lambda(x | \mathcal{H}_{m-1})$ is weakly increasing, then the optimal policy will be a pre-warming window of length decided by the position of the change of sign. We obtain $\tau_{\text{pw}, \mathcal{H}_{m-1}}\;$ from solving $g(x = \tau_{\text{pw},\mathcal{H}_{m-1}}|\mathcal{H}_{m-1}) = 0$ as follows.

\begin{align*}
& g(x = \tau_{\text{pw}, \mathcal{H}_{m-1}}|\mathcal{H}_{m-1}) \; = \; 0   \\
\implies & F(x = \tau_{\text{pw},\mathcal{H}_{m-1}}|\mathcal{H}_{m-1}) = 1 \quad \text{or} \quad \displaystyle \frac{c_{p}}{c_{cs}} = \lambda(x = \tau_{\text{pw},\mathcal{H}_{m-1}}|\mathcal{H}_{m-1}) = \frac{f(x =\tau_{\text{pw},\mathcal{H}_{m-1}}|\mathcal{H}_{m-1})}{1 - F(x = \tau_{\text{pw},\mathcal{H}_{m-1}}| \mathcal{H}_{m-1})}
\end{align*}

After the sign changes to negative, the keep-alive window should always be active,, that is, $\tau_{\text{ka},\mathcal{H}_{m-1}} = \infty$. 
\end{proof}

Corollary \ref{optimal_decreasing} characterizes the optimal policy when the distribution of arrival requests follow the Hawkes process to be one of the following policies,
\begin{itemize}
    \item The keep-alive window is to always be active with $\tau_{\text{opt}, \mathcal{H}_{m-1}} = \infty$ when, as in the Poisson case, the background intensity is sufficiently high: $\displaystyle \frac{c_{p}}{c_{cs}} < \; \lambda_{0}$.
    \item Experience a cold start with $\tau_{\text{opt}, \mathcal{H}_{m-1}} = 0$ when $\displaystyle \frac{c_{p}}{c_{cs}} > \; \lambda(x| \mathcal{H}_{m-1}))$, after the most recent arrival request.
    \item The keep-alive window is given by the expression 
    \begin{equation*}
        \tau_{\text{opt}, \mathcal{H}_{m-1}} = \frac{1}{\beta} \Big( \log \alpha + \log \Big(\sum_{j=1}^{m-1} e^{\beta  (t_{j}-t_{m-1})} \Big) - \log \Big( \frac{c_{p}}{c_{cs}} - \lambda_{0}\Big) \Big)
    \end{equation*}
    otherwise. \\To compute $\tau_{\text{opt}, \mathcal{H}_{m-1}}$, we  know that the length of the optimal keep-alive window is $\tau_{\text{opt}, \mathcal{H}}  = t_{\text{opt}, \mathcal{H}} - \; t_{m-1}$, where $t_{m-1}$ is the most recent arrival request. This expression is obtained by substituting the conditional intensity of the Hawkes process in Corollary \ref{optimal_decreasing} and solving for $t_{\text{opt}, \mathcal{H}_{m-1}}$.
\end{itemize}

\begin{align*}
\frac{c_{p}}{c_{cs}} & = \lambda_{0} + \sum_{t_{j} \in \mathcal{H}_{m-1}} \alpha \cdot e^{- \beta \cdot (t_{\text{opt}, \mathcal{H}_{m-1}} - t_{j})}\\
\frac{1}{\alpha}\Big(\frac{c_{p}}{c_{cs}} - \lambda_{0} \Big) & = \sum_{t_{j} \in \mathcal{H}_{m-1}} e^{- \beta \cdot (t_{\text{opt}, \mathcal{H}_{m-1}} - t_{j})}\\
\frac{1}{\alpha}\big(\frac{c_{p}}{c_{cs}} - \lambda_{0} \Big) & = e^{- \beta \cdot t_{\text{opt}, \mathcal{H}_{m-1}}} \cdot \sum_{t_{j} \in \mathcal{H}_{m-1}} e^{\beta \cdot t_{j}} \\
\log \left(\frac{1}{\alpha}\big(\frac{c_{p}}{c_{cs}} - \lambda_{0} \Big)\right) & = - \beta \cdot t_{\text{opt}, \mathcal{H}_{m-1}} + \log \left(\sum_{t_{j} \in \mathcal{H}_{m-1}} e^{\beta \cdot t_{j}}\right) \\
\beta \cdot t_{\text{opt}, \mathcal{H}_{m-1}} & = \log \alpha + \log \left(\sum_{t_{j} \in \mathcal{H}_{m-1}} e^{\beta \cdot t_{j}}\right) - \log \Big(\frac{c_{p}}{c_{cs}} - \lambda_{0}\Big)\\
\beta \cdot t_{\text{opt}, \mathcal{H}_{m-1}} & = \log \alpha + \log \left(\Big(e^{\beta \cdot t_{m-1}} \Big) \cdot \Big(\sum_{t_{j} \in \mathcal{H}_{m-1}} e^{\beta \cdot t_{j} - \beta \cdot t_{m-1}}\Big) \right) - \log \Big(\frac{c_{p}}{c_{cs}} - \lambda_{0}\Big)\\
\beta \cdot t_{\text{opt}, \mathcal{H}_{m-1}} & = \log \alpha + \beta \cdot t_{m-1} + \log \left(\sum_{t_{j} \in \mathcal{H}_{m-1}} e^{\beta \cdot (t_{j} - t_{m-1})}\right) - \log \Big(\frac{c_{p}}{c_{cs}} - \lambda_{0}\Big)\\
t_{\text{opt}, \mathcal{H}_{m-1}} & = t_{m-1} + \frac{1}{\beta} \left(\log \alpha +  \log \left(\sum_{t_{j} \in \mathcal{H}_{m-1}} e^{\beta \cdot (t_{j} - t_{m-1})}\right) - \log \Big(\frac{c_{p}}{c_{cs}} - \lambda_{0}\Big)\right)\\
\tau_{\text{opt}, \mathcal{H}_{m-1}}  & = \frac{1}{\beta} \cdot \left( \log \alpha + \log \big(\sum_{t_{j} \in \mathcal{H}_{m-1}} e^{\beta \cdot (t_{j}-t_{m-1})} \big) - \log \Big( \frac{c_{p}}{c_{cs}} - \lambda_{0}\Big) \right)
\end{align*}

\begin{corollary}
When the parameters of the Hawkes process are such that $c_{p} - (c_{cs} \cdot \lambda(x | \mathcal{H})) = 0$ has a solution, the optimal policy has a history independent lower bound,  and an upper bound expressed as follows
\begin{align*}
\tau_{\text{opt}, \mathcal{H}} & \geq \frac{1}{\beta} \cdot \left(\log \alpha - \log \Big(\frac{c_{p}}{c_{cs}} - \lambda_{0}\Big)\right) \\
\tau_{\text{opt}, \mathcal{H}} & \leq \frac{1}{\beta} \cdot \left(\log \alpha + \log \delta + 1 - \log \Big(\frac{c_{p}}{c_{cs}} - \lambda_{0}\Big)\right)
\end{align*}
\noindent where $\delta$ satisfies
\begin{equation*}
\sum_{i=m-\delta}^{m-1} e^{\beta \cdot (t_i - t_{m-1})} 
\geq \frac{1}{2} \sum_{i=1}^{m-1} e^{\beta \cdot (t_i - t_{m-1})} 
\end{equation*}
\end{corollary}

\begin{proof}
We can rewrite the formula for the optimal policy for a Hawkes process with a given history as:
 
\begin{align*}
\tau_{\text{opt}, \mathcal{H}} & = t_{\text{opt},\mathcal{H}} - t_{m} \\
& = \frac{1}{\beta} \cdot \left(\log \alpha + \log \left(\sum_{t_{j} \in \mathcal{H}} e^{\beta \cdot t_{j}}\right) - \log \Big(\frac{c_{p}}{c_{cs}} - \lambda_{0}\Big)\right) - t_{m} \\
& = \frac{1}{\beta} \cdot \left(\log \alpha + \log \left( e^{\beta \cdot t_{1}} + e^{\beta \cdot t_{2}} + \cdots + e^{\beta \cdot t_{m}}\right) - \log \Big(\frac{c_{p}}{c_{cs}} - \lambda_{0}\Big)\right) - t_{m}\\
& = \frac{1}{\beta} \cdot \left(\log \alpha + \log \left( e^{\beta \cdot t_{1}} + e^{\beta \cdot t_{2}} + \cdots + e^{\beta \cdot t_{m}}\right) - \log \Big(\frac{c_{p}}{c_{cs}} - \lambda_{0}\Big)\right) - t_{m}\\
& = \frac{1}{\beta} \cdot \left(\log \alpha + \log \left( (e^{\beta \cdot t_{m}}) \cdot ( e^{\beta \cdot (t_{1} - t_{m})} + e^{\beta \cdot (t_{2} - t_{m})} + \cdots + e^{\beta \cdot (t_{m} - t_{m})})\right) - \log \Big(\frac{c_{p}}{c_{cs}} - \lambda_{0}\Big)\right) - t_{m}\\
& = \frac{1}{\beta} \cdot \left(\log \alpha + \log \left( e^{\beta \cdot (t_{1} - t_{m})} + e^{\beta \cdot (t_{2} - t_{m})} + \cdots + e^{\beta \cdot (t_{m} - t_{m})}\right) - \log \Big(\frac{c_{p}}{c_{cs}} - \lambda_{0}\Big)\right)\\
& = \frac{1}{\beta} \cdot \left(\log \alpha + \log \left( e^{\beta \cdot (t_{1} - t_{m})} + e^{\beta \cdot (t_{2} - t_{m})} + \cdots + 1 \right) - \log \Big(\frac{c_{p}}{c_{cs}} - \lambda_{0}\Big)\right)\\
\end{align*}

This has three terms, two of which are independent of the history.  Thus we can obtain a lower bound on the optimal policy for {\em any} history as$\tau_{\text{opt}, \mathcal{H}} \geq \frac{1}{\beta} \cdot \left(\log \alpha - \log \Big(\frac{c_{p}}{c_{cs}} - \lambda_{0}\Big)\right)$.  In fact, this is the optimal policy for the empty history. 

For the term that depends on history, all exponents are negative so each term is at most 1.  This yields a trivial upper bound of $\tau_{\text{opt}, \mathcal{H}} \leq \frac{1}{\beta} \cdot \left(\log \alpha + \log m - \log \Big(\frac{c_{p}}{c_{cs}} - \lambda_{0}\Big)\right)$. 
While it grows slowly due to the log,
this bound is unappealing to apply directly because it grows with the length of the history.  In reality, many of the terms of the sum are close to 0 because $t_i - t_m$ is very negative for $t_i$ substantially in the past.

To get a better estimate, let $\delta$ be such that 
\begin{equation*}
\sum_{i=m-\delta+1}^m e^{\beta \cdot (t_i - t_m)} 
\geq \frac{1}{2} \sum_{i=1}^m e^{\beta \cdot (t_i - t_m)} 
\end{equation*}

That is, the most recent $\delta$ arrivals provide at least half the total weight.  This can be thought of as only having $\delta$ arrivals that are recent enough to matter.  Then we have the 
upper bound of $\tau_{\text{opt}, \mathcal{H}} \leq \frac{1}{\beta} \cdot \left(\log \alpha + \log \delta + 1 - \log \Big(\frac{c_{p}}{c_{cs}} - \lambda_{0}\Big)\right)$.  
\end{proof}

\section{Omitted Figures from Section 5.1}

\begin{figure*}[h!]
    \centering
    \begin{subfigure}[t]{0.32\textwidth}
        \centering
        \includegraphics[height=1.5in]{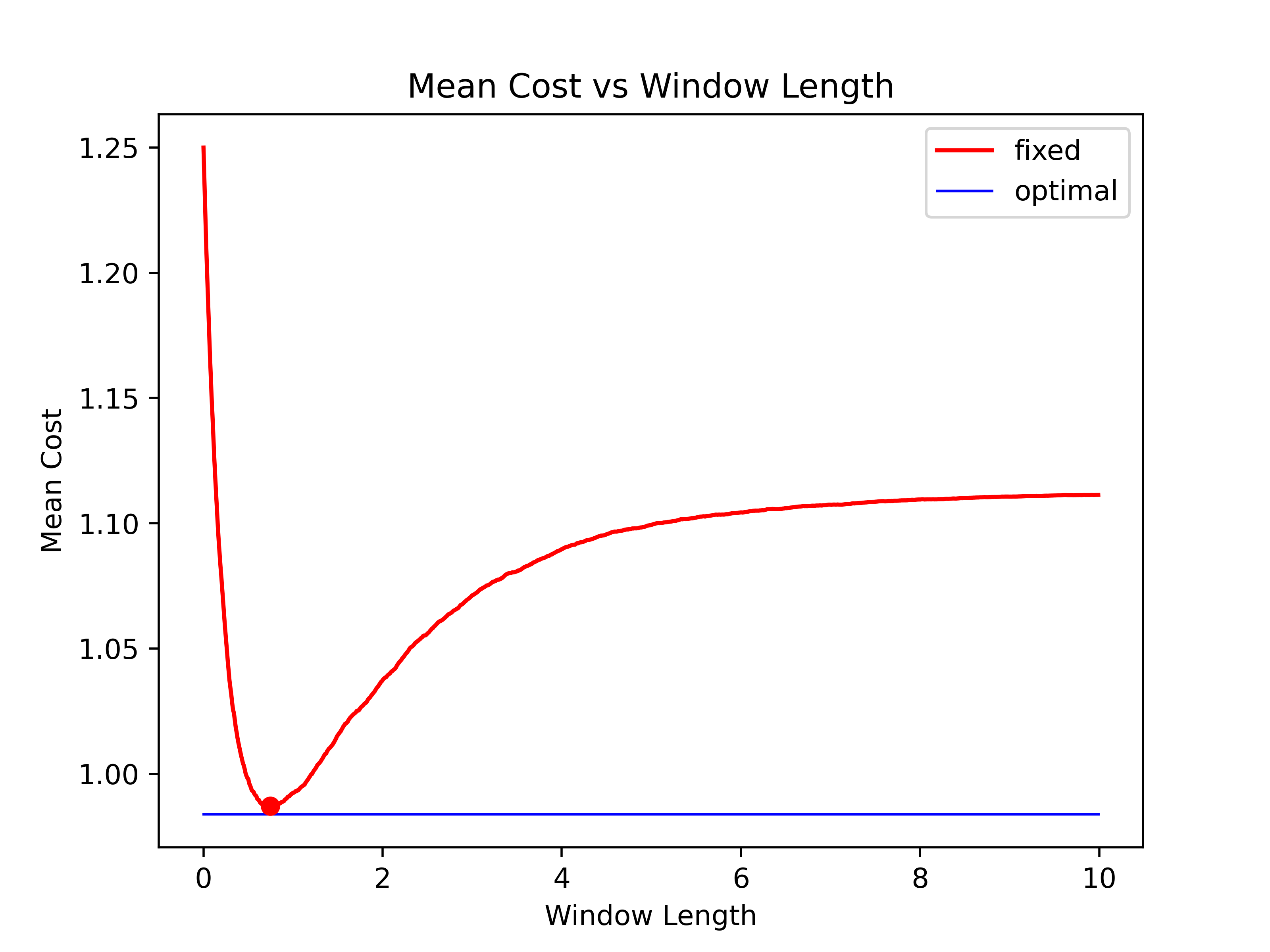}
        \caption{$\alpha = 0.8$}
    \end{subfigure}
    ~ 
    \begin{subfigure}[t]{0.32\textwidth}
        \centering
        \includegraphics[height=1.5in]{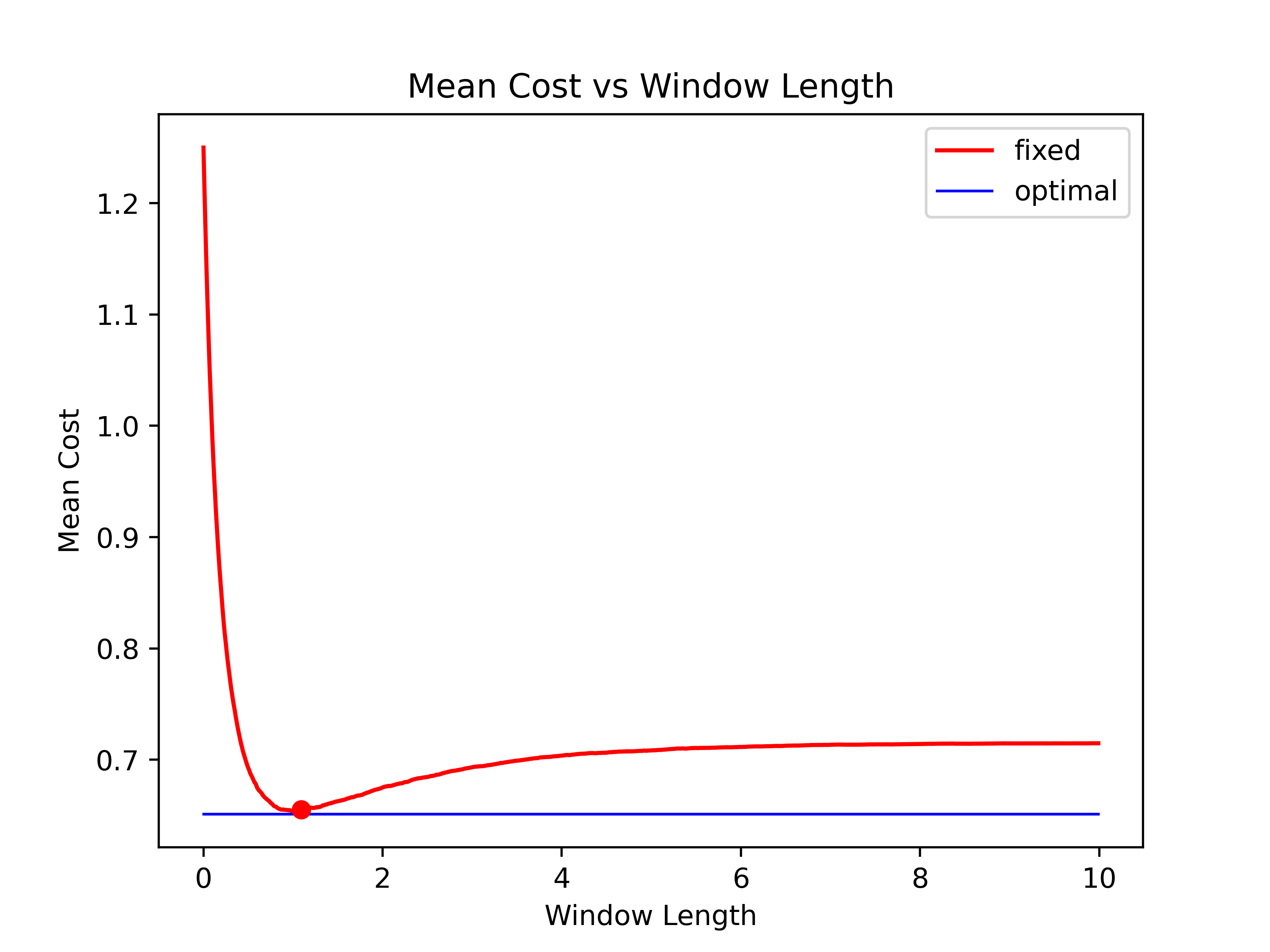}
        \caption{$\alpha = 1.4$}
    \end{subfigure}
    ~
    \begin{subfigure}[t]{0.32\textwidth}
        \centering
        \includegraphics[height=1.5in]{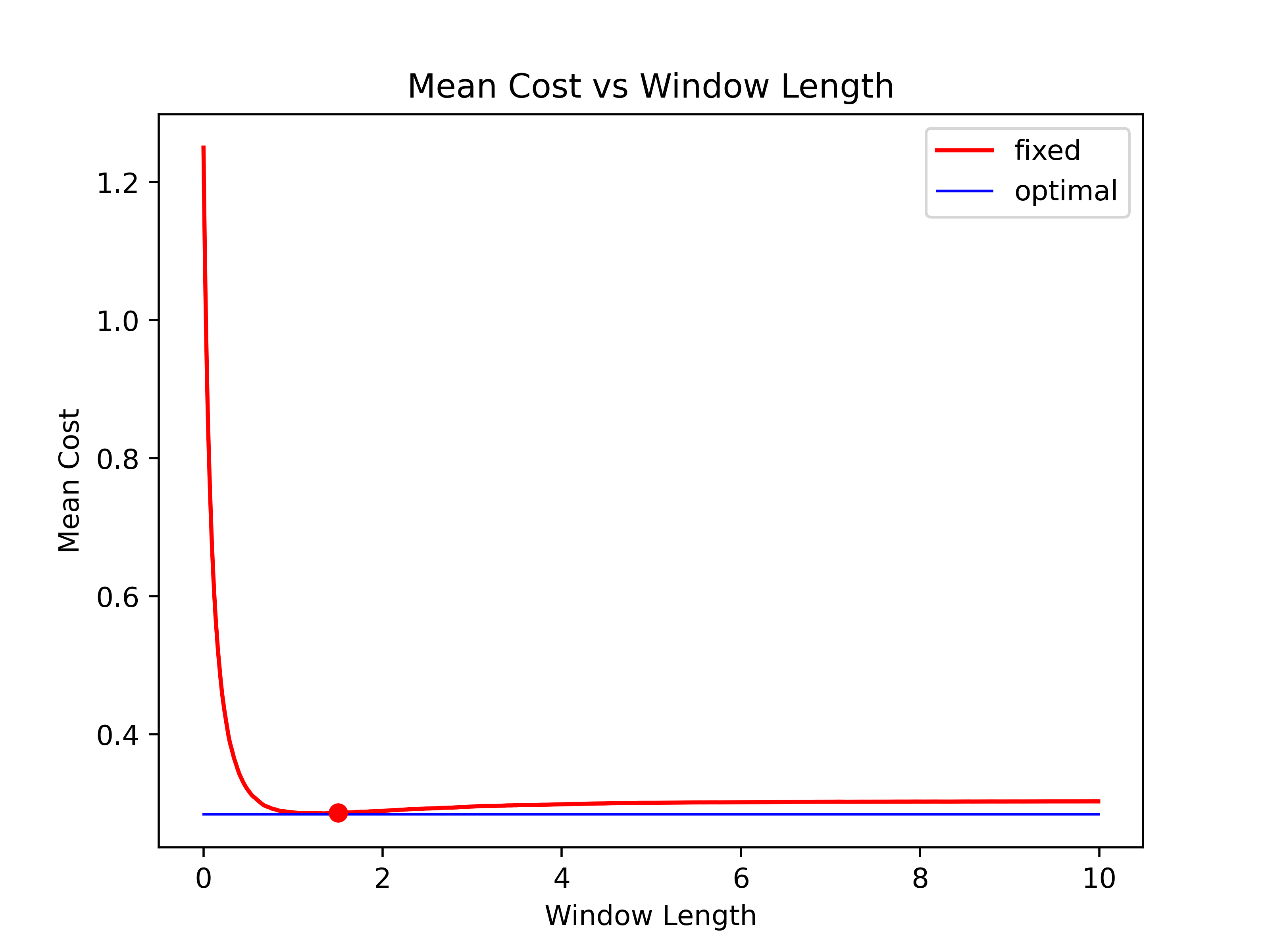}
        \caption{$\alpha = 2.0$}
    \end{subfigure}    
    \caption{Plots of average costs for policies when $\alpha$ is increased, given $\lambda_{0} = 0.6, \beta = 2.4, c_{p} = 1.0, c_{cs} = 1.25$}
    \label{cost_compare_alpha}
\end{figure*}

\begin{figure*}[h!]
    \centering
    \begin{subfigure}[t]{0.32\textwidth}
        \centering
        \includegraphics[height=1.5in]{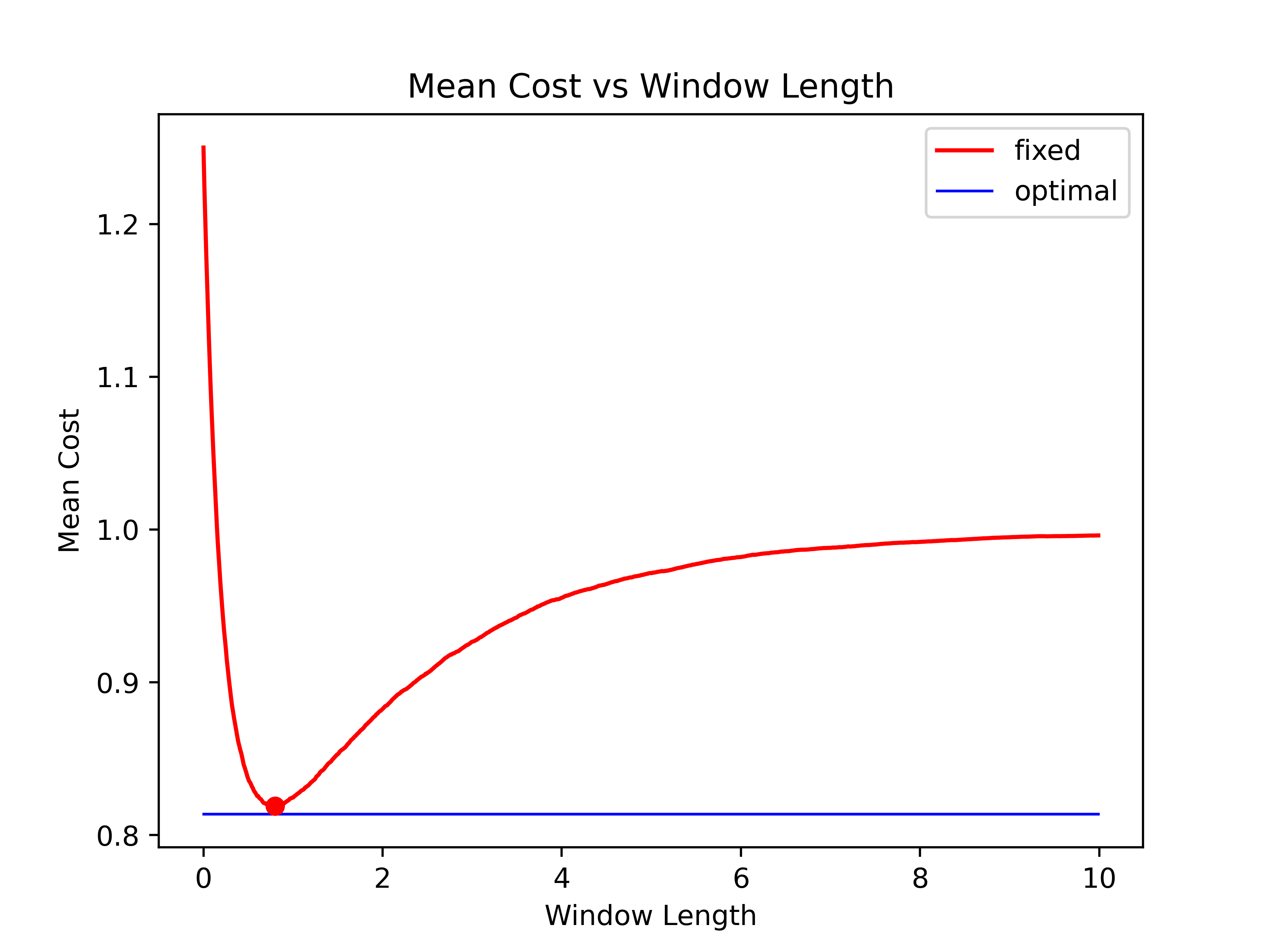}
        \caption{$\lambda_{0} = 0.5$}
    \end{subfigure}
    ~ 
    \begin{subfigure}[t]{0.32\textwidth}
        \centering
        \includegraphics[height=1.5in]{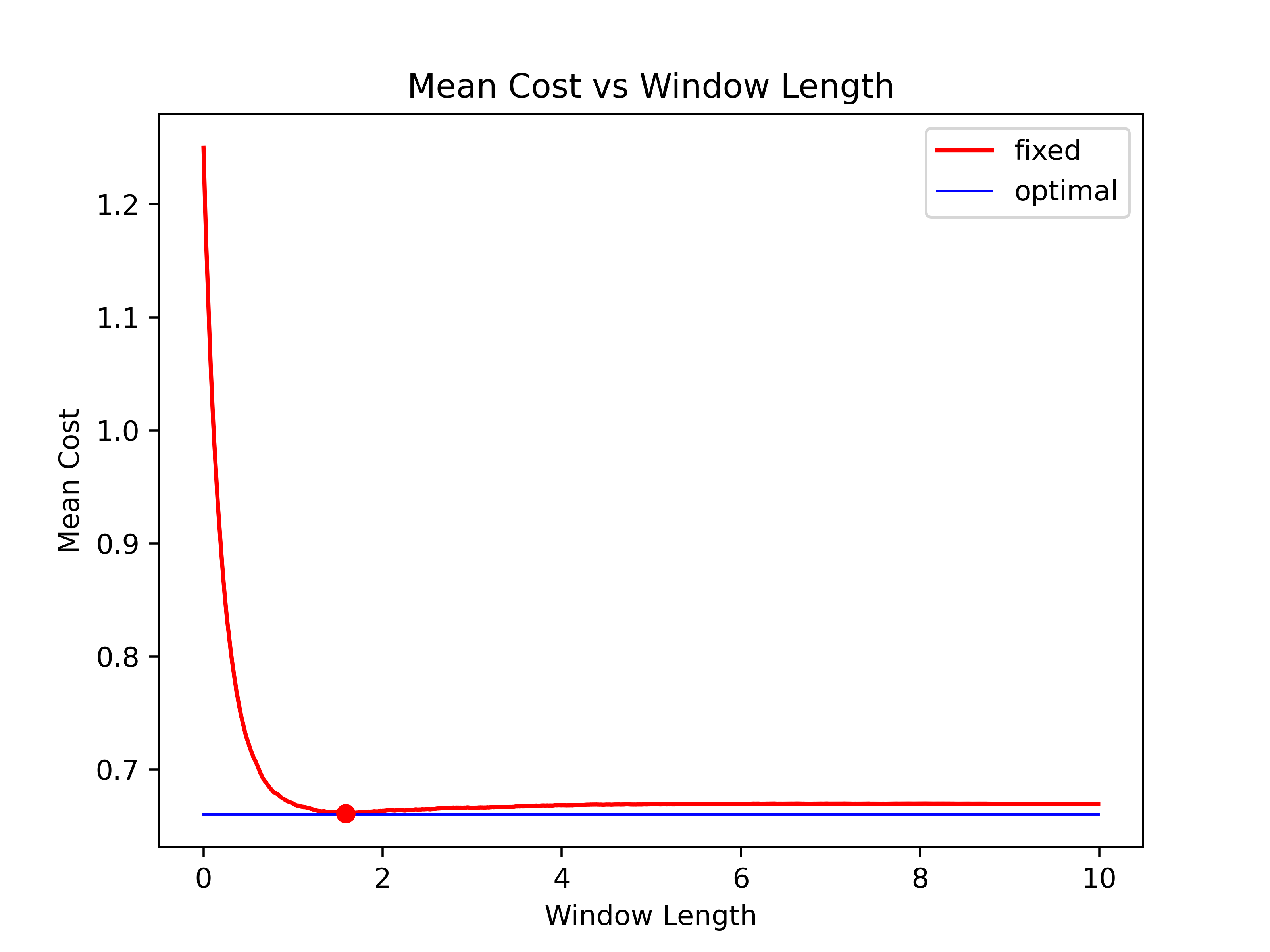}
        \caption{$\lambda_{0} = 0.75$}
    \end{subfigure}
    ~
    \begin{subfigure}[t]{0.32\textwidth}
        \centering
        \includegraphics[height=1.5in]{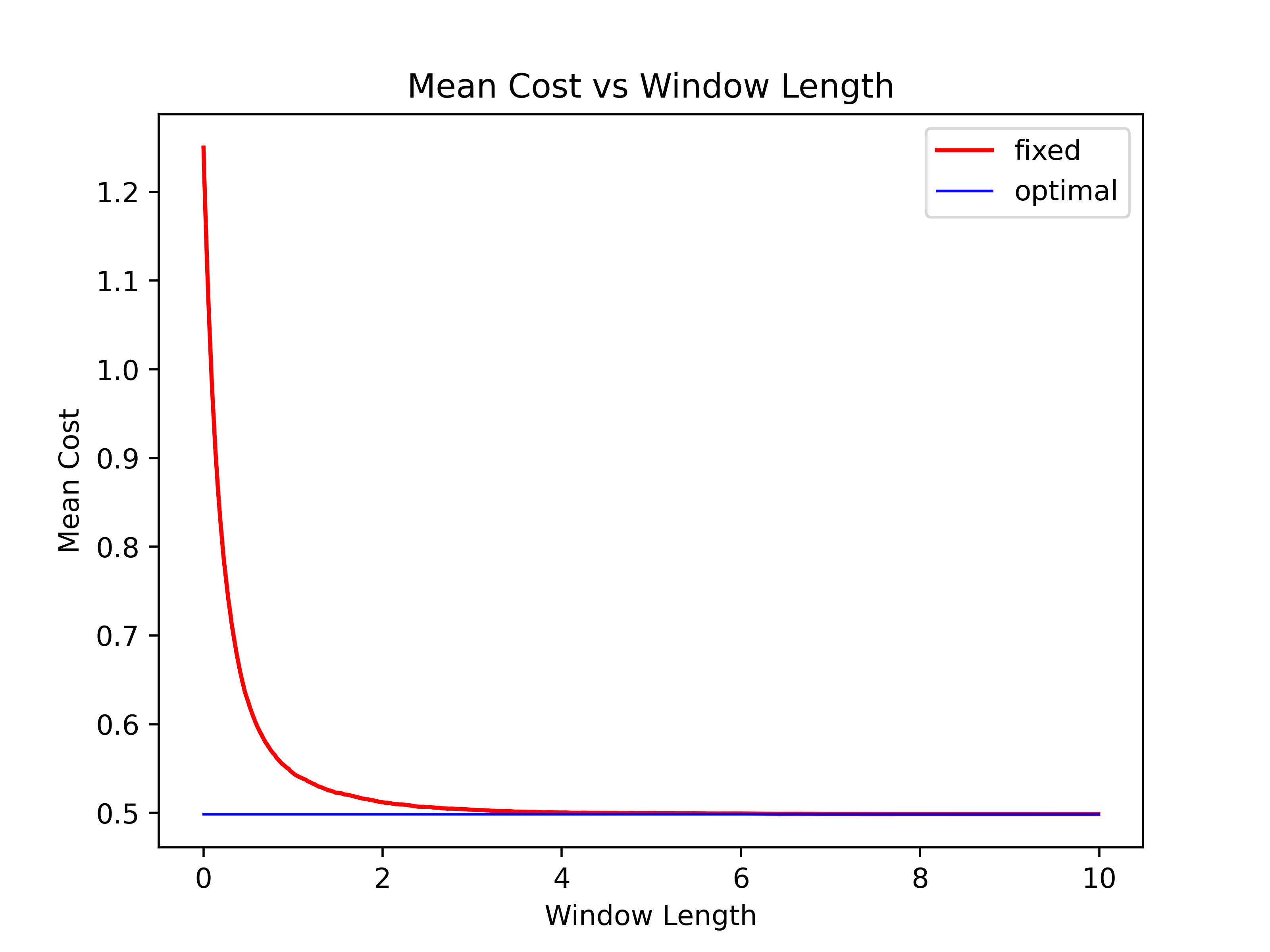}
        \caption{$\lambda_{0} = 1.0$}
    \end{subfigure}    
    \caption{Plots of average costs for policies when $\lambda_{0}$ is increased, given $\alpha = 1.2, \beta = 2.4, c_{p} = 1.0, c_{cs} = 1.25$}
    \label{cost_compare_lambda}
\end{figure*}

\begin{figure*}[h!]
    \centering
    \begin{subfigure}[t]{0.32\textwidth}
        \centering
        \includegraphics[height=1.5in]{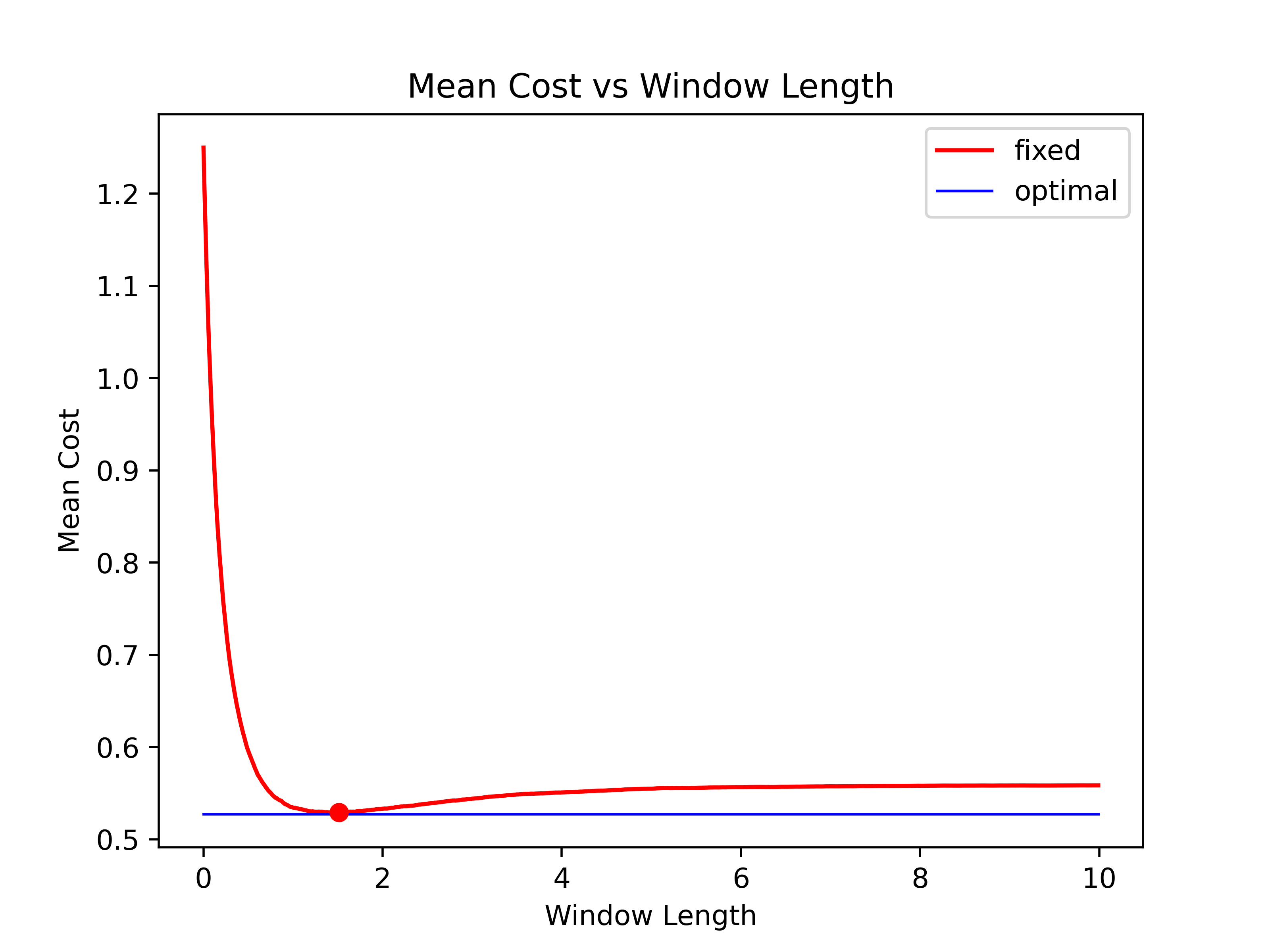}
        \caption{$\beta = 1.8$}
    \end{subfigure}
    ~ 
    \begin{subfigure}[t]{0.32\textwidth}
        \centering
        \includegraphics[height=1.5in]{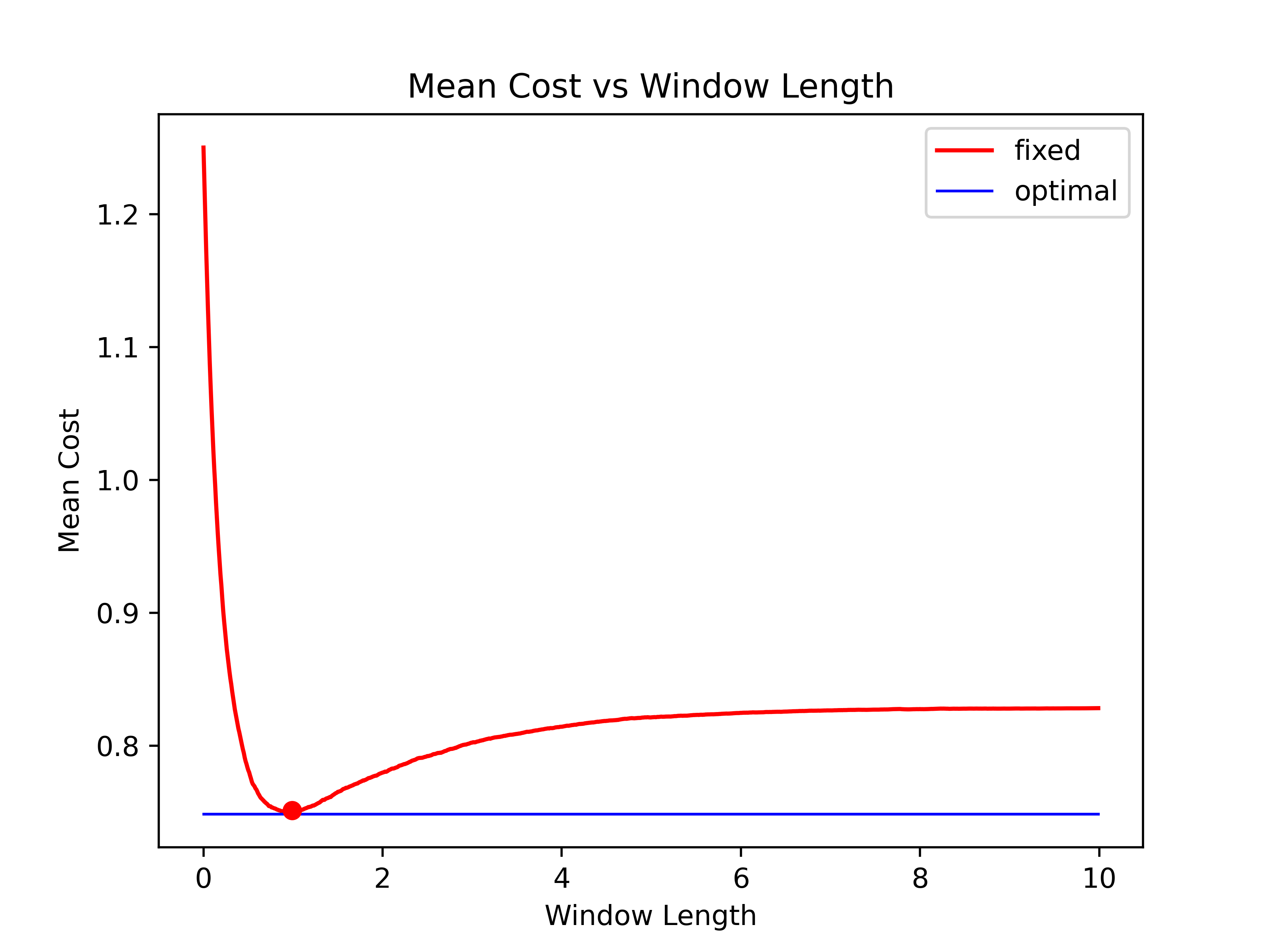}
        \caption{$\beta = 2.4$}
    \end{subfigure}
    ~
    \begin{subfigure}[t]{0.32\textwidth}
        \centering
        \includegraphics[height=1.5in]{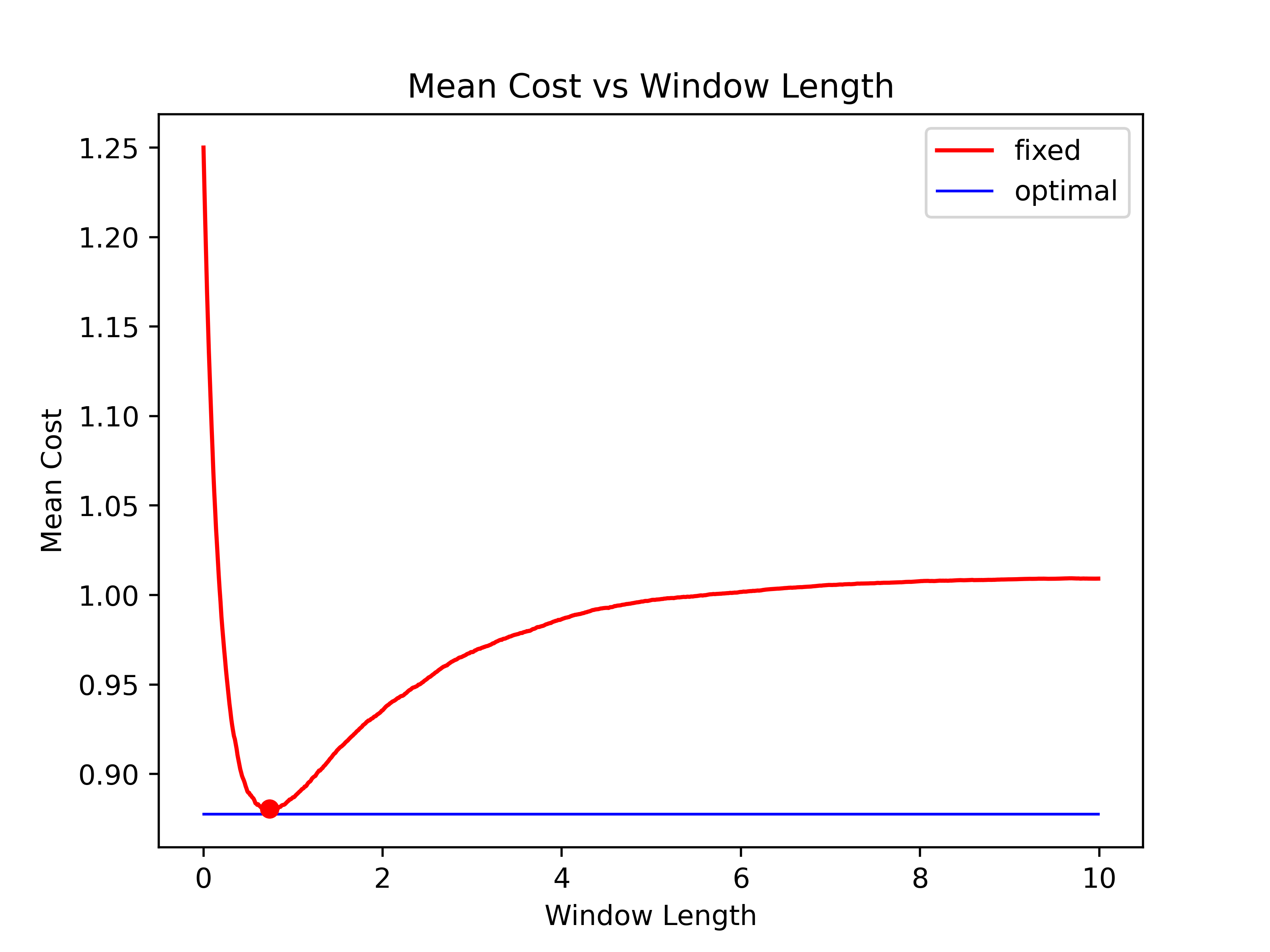}
        \caption{$\beta = 3.0$}
    \end{subfigure}    
    \caption{Plots of average costs for policies when $\beta$ is increased, given $\lambda_{0} = 0.6, \alpha = 1.2, c_{p} = 1.0, c_{cs} = 1.25$}
    \label{cost_compare_beta}
\end{figure*}

These additional figures demonstrate the robustness of the performance of Optimized-TTL with respect to a range of parameters. In them, we examine how the average costs behave with respect to the Hawkes process parameters $\lambda_{0}, \; \alpha, \text{and} \; \beta$ while holding the costs fixed. Figure \ref{cost_compare_alpha} shows the behavior of the average cost of the policies for different values of $\alpha$ of the Hawkes process. We see that as $\alpha$ increases, the average length of the optimal keep-alive window increases. This is because for higher values of $\alpha$, the intensity of the subsequent arrival will be larger making it larger keep-alive windows more desirable. This intuition can be made more precise with Corollary \ref{optimal_decreasing} since $g(x|\mathcal{H}_{m-1}) = (1 - F(x|\mathcal{H}_{m-1})) \cdot (c_{p} - c_{cs}\lambda(x|\mathcal{H}_{m-1}))$. Increasing $\alpha$, increases $\lambda(x|\mathcal{H}_{m-1})$ causing $g(x|\mathcal{H}_{m-1})$ to be more negative. Therefore, the point $\tau_{\text{opt}}$ where $g(x|\mathcal{H}_{m-1})$ changes sign from negative to positive is larger for a larger $\alpha$. The behavior of the average costs of the policies when $\lambda_{0}$ increases is similar to that of $\alpha$ as shown in Figure \ref{cost_compare_lambda}. From Figure \ref{cost_compare_beta} we see that as $\beta$ increases, the average length of the optimal keep-alive window decreases. The decay rate of the arrivals' influence is larger for a larger $\beta$ which makes shorter keep-alive windows more optimal. This connects to Corollary \ref{optimal_decreasing}, where a higher value of $\beta$ causes $g(x|\mathcal{H}_{m-1}) = (1 - F(x|\mathcal{H}_{m-1})) \cdot (c_{p} - c_{cs}\lambda(x|\mathcal{H}_{m-1}))$ to change from negative to positive earlier.

\section{Extension: worst-case guarantees for Hawkes processes}

We know from Theorem \ref{opt_policy}, that computing the optimal keep-alive policy requires the history $\mathcal{H}_{m-1}$ of previous $m-1$ invocations. As described in Section 4.2, the computational complexity of the optimal policy increases as the history of invocations increase. Hence, we propose history independent policies that do not require any information regarding past arrival requests. This problem is similar in spirit to the \textit{Ski Rental} problem in online algorithms, where the customer can buy an item for $\$$ B or rent the item for $\$$ R per the unit of time. There, a 2-approximation results from renting until the cost of buying has been paid in rental fees as if the input ends during the rental period the policy was optimal and otherwise buying immediately would have been optimal so the policy overpaid by a factor of 2.  Similarly, in our setting a fixed keep-alive policy can achieve a 2-approximation (Theorem \ref{fixed_policy}). This bound does not use any information about the parameters of the Hawkes process. In Theorem \ref{approx_policy}, we propose a history independent approximate policy that requires only the parameters of the Hawkes process (i.e. is independent of the history), and approximates the optimal cost by a factor of $ \displaystyle 1 + \left( \frac{1}{ \frac{c_{p}}{c_{cs}} \cdot \tau_{\text{opt}, \mathcal{H} = \phi} + 1} \right)^{\frac{1}{2}} $. Both results follow from the following lemma, which bounds the performance of arbitrary history independent keep-alive policies.

\begin{lemma} \label{simple_policy}
A policy with keep-alive window $\tau$ which does not depend on the history of arrivals of invocations is at least \quad $\displaystyle \max \Big\{\frac{c_{p} \cdot \tau}{c_{p} \cdot \tau_{\text{opt}, \mathcal{H} = \phi} +c_{cs}} + 1, 1 + \frac{c_{cs}}{c_{p} \cdot \tau} \Big\} \;$ approximation to the cost of the optimal policy $\tau_{\text{opt}, \mathcal{H}_{m-1}}$ for any history $\mathcal{H}_{m-1}$. 
\end{lemma}

\begin{proof}
Given history of application invocations $\mathcal{H}_{m-1}$, we denote the length of the optimal keep-alive window  by $\tau_{\text{opt}, \mathcal{H}_{m-1}}$. Let $\tau$ denote the length of a history independent policy. We examine the upper bound of the ratio of the cost of the history independent policy to the cost of the optimal policy when the application is invoked at the $m$-th inter-arrival $x_{m}$, that is, $\displaystyle \frac{cost(x_{m}, \tau)}{cost(x_{m}, \tau_{\text{opt}, \mathcal{H}_{m-1}})}$, where $x_{m} = t_{m} - t_{m-1}$ is the length of the $m$-th inter-arrival. There are three possibilities when comparing keep-alive window $\tau$ with $\tau_{\text{opt}, \mathcal{H}_{m-1}}$. They are,
\begin{enumerate}
    \item $\tau \; = \; \tau_{\text{opt}, \mathcal{H}_{m-1}}$
    \item $\tau \;  < \; \tau_{\text{opt}, \mathcal{H}_{m-1}}$
    \item $\tau \;  > \; \tau_{\text{opt}, \mathcal{H}_{m-1}}$
\end{enumerate}

We examine the upper bound of the ratio of the cost of the history independent policy to the cost of the optimal policy for each case listed above.

\noindent \textbf{Case 1}: When $\tau = \; \tau_{\text{opt}, \mathcal{H}_{m-1}}$, both the history independent policy and the optimal policy have the same cost. That is , 
\begin{equation*}
    cost(x_{m}, \tau) \; = \; cost(x_{m}, \tau_{\text{opt}, \mathcal{H}_{m-1}})
\end{equation*}

\noindent \textbf{Case 2}: When $\tau \; < \; \tau_{\text{opt}, \mathcal{H}_{m-1}}$, the cost of the policies can be compared based on when the application invocation occurs.

\begin{itemize}
    \item When the application invocation occurs before the end of the history independent keep-alive window, that is, $x_{m} \leq \tau$,  then both policies encounter a warm start. Hence, both policies have the same cost. That is, 
    \begin{equation*}
    cost(x_{m}, \tau) \; = \; cost(x_{m}, \tau_{\text{opt}, \mathcal{H}_{m-1}}) = c_{p} \cdot x_{m}
    \end{equation*}
    \item When the application invocation is after the keep-alive window $\tau$, but before the end of the optimal policy, that is, $\tau < \; x_{m} \; \leq \; \tau_{\text{opt}, \mathcal{H}_{m-1}}$, then the history independent policy encounters a cold start whereas the optimal policy experiences a warm start. The ratio of the cost of the history independent policy to the cost of the optimal policy is expressed as follows, 
    \begin{align*}
    \frac{cost(x_{m}, \tau)}{cost(x_{m}, \tau_{\text{opt}, \mathcal{H}_{m-1}})} & = \; \frac{c_{p} \cdot \tau + c_{cs}}{c_{p} \cdot x_{m}} \\
    & \leq \; \frac{c_{p} \cdot \tau + c_{cs}}{c_{p} \cdot \tau} \quad \qquad \text{(6)}\\
    & = \; 1 + \frac{c_{cs}}{c_{p} \cdot \tau}
    \end{align*}
    In Equation (6) above we see that the minimum possible cost of the optimal policy in this scenario is when the application gets invoked just after the fixed keep-alive window, that is, when $x_{m}= \tau$.
    \item When the application invocation is after the optimal keep-alive policy, that is, $\tau \; < \; \tau_{\text{opt}, \mathcal{H}_{m-1}} \; <\; x_{m}$, then both policies encounter a cold start. Here, the cost of the optimal policy is larger than the cost of the history independent keep-alive policy because the cost of a policy when a cold start occurs is proportional to the length of the keep-alive window. That is,
    \begin{align*}
    \tau \; & \leq \; \tau_{\text{opt}, \mathcal{H}_{m-1}} \\
    \implies c_{p} \cdot \tau  + c_{cs} \; & \leq \; c_{p} \cdot \tau_{\text{opt}, \mathcal{H}_{m-1}} + c_{cs} \\
    \implies cost(x_{m}, \tau) \; & \; \leq \; cost(x_{m}, \tau_{\text{opt}, \mathcal{H}_{m-1}})
    \end{align*}
\end{itemize}

\noindent \textbf{Case 3}: When $\tau \; > \; \tau_{\text{opt}, \mathcal{H}_{m-1}}$, the cost of the policies can be compared based on the arrival of application invocations.
\begin{itemize}
    \item When the application is invoked before the end of the optimal keep-alive window, that is, $x_{m} \; \leq \; \tau_{\text{opt},\mathcal{H}_{m-1}}$, then both the policies encounter a warm start. Hence, both the policies have the same costs. That is,
    \begin{equation*}
    cost(x_{m}, \tau) = \; cost(x_{m}, \tau_{\text{opt}, \mathcal{H}_{m-1}}) = \; c_{p} \cdot x_{m} 
    \end{equation*}
    \item When the application invocation occurs after the optimal keep-alive window, that is, $x_{m} \; > \; \tau_{\text{opt}, \mathcal{H}_{m-1}}$, then the optimal policy experiences a cold start. The ratio of the cost of the history independent policy to the cost of the optimal policy is upper bounded when the history independent policy has a cold start. We compute the upper bound on the ratio of costs as follows,
    \begin{align*}
    \frac{cost(x_{m}, \tau)}{cost(x_{m}, \tau_{\text{opt},\mathcal{H}_{m-1}})} \;& \leq \; \frac{c_{p} \cdot \tau + c_{cs}}{c_{p} \cdot \tau_{\text{opt}, \mathcal{H}_{m-1}}+ c_{cs}} \\
    \; & \leq \; \frac{c_{p} \cdot \tau}{c_{p} \cdot \tau_{\text{opt}, \mathcal{H} = \phi} + c_{cs}} + 1  \quad \qquad \text{(7)}
    \end{align*}
    where in Equation (7) we have substituted $\mathcal{H} = \phi$ to compute the upper bound.
\end{itemize}

We have now established two separate upper bounds for cases 2 and 3. The approximation factor of the history independent policy with respect to the optimal policy for an arbitrary history is the maximum of the two upper bounds, that is,\\ $\displaystyle \max \Big\{\frac{c_{p} \cdot \tau}{ c_{p} \cdot \tau_{\text{opt}, \mathcal{H} = \phi} + c_{cs}} + 1, 1 + \frac{c_{cs}}{c_{p} \cdot \tau} \Big\}$.

\end{proof}

\subsection{Fixed Keep-Alive Policy}

We first show our Ski rental style result.

\begin{theorem}
The cost of the fixed keep-alive policy $\tau_{\text{fixed}}=c_{cs}/c_{p}$ is at most twice the cost of the optimal keep-alive policy $\tau_{\text{opt}, \mathcal{H}_{m-1}}$. That is, when a function is invoked at time $t_{m}$ after previous $m-1$ arrivals we have, $cost(x_{m}, \tau_{\text{fixed}}) \; \leq \;  2 \cdot cost(x_{m}, \tau_{\text{opt}, \mathcal{H}_{m-1}})$, where $x_{m} = t_{m} - t_{m-1}$ is the length of the $m$-th inter-arrival, and $c_{p} \cdot \tau_{\text{fixed}} = c_{cs}$.
\label{fixed_policy}
\end{theorem} 

While this has a simple direct proof, we illustrate how it follows from Lemma~\ref{simple_policy}.

\begin{proof}

From Lemma \ref{simple_policy}, we know that the approximation factor of the fixed policy with respect to the optimal policy is $\displaystyle \max \Big\{\frac{c_{p} \cdot \tau_{\text{fixed}}}{c_{p} \cdot \tau_{\text{opt}, \mathcal{H} = \phi}+ c_{cs}} + 1, 1 + \frac{c_{cs}}{c_{p} \cdot \tau_{\text{fixed}}} \Big\}$. The upper bound  of $\displaystyle \frac{c_{p} \cdot \tau_{\text{fixed}}}{ c_{p} \cdot \tau_{\text{opt}, \mathcal{H} = \phi}+ c_{cs}} + 1$ can further be reduced to $\displaystyle \frac{c_{p} \cdot \tau_{\text{fixed}}}{ c_{p} \cdot \tau_{\text{opt}, \mathcal{H} = \phi}+ c_{cs}} + 1 \; \leq \; \frac{c_{p} \cdot \tau_{\text{fixed}}}{c_{cs}} + 1$ by substituting $\tau_{\text{opt}, \mathcal{H} = \phi} = 0$ because a fixed policy should accommodate for any history independent policy.

The best length of the keep-alive window for the fixed policy is the length which minimizes the maximum of the two upper bounds on the ratio of the cost of the fixed policy and the optimal policy. Mathematically, the best length of the fixed policy is expressed as 
\begin{equation*}
    \arg \min_{\tau_{\text{fixed}}} \max \Big\{\frac{c_{p} \cdot \tau_{\text{fixed}}}{c_{cs}} + 1, 1 + \frac{c_{cs}}{c_{p} \cdot \tau_{\text{fixed}}} \Big\}
\end{equation*}

We obtain the length of the fixed policy by solving the above expression,

\begin{align*}
    \frac{c_{p} \cdot \tau_{\text{fixed}}}{c_{cs}} + 1 \; & = \; 1 + \frac{c_{cs}}{c_{p}\cdot \tau_{\text{fixed}}}\\
    \frac{c_{p} \cdot \tau_{\text{fixed}}}{c_{cs}} \; & = \;  \frac{c_{cs}}{c_{p}\cdot \tau_{\text{fixed}}} \\
    \big(c_{p} \cdot \tau_{\text{fixed}}\big)^{2} \; & = \; (c_{cs})^{2}\\
    c_{p} \cdot \tau_{\text{fixed}} \; & = \;  c_{cs} 
\end{align*}

Substituting this back to the upper bound on the ratio of the cost of the fixed policy to the cost of the optimal policy, we get

\begin{align*}
    \frac{cost(x_{m}, \tau_{\text{fixed}})}{cost(x_{m}, \tau_{\text{opt}, \mathcal{H}_{m-1}})} \; & \leq \; 1 + \frac{c_{p} \cdot \tau_{\text{fixed}}}{c_{cs}}\\
    & = \; 1  + \frac{c_{cs}}{c_{cs}} \\
    & = \; 2
\end{align*}
\end{proof}

\subsection{History Independent Keep-Alive Policies}

More generally, we can take advantage of Lemma~\ref{simple_policy} to achieve a tighter bound that makes use of the parameters of the Hawkes process only through the policy they induce given the empty history.

\begin{theorem} \label{approx_policy}
There exists a policy with keep-alive window $\tau_{\text{approx}}$ which does not require the history of arrivals of application invocations with its cost upper bounded by  a factor of $ \; 1 + \left(\frac{1}{ \frac{c_{p}}{c_{cs}} \cdot \tau_{\text{opt}, \mathcal{H} = \phi} + 1} \right)^{\frac{1}{2}}$ with respect to the cost of the optimal keep-alive policy $\tau_{\text{opt}, \mathcal{H}_{m-1}}$. In other words, for a given history of invocations $\mathcal{H}_{m-1}$, when the application invocation has an inter-arrival of length $x_{m}$,  $ \displaystyle \frac{cost(x_{m}, \tau_{\text{approx}})}{cost(x_{m}, \tau_{\text{opt}, \mathcal{H}_{m-1}})} \; \leq 1 + \left( \frac{1}{ \frac{c_{p}}{c_{cs}} \cdot \tau_{\text{opt}, \mathcal{H}_{m-1} = \phi} + 1} \right)^{\frac{1}{2}} \leq 2$.
\end{theorem}

\begin{proof}

From Lemma \ref{simple_policy}, we know that the approximation factor of the approximate policy with respect to the optimal policy is $\displaystyle \max \Big\{\frac{c_{p} \cdot \tau_{\text{approx}}}{c_{p} \cdot \tau_{\text{opt}, \mathcal{H} = \phi}+ c_{cs}} + 1, 1 + \frac{c_{cs}}{c_{p} \cdot \tau_{\text{approx}}} \Big\}$. The best length of the keep-alive window of the approximate policy would minimize the maximum of the upper bounds of the ratio of the costs of the approximate policy and the optimal policy. Mathematically, the best approximate policy keep-alive window is expressed as,
\begin{equation*}
    \arg \min_{\tau_{\text{approx}}} \max \Big\{\frac{c_{p} \cdot \tau_{\text{approx}}}{c_{p} \cdot \tau_{\text{opt}, \mathcal{H} = \phi} +c_{cs}} + 1, 1 + \frac{c_{cs}}{c_{p} \cdot \tau_{\text{approx}}} \Big\}
\end{equation*}

\noindent We can obtain the length of the approximate keep-alive window by solving the above expression. 
\begin{align*}
    \frac{c_{p} \cdot \tau_{\text{approx}}}{c_{p} \cdot \tau_{\text{opt}, \mathcal{H} = \phi} + c_{cs}} + 1  & =  1 + \frac{c_{cs}}{c_{p} \cdot \tau_{\text{approx}}} \\
    \frac{c_{p} \cdot \tau_{\text{approx}}}{c_{p} \cdot \tau_{\text{opt}, \mathcal{H} = \phi} + c_{cs}}   & =  \frac{c_{cs}}{c_{p} \cdot \tau_{\text{approx}}} \\
    (c_{p} \cdot \tau_{\text{approx}})^{2} & = c_{cs} \cdot (c_{p} \cdot \tau_{\text{opt}, \mathcal{H} = \phi} + c_{cs}) \\
    \tau_{\text{approx}} & = \left( \frac{c_{cs}}{c_{p}} \cdot \Big( \tau_{\text{opt}, \mathcal{H} = \phi} + \frac{c_{cs}}{c_{p}} \Big) \right)^{\frac{1}{2}} \\
\end{align*}

\noindent Substituting $\tau_{\text{approx}}$ in the expression for the upper bound of the ratio of the cost of the approximate policy to the cost of the optimal policy, we get

\begin{align*}
    \frac{cost(x_{m}, \tau_{\text{approx}})}{cost(x_{m}, \tau_{\text{opt}, \mathcal{H}_{m-1}})} & = \;  1 + \frac{c_{p} \cdot \tau_{\text{approx}}}{c_{p} \cdot \tau_{\text{opt}, \mathcal{H} = \phi} + c_{cs}} \\
    & = \;  1 + \frac{\tau_{\text{approx}}}{ \tau_{\text{opt}, \mathcal{H} = \phi} + \displaystyle \frac{c_{cs}}{c_{p}}} \\
    & = \;  1 + \frac{\displaystyle \left( \frac{c_{cs}}{c_{p}} \cdot \Big( \tau_{\text{opt}, \mathcal{H} = \phi} + \frac{c_{cs}}{c_{p}} \Big) \right)^{\frac{1}{2}}}{ \tau_{\text{opt}, \mathcal{H} = \phi} + \displaystyle \frac{c_{cs}}{c_{p}}}\\[6 pt]
    & = \; 1 + \left( \frac{c_{cs}}{c_{p} \cdot \tau_{\text{opt}, \mathcal{H} = \phi} + c_{cs}} \right)^{\frac{1}{2}} \\
    & = \; 1 + \left( \frac{1}{ \displaystyle \frac{c_{p}}{c_{cs}} \cdot \tau_{\text{opt}, \mathcal{H} = \phi} + 1} \right)^{\frac{1}{2}} \\
    & \leq \; 1 + 1 = 2 \\
\end{align*}

\end{proof}

\subsection{Application to Poisson and Hawkes Processes}
As previously observed, the fixed policy from Theorem~\ref{fixed_policy} is independent of the process and so has a keep alive window of $\tau_{\text{fixed}}=c_{cs}/c_{p}$ for both Poisson and Hawkes processes. The behavior of $\tau_{\text{approx}}$
from Theorem~\ref{approx_policy} is more interesting.  For Poisson processes, we know that $\tau_{\text{opt}, \mathcal{H} = \phi}$ is either 0 or $\infty$.  In the former case, $\tau_{\text{approx}} = \tau_{\text{fixed}}$ and the approximation ratio of 2 is tight.  (Consider any input where $x_{m} > c_{p} / c_{cs}$.)  In the latter case however, $\tau_{\text{approx}} = \infty =  \tau_{\text{opt}}$ and so the approximation is 1.

For Hawkes processes, the length of the keep-alive window for the approximate policy is, 
\begin{align*}
    \tau_{\text{approx}}  = \left( \frac{c_{cs}}{c_{p}} \cdot \Big( \tau_{\text{opt}, \mathcal{H} = \phi} + \frac{c_{cs}}{c_{p}} \Big) \right)^{\frac{1}{2}} 
\end{align*}

From the more general expression for $\tau_{\text{opt}, \mathcal{H}}$ for Hawkes processes,
\begin{align*}
    \tau_{\text{opt}, \mathcal{H} = \phi} = \frac{1}{\beta} \cdot \Big( \log \alpha - \log \big( \frac{c_{p}}{c_{cs}} - \lambda_{0}\big) \Big)
\end{align*}

Combining these, we have 
\begin{align*}
    \tau_{\text{approx}}  = \left( \frac{c_{cs}}{c_{p}} \cdot \Bigg( \frac{1}{\beta} \cdot \Big( \log \alpha - \log \big( \frac{c_{p}}{c_{cs}} - \lambda_{0}\big) \Big) + \frac{c_{cs}}{c_{p}} \Bigg) \right)^{\frac{1}{2}} 
\end{align*}
This illustrates how $\tau_{\text{opt}, \mathcal{H} = \phi}$ implicitly brings the parameters of the Hawkes process into $\tau_{\text{approx}}$.

\subsection{Performances on Simulated Hawkes Processes}

We present similar simulations from before with two additional policies included (approximate policy and fixed policy of length $c_{cs}$).   Generally, they demonstrate the conservative approach these policies take to achieve their worst case guarantees.  In general, both perform worse than both the optimal policy (blue) and best fixed policy (red).  The relative performance of the two policies is not consistent, with each better in some cases.  The gap between the yellow line and the blue line is what provided room for the improvement of Optimal-TTL policy over Fixed policy.


\begin{figure*}[h!]
    \centering
    \begin{subfigure}[t]{0.45\textwidth}
        \centering
        \includegraphics[height=2.0in]{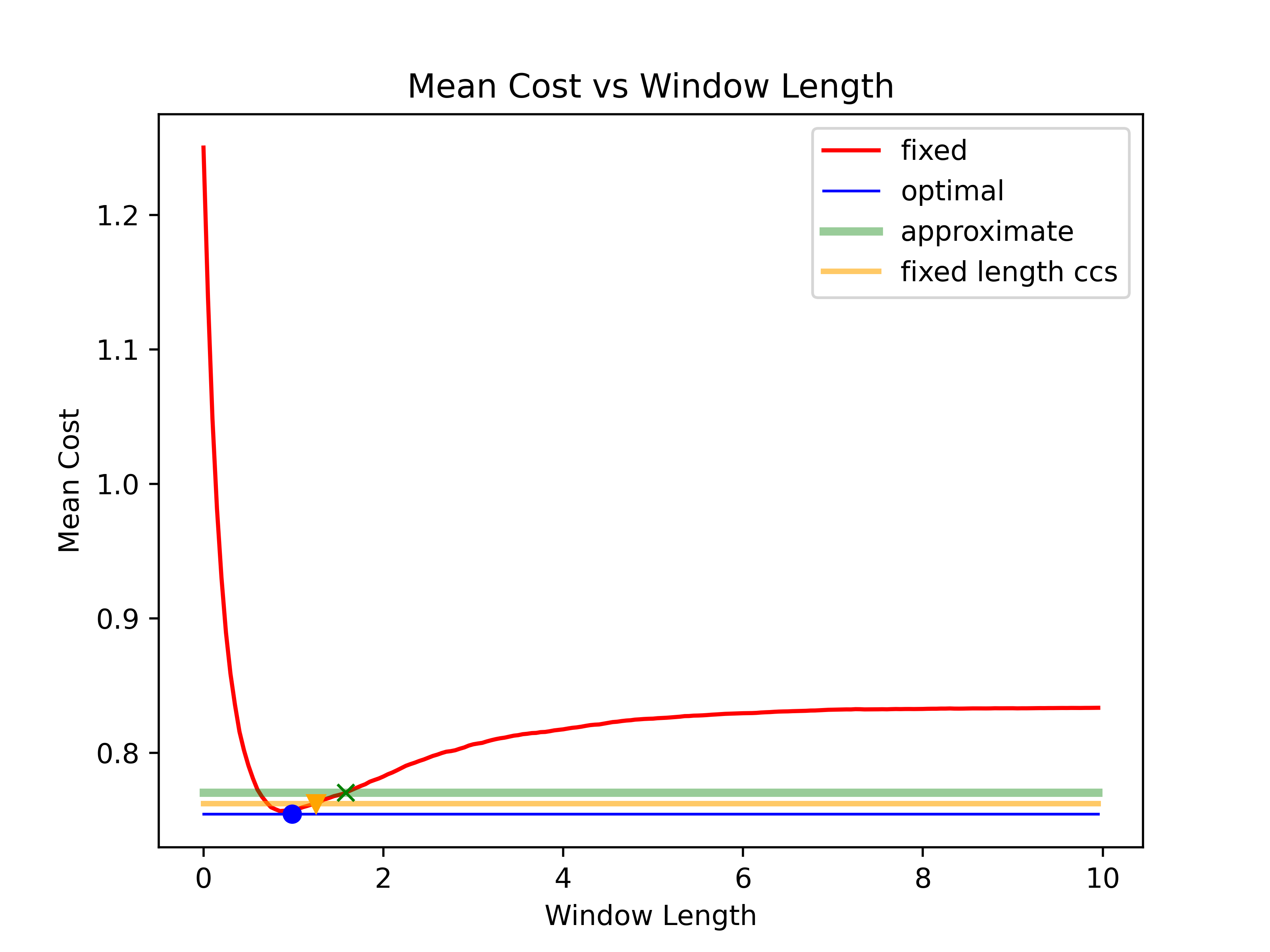}
        \caption{$\lambda_{0} = 0.6, \alpha = 1.2, \beta = 2.4, c_{p} = 1.0, c_{cs} = 1.25$}
    \end{subfigure}
    ~ 
    \begin{subfigure}[t]{0.5\textwidth}
        \centering
        \includegraphics[height=2.0in]{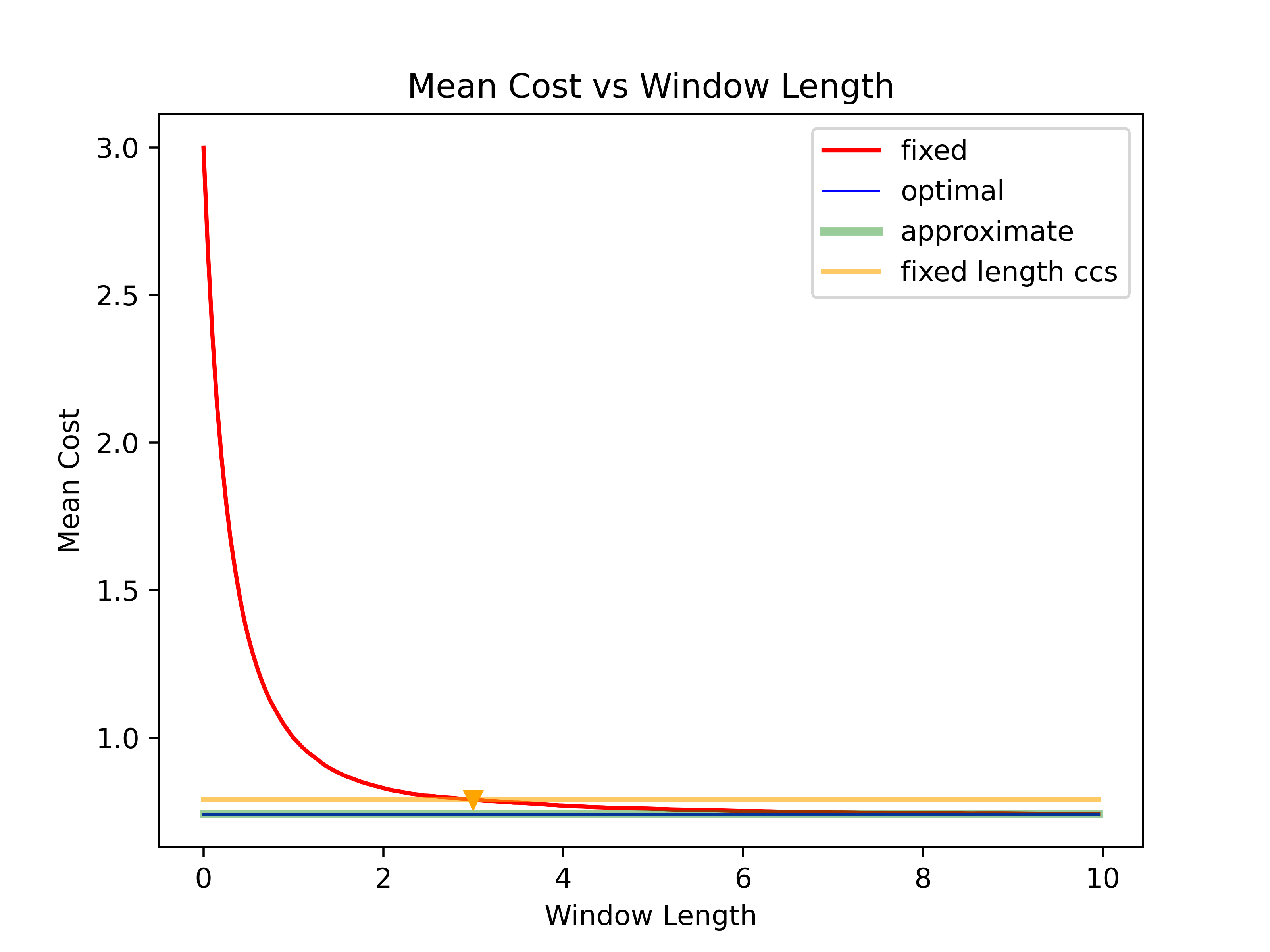}
        \caption{$\lambda_{0} = 0.45, \alpha = 0.8, \beta = 1.2, c_{p} = 1.0, c_{cs} = 3.0$}
    \end{subfigure}
    \caption{Plots of average cost comparisons between different policies for cases where $c_{p} \; \leq c_{cs}$}
    \label{cost_compare_sim1b}
\end{figure*}

\begin{figure*}[h!]
    \centering
    \begin{subfigure}[t]{0.45\textwidth}
        \centering
        \includegraphics[height=2.0in]{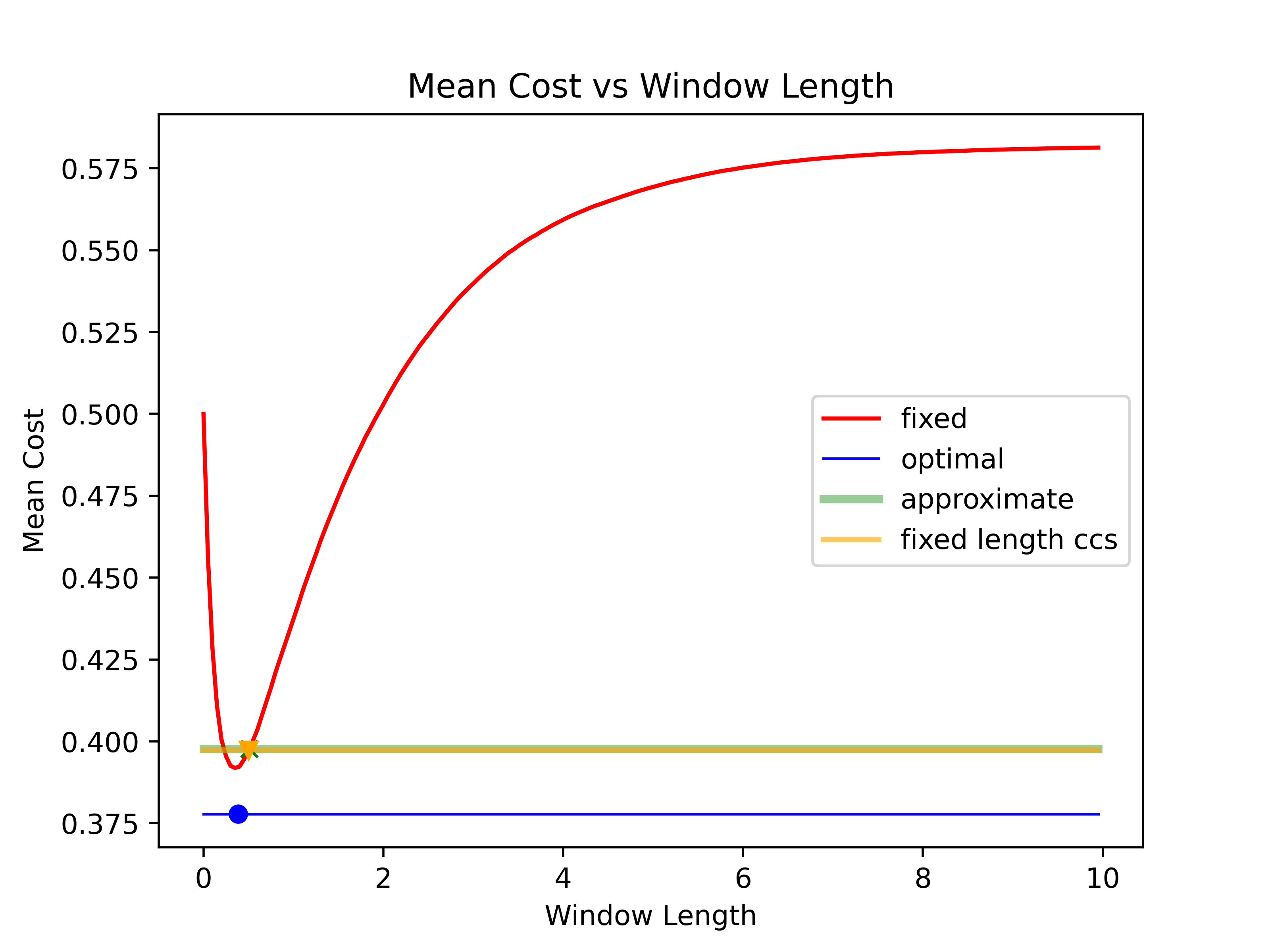}
        \caption{$\lambda_{0} = 0.65, \alpha = 1.4, \beta = 2.2, c_{p} = 1.0, c_{cs} = 0.5$}
    \end{subfigure}
    ~ 
    \begin{subfigure}[t]{0.5\textwidth}
        \centering
        \includegraphics[height=2.0in]{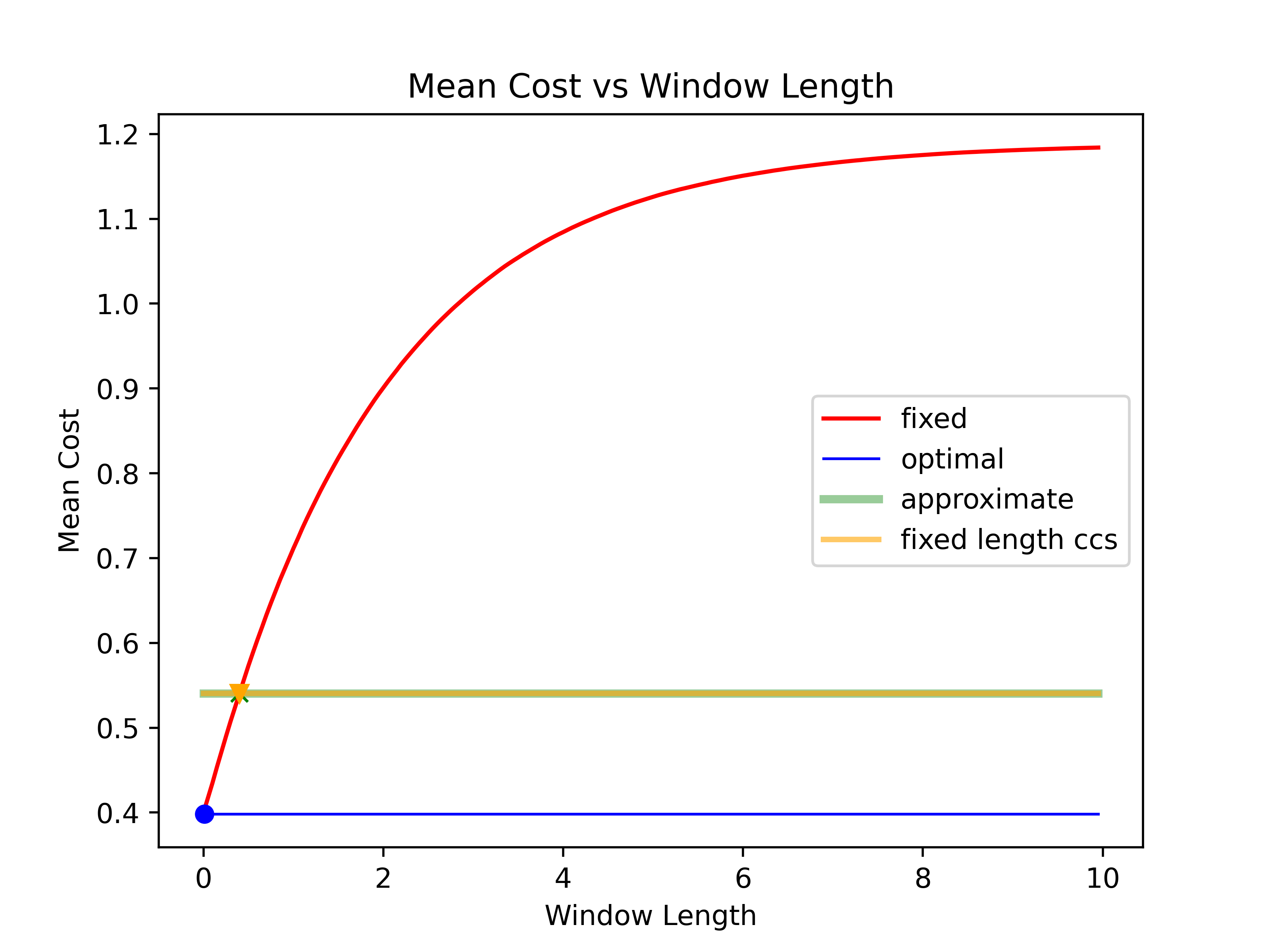}
        \caption{$\lambda_{0} = 0.5, \alpha = 0.6, \beta = 1.5, c_{p} = 1.0, c_{cs} = 0.4$}
    \end{subfigure}
    \caption{Plots of average cost comparisons between different policies for cases where $c_{p} \; \geq c_{cs}$}
    \label{cost_compare_sim2b}
\end{figure*}

As Figures \ref{cost_compare_sim1b}  and \ref{cost_compare_sim2b} illustrate, $\tau_{\text{approx}}$ is always more conservative than $\tau_{\text{fixed}}$ in that it chooses a weakly longer window length (which is is how it achieves its stronger worst case performance guarantee).  Despite this, their average performance is often quite similar.  In some situations, like Figure \ref{cost_compare_sim1b} (a), both policies are excessively conservative and so the extra conservatism of $\tau_{\text{approx}}$ causes it to perform worse.  This effect is bounded however, because in situations like Figure \ref{cost_compare_sim2b} (b) where $\tau_{\text{opt}, \mathcal{H} = \phi}$ is close to zero they become essentially the same policy.  In contrast, when the optimal window length is long, like Figure \ref{cost_compare_sim1b} (b), $\tau_{\text{approx}}$ performs better.  Again however the effect is small, this time because when optimal window lengths are relatively long it is typically unlikely that it will actually be a long time until the next arrival.

\begin{figure*}[h!]
    \centering
    \begin{subfigure}[t]{0.32\textwidth}
        \centering
        \includegraphics[height=1.5in]{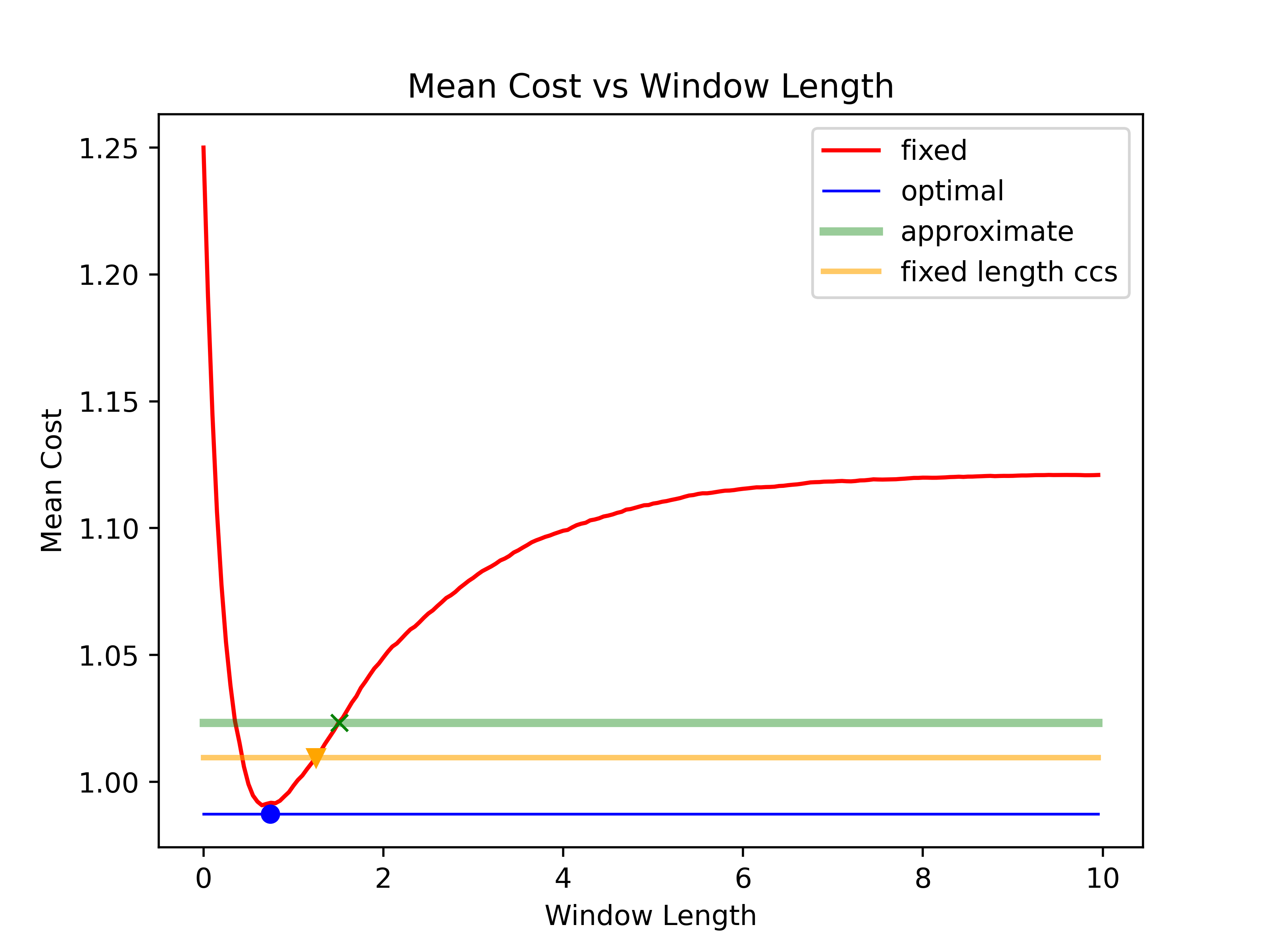}
        \caption{$\alpha = 0.8$}
    \end{subfigure}
    ~ 
    \begin{subfigure}[t]{0.32\textwidth}
        \centering
        \includegraphics[height=1.5in]{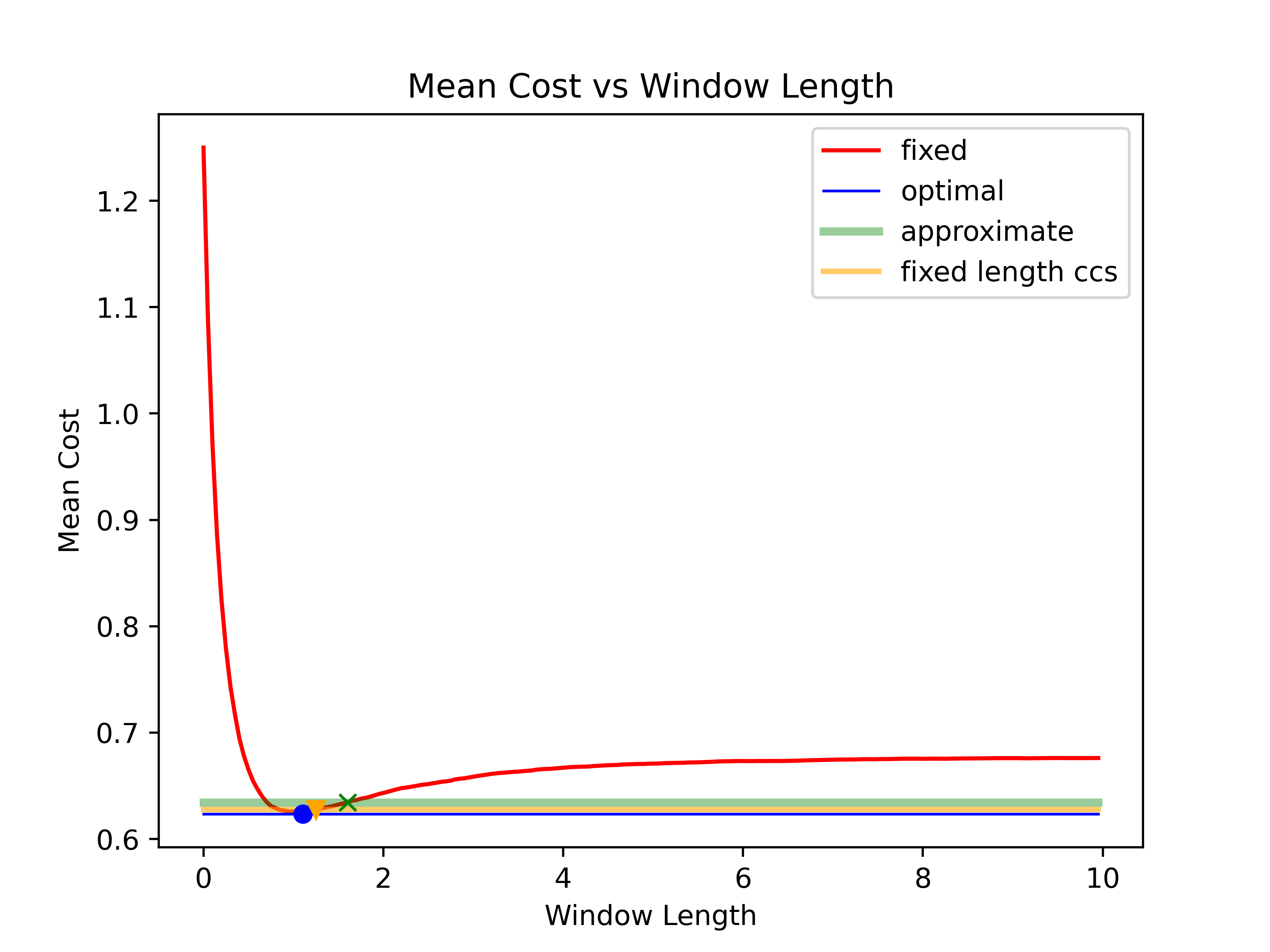}
        \caption{$\alpha = 1.4$}
    \end{subfigure}
    ~
    \begin{subfigure}[t]{0.32\textwidth}
        \centering
        \includegraphics[height=1.5in]{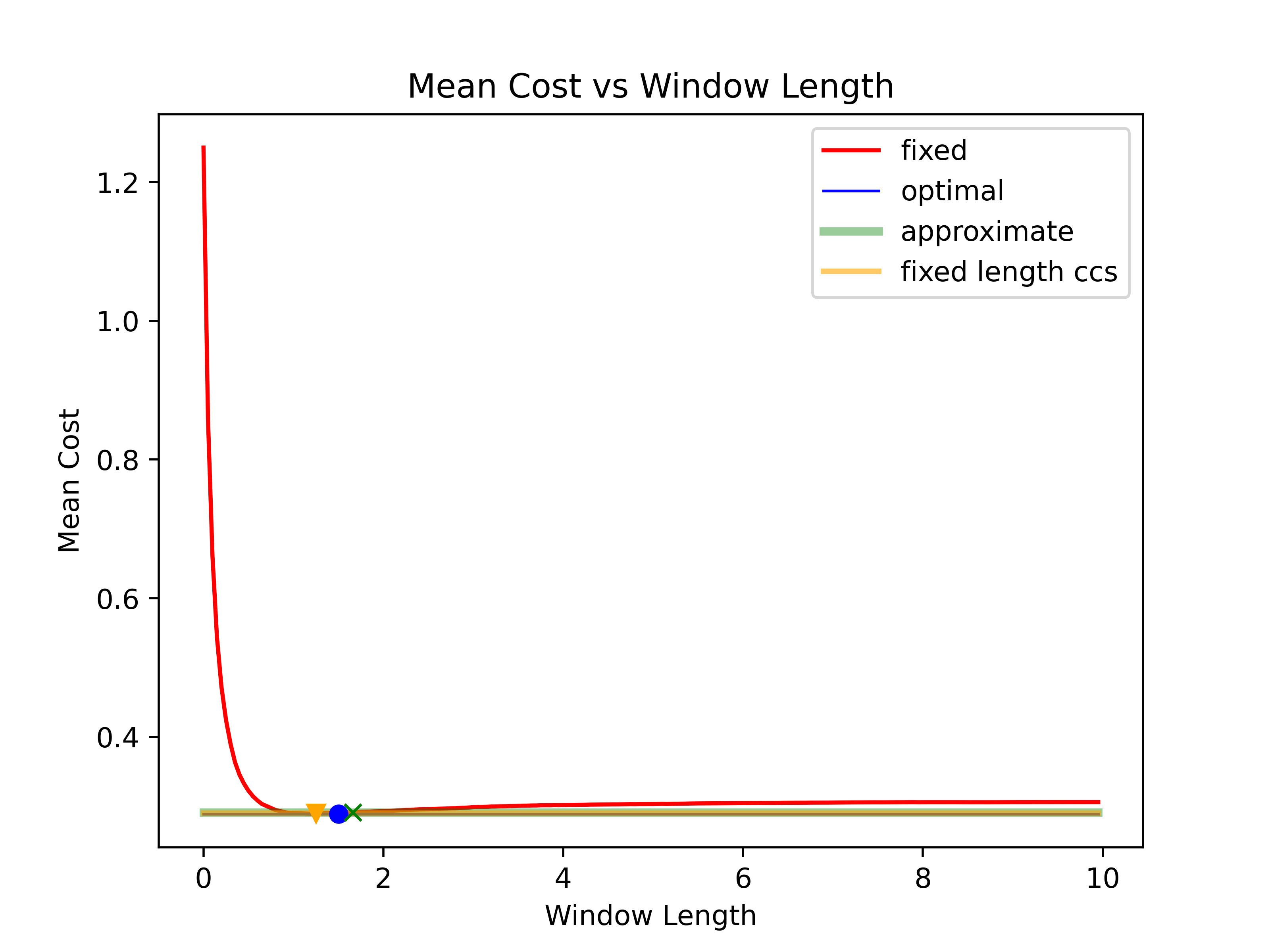}
        \caption{$\alpha = 2.0$}
    \end{subfigure}    
    \caption{Plots of average costs for policies when $\alpha$ is increased, given $\lambda_{0} = 0.6, \beta = 2.4, c_{p} = 1.0, c_{cs} = 1.25$}
    \label{cost_compare_alphab}
\end{figure*}

\begin{figure*}[h!]
    \centering
    \begin{subfigure}[t]{0.32\textwidth}
        \centering
        \includegraphics[height=1.5in]{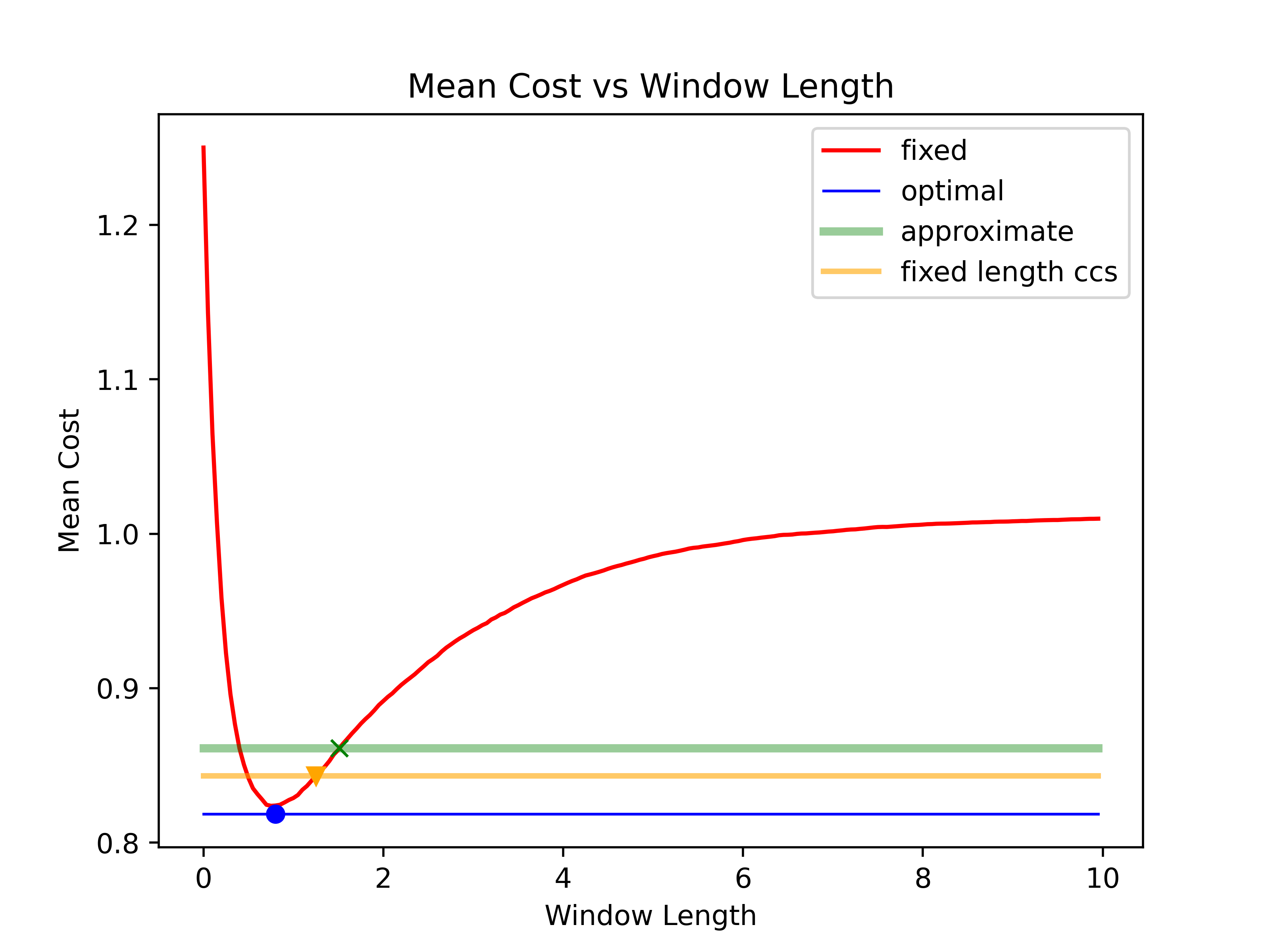}
        \caption{$\lambda_{0} = 0.5$}
    \end{subfigure}
    ~ 
    \begin{subfigure}[t]{0.32\textwidth}
        \centering
        \includegraphics[height=1.5in]{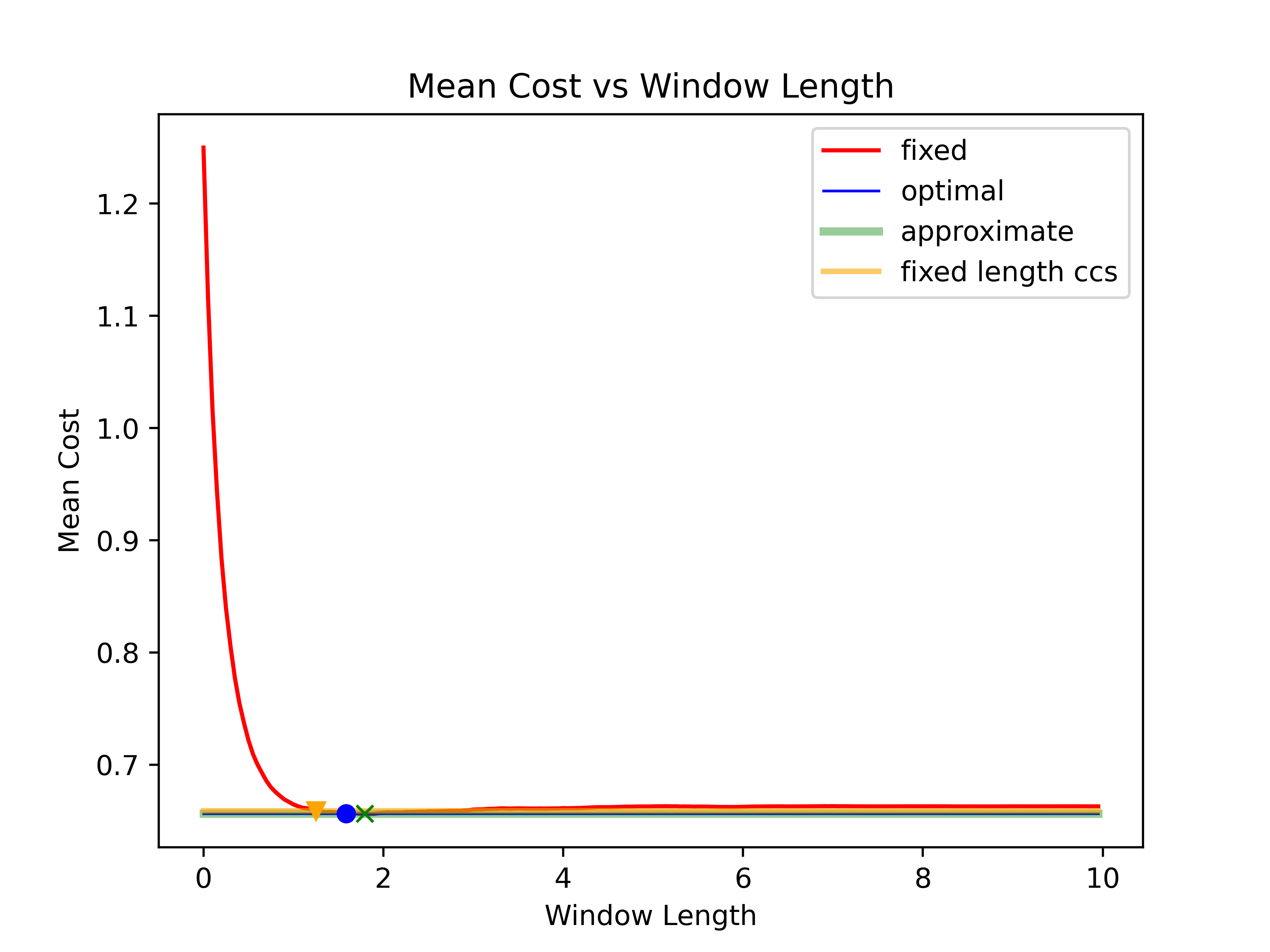}
        \caption{$\lambda_{0} = 0.75$}
    \end{subfigure}
    ~
    \begin{subfigure}[t]{0.32\textwidth}
        \centering
        \includegraphics[height=1.5in]{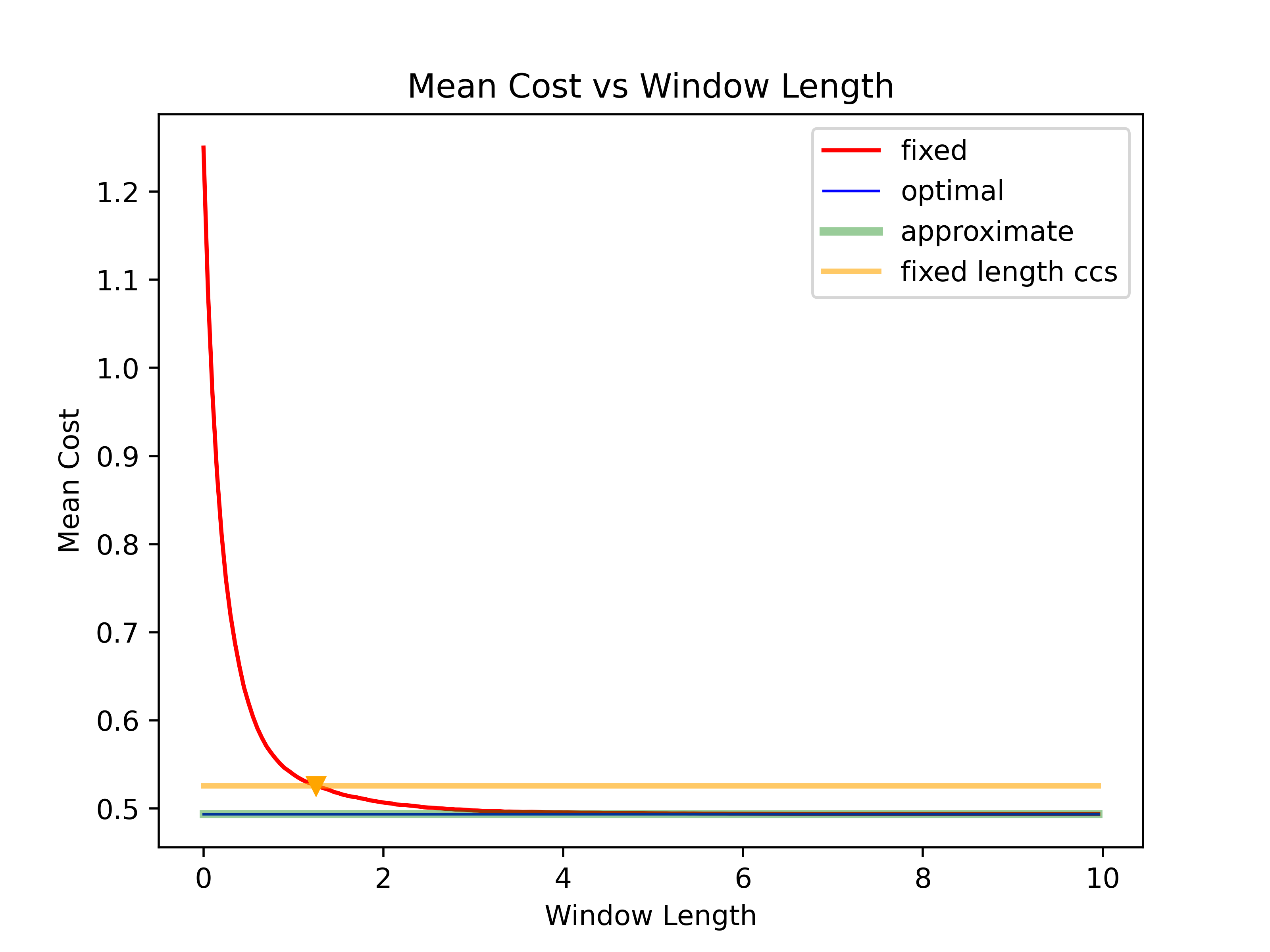}
        \caption{$\lambda_{0} = 1.0$}
    \end{subfigure}    
    \caption{Plots of average costs for policies when $\lambda_{0}$ is increased, given $\alpha = 1.2, \beta = 2.4, c_{p} = 1.0, c_{cs} = 1.25$}
    \label{cost_compare_lambdab}
\end{figure*}

\begin{figure*}[h!]
    \centering
    \begin{subfigure}[t]{0.32\textwidth}
        \centering
        \includegraphics[height=1.5in]{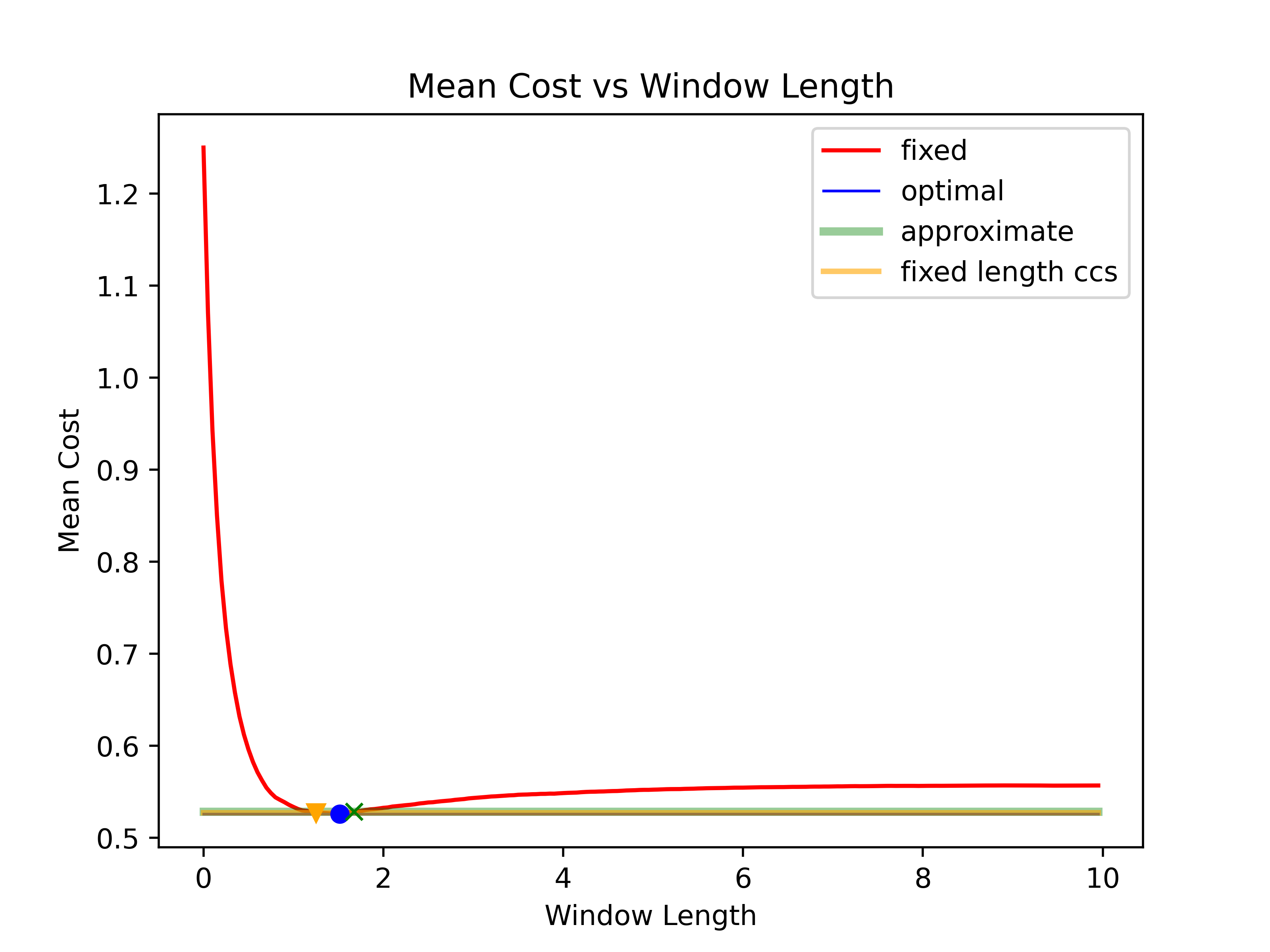}
        \caption{$\beta = 1.8$}
    \end{subfigure}
    ~ 
    \begin{subfigure}[t]{0.32\textwidth}
        \centering
        \includegraphics[height=1.5in]{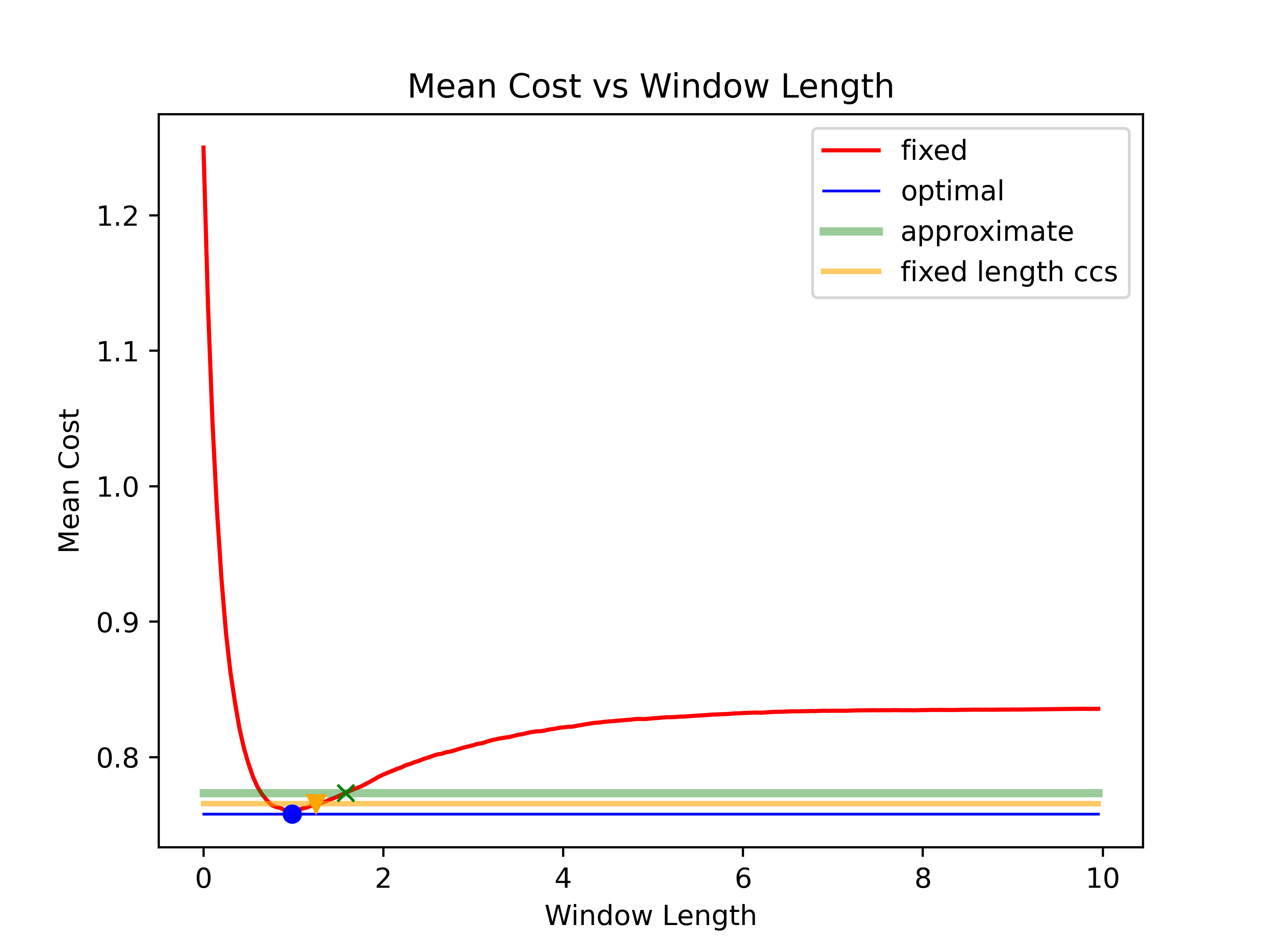}
        \caption{$\beta = 2.4$}
    \end{subfigure}
    ~
    \begin{subfigure}[t]{0.32\textwidth}
        \centering
        \includegraphics[height=1.5in]{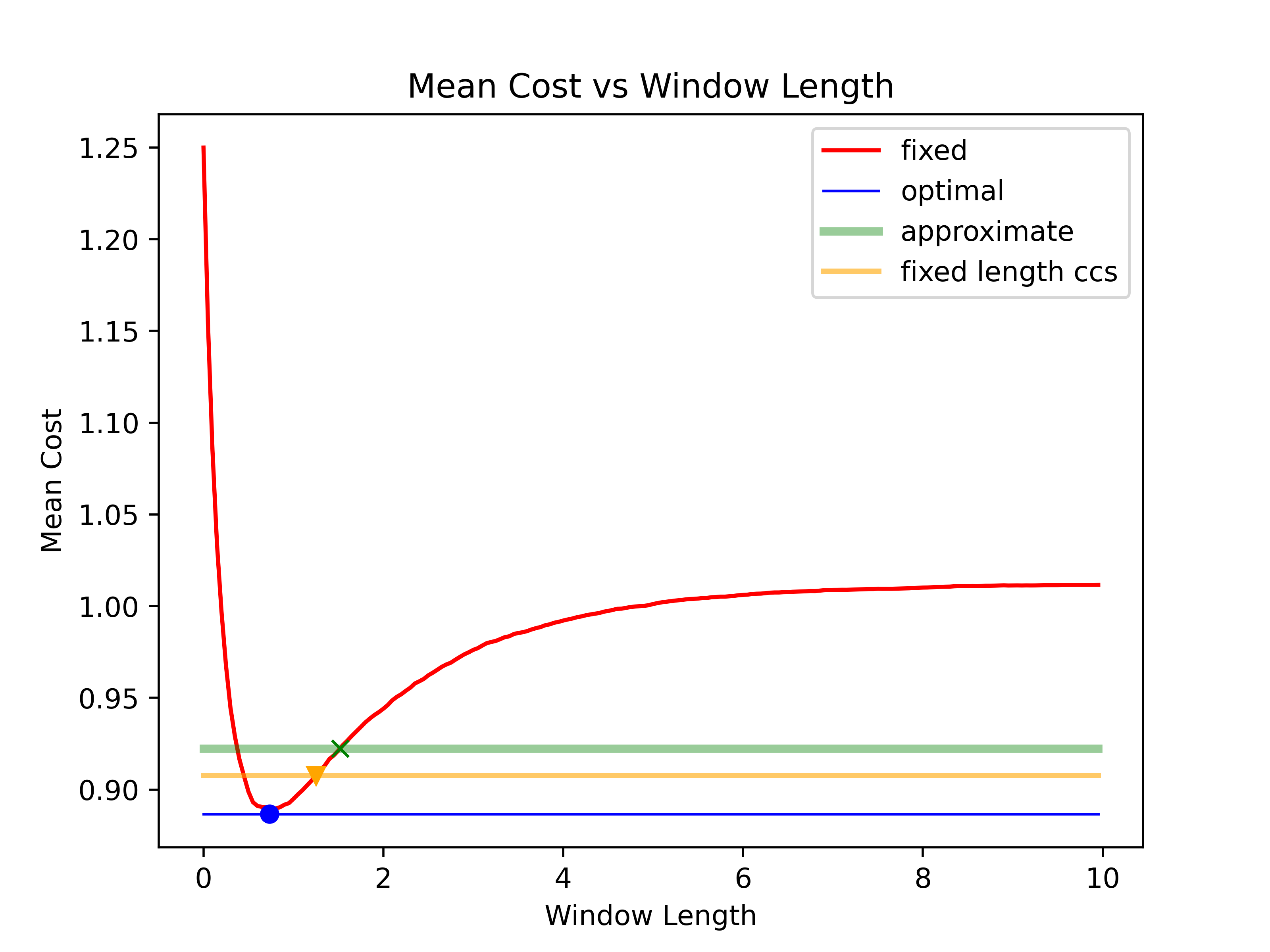}
        \caption{$\beta = 3.0$}
    \end{subfigure}    
    \caption{Plots of average costs for policies when $\beta$ is increased, given $\lambda_{0} = 0.6, \alpha = 1.2, c_{p} = 1.0, c_{cs} = 1.25$}
    \label{cost_compare_betab}
\end{figure*}

\subsection{Azure Datatrace Performance Results} \label{azure_results_appendix}

Figure \ref{azure_fig1} plots the trade-off curve between the average number of cold starts per application vs the normalized wasted memory for optimal, optimized-TTL, approximate and fixed policies. In Figure \ref{azure_fig1} (a), the trade-off curve is plotted when including only those applications that follow the Hawkes process during day 9. 
The trade-off Pareto curve of Figure \ref{azure_fig1} (b) plots the average number of cold starts per application vs the normalized memory for all applications invoked during day 9. The plots in Figure \ref{azure_fig1} show that the trade-off Pareto curve of the approximate policy is very slightly better than the fixed policy, but substantially worse than the optimal policy, and thus Optimized-TTL as well.

\begin{figure*}[h!]
    \centering
    \begin{subfigure}[t]{0.5\textwidth}
        \centering
        \includegraphics[height=2.3in]{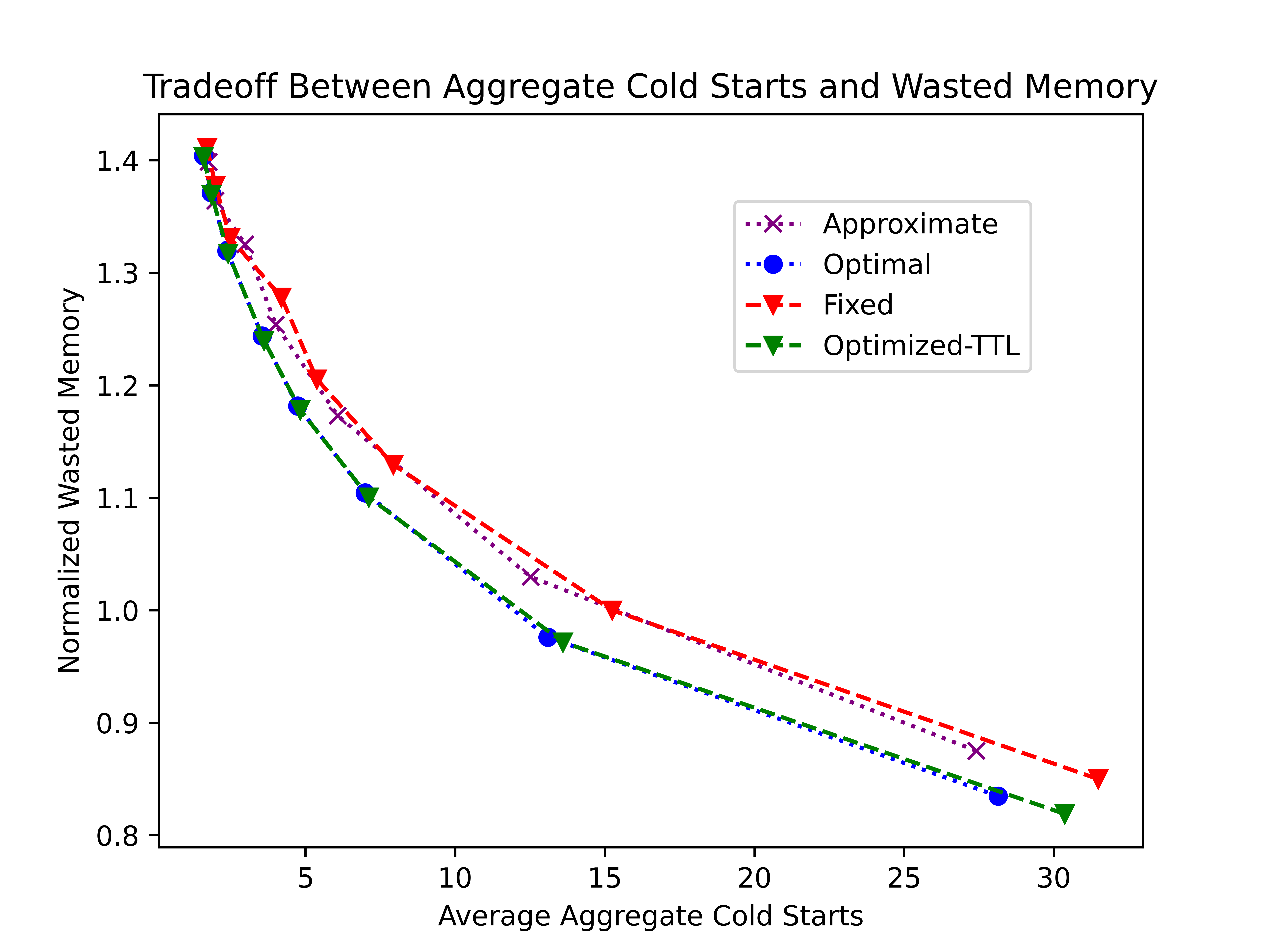}
        \caption{Evaluation only on Hawkes process applications}
    \end{subfigure}%
    ~ 
    \begin{subfigure}[t]{0.5\textwidth}
        \centering
        \includegraphics[height=2.3in]{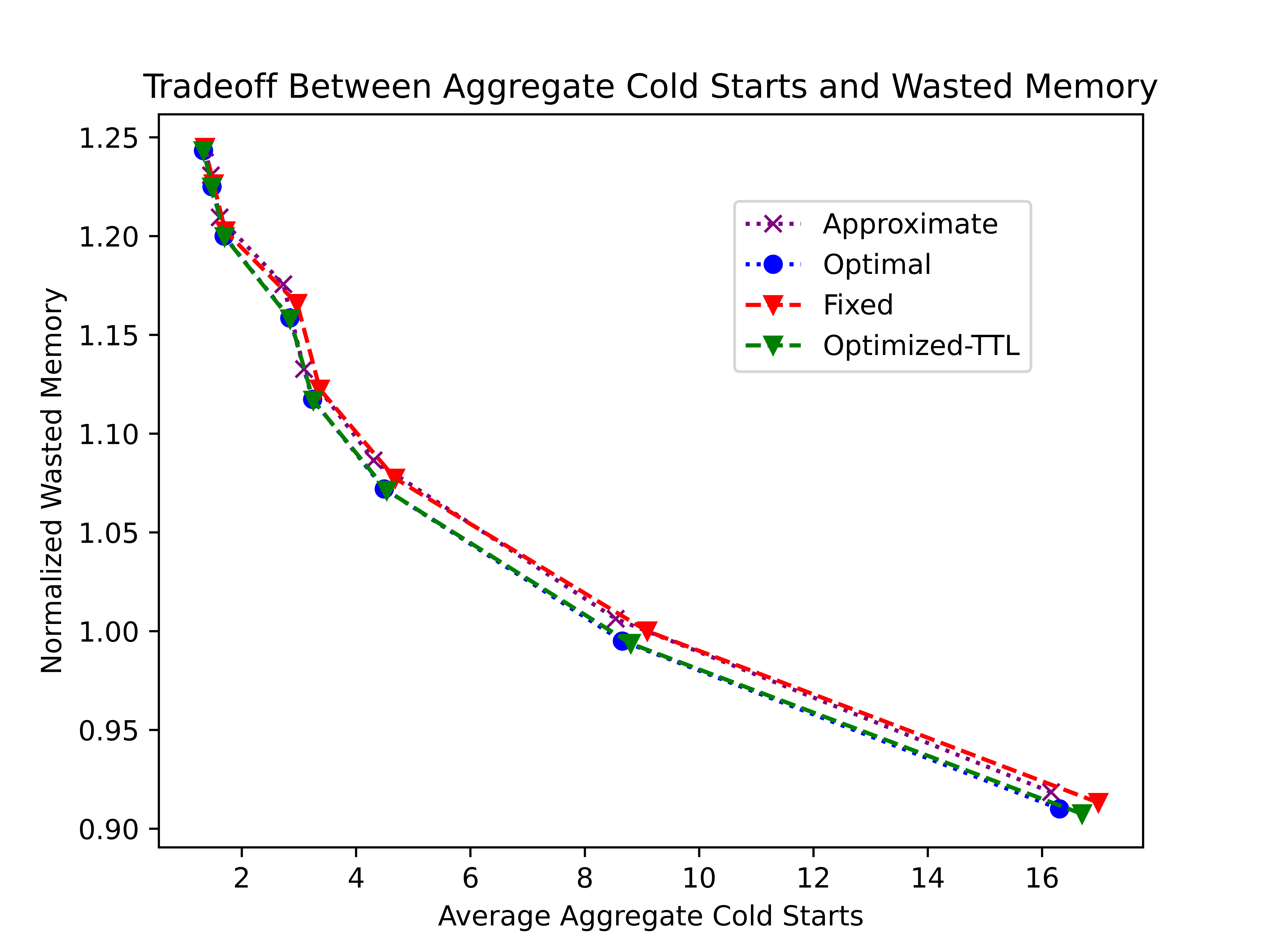}
        \caption{Evaluation on all applications}
    \end{subfigure}
    \caption{Trade-off curve of average number of cold starts vs normalized wasted memory}
    \label{azure_fig1}
\end{figure*}

\section{Use of separate data for goodness of fit}

The goodness of fit test is known to have a few limitations when the same data is used both to estimate the parameters and to compute the KS- statistic. \citet{reynaud2014goodness} show that the Hawkes process parameters when examined  for goodness of fit on the same dataset which was used for parameter estimation leads to a high bias. The authors propose sub-sampling as a reasonable solution to this problem. Rather than sub-sampling we took the advantage of additional data we are not currently using (e.g. day 7). \citet{van2016deep, kash2019combining} show a similar problem and solution for training and applying double Deep Q-learning Networks (DQNs). 

We report the results for the Optimized-TTL policy. We refer to the procedure where the goodness of fit is based on arrivals of application invocations on day 7 as "fix", whereas the procedure where the goodness of fit is based on arrivals of applications invocations on day 8 (same day as parameters estimated) is referred to as "no-fix". To compare the "fix" and "no-fix" procedures of selecting appropriate Hawkes process applications, we collect the common pool of applications invoked on day 7, day 8 and day 9. The Hawkes process applications in "fix" refer to applications where the parameters were estimated on day 8, and the KS test was performed on day 7 of the corresponding applications. The Hawkes process applications in "no-fix" refer to applications where the parameters were estimated on day 8, and the KS test was performed on the same day 8 of the corresponding applications (these are the common pool of applications present on day 7 and day 8). The number of common pool applications on day 7, day 8, and day 9 = 14788. Of these 3,694 applications fall into the 25 percentile apps that were selected as Hawkes process apps for each procedure ("fix", and "no-fix"). The amount of overlap on applications between the two tests, that is, the overlap of applications that passed the test on day 7 and applications that passed the test on day 8 = 2754. Overlap percentage = 2754/3694 = 0.745 . We show the plots of the trade-off curve between the fixed policy and the optimized-TTL policies for the overlapped apps in Figures \ref{azure_fix_test_figures}, and \ref{azure_no_fix_test_figures} for "fix", and "no-fix" procedures. Figures \ref{azure_fix_test_figures} (a), and \ref{azure_no_fix_test_figures} (a) show the trade-off curves for treated apps, whereas Figures \ref{azure_fix_test_figures} (b), and \ref{azure_no_fix_test_figures} (b) show the trade-off curves for all apps. We compute the cold start savings as the area between the optimized-TTL curve and the fixed policy curve divided by the maximum amount of wasted memory. Similarly, the wasted memory savings is the area between the optimized-TTL curve and the fixed policy curve divided by the maximum number of average cold-starts. The corresponding versions of the cold start savings and memory savings for optimized-TTL  policy are given in Table \ref{tab_appendix}. The “fix” version shows a slightly weaker performance on the treated apps, but a noticeably better performance on all apps than either other version.  So we view this as a demonstration that the "fix" does improve the selection of which apps to treat as Hawkes process and proceed to test the goodness of fit based on the arrivals of application invocations on day 7.


\begin{figure*}[h!]
    \centering
    \begin{subfigure}[t]{0.46\textwidth}
        \centering
        \includegraphics[height=2.3in]{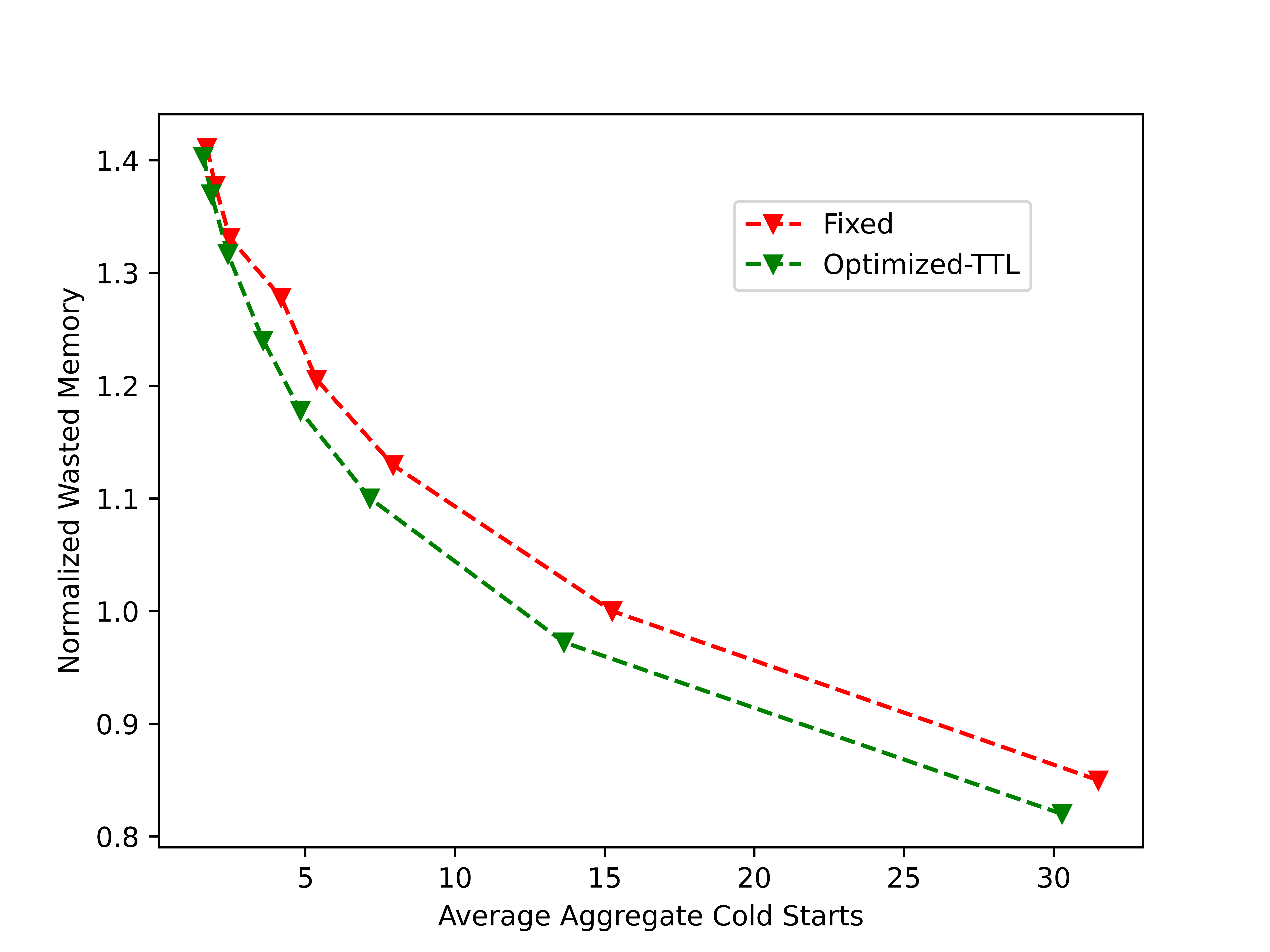}
        \caption{Hawkes process applications}
    \end{subfigure}%
    ~ 
    \begin{subfigure}[t]{0.46\textwidth}
        \centering
        \includegraphics[height=2.3in]{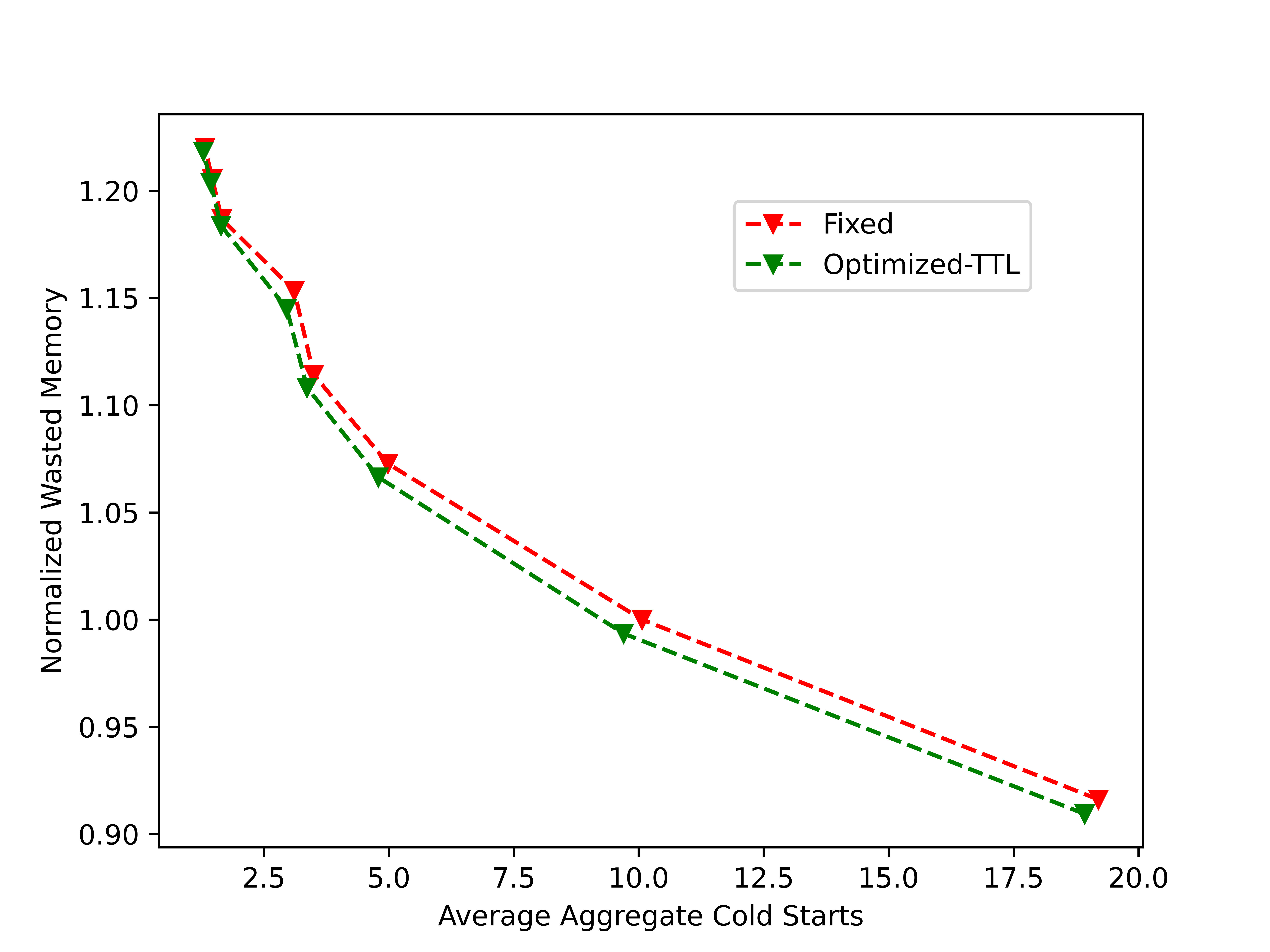}
        \caption{All applications}
    \end{subfigure}
\caption{Trade-off curve for optimized-TTL and fixed policies where goodness of fit is evaluated on day 7}
\label{azure_fix_test_figures}
\end{figure*}

\begin{figure*}[t!]
    \centering    
    \begin{subfigure}[t]{0.46\textwidth}
        \centering
        \includegraphics[height=2.3in]{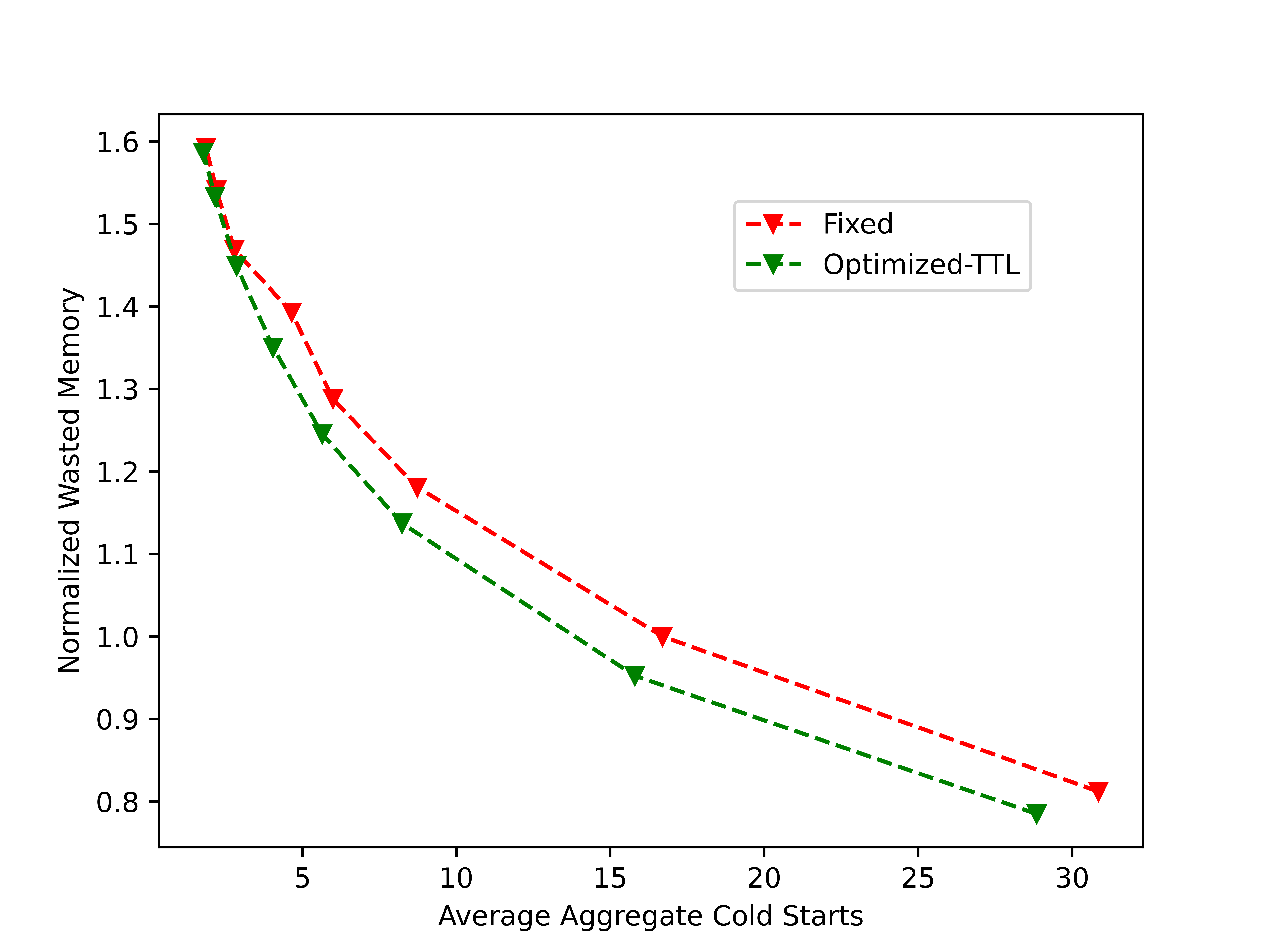}
        \caption{Hawkes process applications}
    \end{subfigure}%
    ~ 
    \begin{subfigure}[t]{0.46\textwidth}
        \centering
        \includegraphics[height=2.3in]{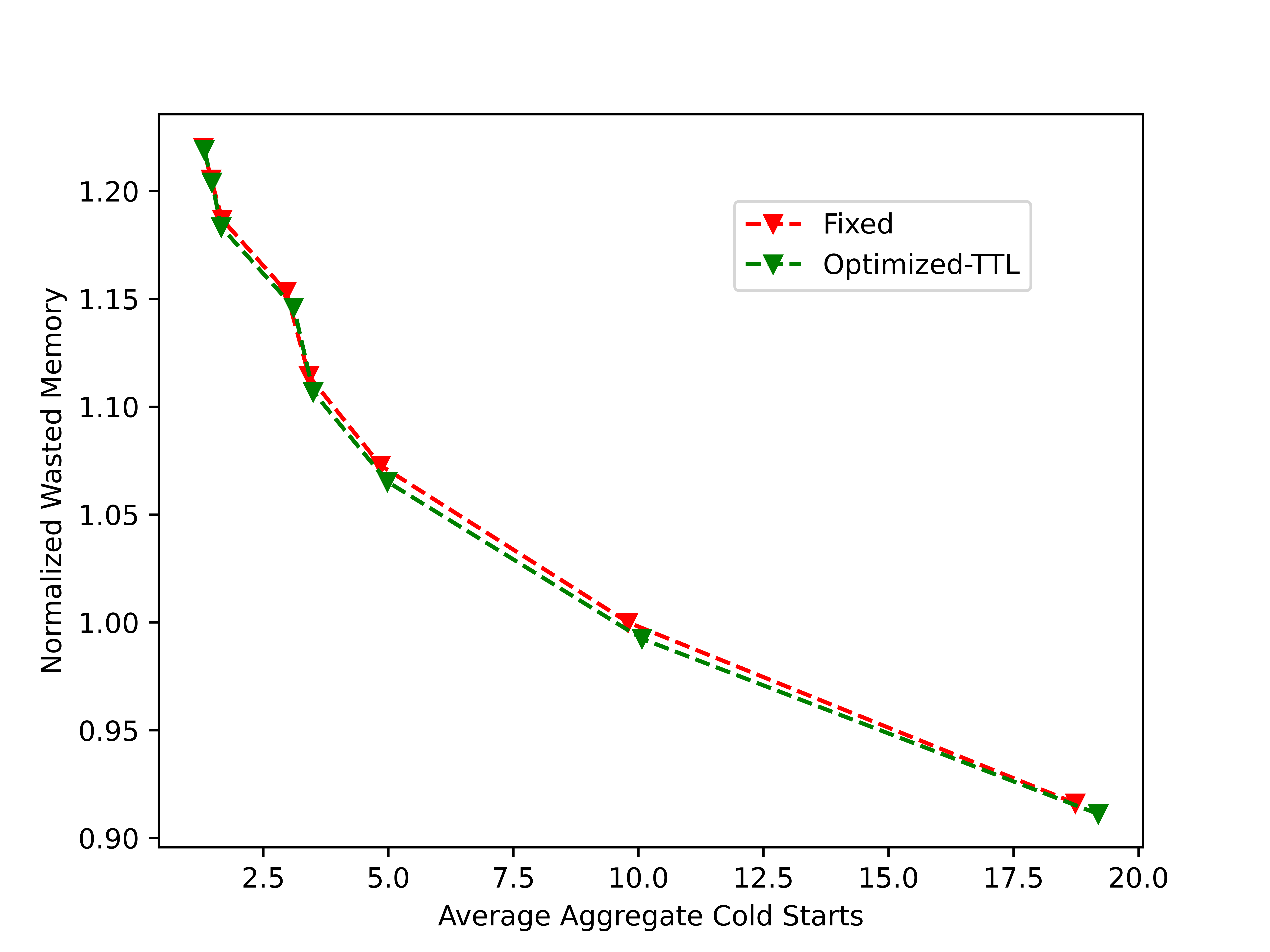}
        \caption{All applications}
    \end{subfigure}
   \caption{Trade-off curve for optimized-TTL and fixed policies where goodness of fit is evaluated on day 8}
    \label{azure_no_fix_test_figures}
\end{figure*}

\begin{table*}[t]
\centering
\begin{tabular}{|c|c|c|c|c|}
\hline
Procedure & Avg. Cold Start Savings (Hawkes) & (All) & Avg. Memory Savings (Hawkes) & (All)\\
\hline
Optimized-TTL (fix) & 0.834 &  0.1393 &  0.043  & 0.0085 \\
\hline
Optimized-TTL (no-fix) & 1.037  & 0.0574 & 0.053 & 0.0035 \\
\hline
\end{tabular}
\caption{\label{tab_appendix} Average performance improvement over fixed policy}
\end{table*}

\end{document}